\documentclass[11pt]{article}
\usepackage[a4paper,
            left=1.25in,
            right=1.25in,
            top=1.25in,
            bottom=1.25in,
            footskip=.5in]{geometry}

\usepackage{authblk}
\usepackage{graphicx}
\usepackage{dsfont}
\usepackage{enumerate}
\usepackage{mathtools}
\usepackage[colorlinks=true,linkcolor=blue,citecolor=blue,urlcolor=blue,plainpages=false,pdfpagelabels]{hyperref}
\usepackage{amsmath}
\usepackage{amssymb}
\usepackage{physics}
\usepackage{bbm}
\usepackage{nicefrac}
\usepackage{url}
\usepackage{complexity}
\usepackage{xcolor}
\usepackage[normalem]{ulem}
\usepackage{amsthm}
\usepackage{appendix}
\usepackage[capitalize,noabbrev]{cleveref}

\usepackage[backend=biber, style=alphabetic, backref=true, hyperref=true, maxbibnames=99]{biblatex}
\usepackage[babel,english=british]{csquotes}
\DefineBibliographyStrings{english}{%
    backrefpage  = {cited on p.}, 
    backrefpages = {cited on pp.} 
}
\newcommand{\stkout}[1]{\ifmmode\text{\sout{\ensuremath{#1}}}\else\sout{#1}\fi}
\newif\ifverbose
\verbosetrue

\newtheorem{theorem}{Theorem}[section]

\newtheorem{proposition}[theorem]{Proposition}

\newtheorem{corollary}[theorem]{Corollary}

\newtheorem{lemma}[theorem]{Lemma}
\newtheorem{definition}[theorem]{Definition}

\theoremstyle{definition}
\newtheorem{example}[theorem]{Example}
\newtheorem{remark}[theorem]{Remark}

\numberwithin{equation}{section}

\begin{document}

\title{Learning Quantum Processes and Hamiltonians via the Pauli Transfer Matrix}

\author[1,2]{Matthias C.~Caro\thanks{\href{mailto:mcaro@caltech.edu}{mcaro@caltech.edu}}}
\affil[1]{Institute for Quantum Information and Matter, Caltech, Pasadena, CA, USA}
\affil[2]{Dahlem Center for Complex Quantum Systems, Freie Universit\"at Berlin, Berlin, Germany}
  
\date{}
\setcounter{Maxaffil}{0}
\renewcommand\Affilfont{\itshape\small}

\maketitle

\begin{abstract}
    Learning about physical systems from quantum-enhanced experiments, relying on a quantum memory and quantum processing, can outperform learning from experiments in which only classical memory and processing are available. Whereas quantum advantages have been established for a variety of state learning tasks, quantum process learning allows for comparable advantages only with a careful problem formulation and is less understood. We establish an exponential quantum advantage for learning an unknown $n$-qubit quantum process $\mathcal{N}$.
    
    We show that a quantum memory allows to efficiently solve the following tasks: (a) learning the Pauli transfer matrix, i.e., the matrix representation w.r.t.~the normalized Pauli basis, of an arbitrary $\mathcal{N}$, (b) predicting expectation values of bounded Pauli-sparse observables measured on the output of an arbitrary $\mathcal{N}$ upon input of a Pauli-sparse state, and (c) predicting expectation values of arbitrary bounded observables measured on the output of an unknown $\mathcal{N}$ with sparse Pauli transfer matrix upon input of an arbitrary state. With quantum memory, these tasks can be solved using linearly-in-$n$ many copies of the Choi state of $\mathcal{N}$, and even time-efficiently in the case of (b). In contrast, any learner without quantum memory requires exponentially-in-$n$ many queries, even when querying $\mathcal{N}$ on subsystems of adaptively chosen states and performing adaptively chosen measurements. This separation extends existing shadow tomography upper and lower bounds from states to channels via the Choi-Jamiolkowski isomorphism. Moreover, we combine Pauli transfer matrix learning with polynomial interpolation techniques to develop a procedure for learning arbitrary Hamiltonians, which may have non-local all-to-all interactions, from short-time dynamics. Our results highlight the power of quantum-enhanced experiments for learning highly complex quantum dynamics.
\end{abstract}

\newpage
\section{Introduction}\label{section:introduction}

A complete characterization of an unknown quantum system or quantum process requires a number of copies of the unknown object that scales (at least) linearly in the dimension \cite{haah2016sample, odonnell2016efficient, odonnell2017efficient}, and thus exponentially in the number of qubits for a multi-qubit system. Hence, there is a need for more resource-efficient alternatives that still suffice to make meaningful (while necessarily incomplete) predictions about the unknown quantum object. 
In the case of quantum states, the frameworks of classical shadows and shadow tomography, motivated by learning-theoretic perspectives on the task of predicting properties of an unknown quantum state, have led to results on sample-efficiently predicting a plethora of physically relevant properties of a quantum system \cite{aaronson02007learnability, aaronson2018shadow, paini2019approximate, huang2020predicting, badescu2021improved, huang2021information-theoretic, hu2021classical, rouze2021learning, huang2022provably, akhtar2022scalable, bertoni2022shallow, becker2022classical, arienzo2022closedform}.
Classical shadows have the desirable feature that they do not rely on complicated quantum processing, interesting predictions can already be made with non-adaptive single-copy measurements, which are amenable to implementation on noisy intermediate-scale (NISQ) \cite{preskill2018quantum} devices.
However, there are limitations to what can be learned efficiently with this kind of access to an unknown quantum state or process. As shown in \cite{chen2022exponential}, state shadow tomography can in general only be solved sample-efficiently in quantum-enhanced experiments, with access to a quantum memory and to quantum processing. \cite{chen2022exponential} identified several other quantum learning problems with similarly strong query complexity separations, leading to a quantum advantage that becomes realizable on current quantum devices \cite{huang2022quantum}.

Whereas state shadow tomography succeeds without any assumptions on the unknown state to be learned and on the observables of interest, we cannot hope to query-efficiently and accurately predict, without prior assumptions, expectation values of arbitrary measurements performed on the output of an arbitrary quantum channel upon input of an arbitrary state. This can be seen by embedding the problem of learning an arbitrary Boolean function 
into the task of learning an arbitrary quantum channel (as discussed in \cite{huang2022simons-colloquium, chen2022epfl-seminar}).
Nevertheless, results about learning quantum processes are beginning to emerge \cite{chung2021sample, caro2021binary, flammia2021pauli, levy2021classical-shadows, caro2022generalization, chen2022quantum, huang2022foundations, huang2022learning, fanizza2022learning, kunjummen2023shadow-process}, with different restrictions on the channel to be learned or on the input states and output observables for which expectation values are to be predicted.

Despite the insights of these prior works, the landscape of quantum channel learning is still largely unexplored, especially for general channels with potentially exponential complexity. In particular, while we now know an exponential quantum advantage of quantum-enhanced over conventional experiments in shadow tomography of static quantum systems, analogous results for the case of quantum evolutions are lacking. A central challenge in understanding the potential for a quantum advantage here is that conventional experiments for learning an unknown quantum channel can involve intricate adaptive strategies for choosing both the input states and the measurements to be performed. Thus, tightly bounding the capabilities of such strategies for channel learning has so far been out of reach.
This naturally leads us to study the importance of quantum-enhanced experiments for tasks of channel learning, with a focus on their potential to outperform their adaptive classical counterparts. 
Namely, we investigate the following two questions:
\begin{center}
    \textit{\textbf{(1)}} \textit{What can we learn about a quantum channel without prior assumptions?}\\
    \textit{\textbf{(2)}} \textit{Are quantum-enhanced experiments necessary for quantum channel learning?}
\end{center}

\subsection{Overview of the Results} 

We identify learning all entries of the Pauli transfer matrix of an unknown quantum channel as a task that can be solved efficiently by a learner with quantum memory, but that requires exponentially many queries to the unknown channel if no quantum memory is used. 
Here, we view a quantum channel as a linear map between matrices with complex entries, and the Pauli transfer matrix is its matrix representation with respect to the normalized Pauli basis, an orthonormal basis for the space of matrices of importance in quantum computing. 
Moreover, we show that this quantum advantage carries over to two restricted variants of a channel shadow tomography task. Finally, we demonstrate that Pauli transfer matrix learning can be used as a subroutine to learn an arbitrary quantum Hamiltonian from short-time dynamics.
Thus, we have found one answer to Question (1): Pauli transfer matrix and Hamiltonian learning are meaningful channel learning problems that can be solved efficiently without any prior assumptions on the object to be learned. 
Additionally, query-efficient Pauli transfer matrix learning requires quantum-enhanced experiments, thus emphasizing their relevance to channel learning in answer to Question (2).

Our first main result is an exponential query complexity separation between learners with quantum memory and learners without quantum memory for the task of learning the Pauli transfer matrix of an unknown quantum channel:

\begin{theorem}[Pauli transfer matrix learning with and without quantum memory]\label{theorem:main-result-ptm-learning-separation}
    There is a learning algorithm with quantum memory that uses $\mathcal{O}(\nicefrac{n}{\varepsilon^4})$ copies of the Choi state of an unknown $n$-qubit quantum channel $\mathcal{N}$ to output simultaneously $\varepsilon$-accurate estimates for all the $16^n$ entries $\tfrac{1}{2^n}\tr[\sigma_A \mathcal{N}(\sigma_B)]$, $A,B\in\{0,1,2,3\}^n$, of its Pauli transfer matrix, with high success probability.
    In contrast, any learning algorithm without quantum memory, even when allowed to access the unknown channel on subsystems of adaptively chosen input states and perform adaptively chosen output measurements, has to query the unknown channel $\Omega (\nicefrac{4^n}{\varepsilon^2})$ times to achieve the same guarantee.
\end{theorem}

Here, the lower bound for learners without quantum memory holds even if the unknown channel is promised in advance to be doubly-stochastic, entanglement-breaking, and to have a sparse Pauli transfer matrix.
We demonstrate that this query complexity separation carries over to the tasks (a) of predicting arbitrary expectation values of the form $\tr[O \mathcal{N}(\rho)]$, with $O$ a bounded Pauli-sparse observable, $\rho$ a Pauli-sparse quantum state, and $\mathcal{N}$ an arbitrary quantum channel, and (b) of predicting arbitrary expectation values of the form $\tr[O \mathcal{N}(\rho)]$, with $O$ an arbitrary bounded observable, $\rho$ an arbitrary quantum state, and $\mathcal{N}$ a quantum channel with sparse Pauli transfer matrix. 
In fact, for task (a) we give a learning algorithm with quantum memory that query-efficiently builds a classical memory-efficient representation $\hat{\mathcal{N}}$ of $\mathcal{N}$ from which any $M$ such expectation values can be estimated query- and time-efficiently:

\begin{corollary}[Predicting Pauli-sparse expectation values for arbitrary channels]\label{corollary:main-result-pauli-sparse-expectation-prediction}
    There is a learning algorithm with quantum memory that uses $\mathcal{O}(\nicefrac{n}{\varepsilon^4})$ copies of the Choi state of an unknown $n$-qubit quantum channel $\mathcal{N}$ to build a classical representation $\hat{\mathcal{N}}$ of $\mathcal{N}$ from which any $M$ expectation values of the form $\tr[O \mathcal{N}(\rho)]$, with $O$ a bounded Pauli-sparse observable and $\rho$ a Pauli-sparse quantum state, can be predicted to accuracy $\varepsilon$, with high success probability.
    The classical representation consists of $\mathcal{O}(\nicefrac{n^2}{\varepsilon^4})$ real numbers stored in classical memory, and the expectation values are predicted using $\mathcal{O}(\nicefrac{\log(M)}{\varepsilon^2})$ additional Choi state copies and using classical computation time $\mathcal{O}(\nicefrac{M n^2}{\varepsilon^4})$.
\end{corollary}

Notably, these efficiency guarantees do not require any assumptions on the unknown quantum channel $\mathcal{N}$, so this channel may be a quantum process of arbitrary complexity.
In \Cref{corollary:main-result-pauli-sparse-expectation-prediction}, we implicitly assumed the observables and states to have $\mathcal{O}(1)$-sparse Pauli expansions. However, as our guarantees scale polynomially with the sparsity parameters, we still get an efficient prediction protocol if we allow for polynomially-in-$n$ Pauli-sparse observables and states. This can then be viewed as a restricted generalization of shadow tomography from states to channels.
While we formulate \Cref{theorem:main-result-ptm-learning-separation} and \Cref{corollary:main-result-pauli-sparse-expectation-prediction} for the Pauli transfer matrix, we also give similar protocols with polynomial-in-$n$ copy complexity for transfer matrices with respect to general (appropriately normalized) unitary orthonormal bases and for predicting expectation values with states and observables having a sparse expansion in such a basis, again for an arbitrary unknown channel.

As an application of our protocol for learning the Pauli transfer matrix with quantum memory, in our second main result we show how this can be combined with polynomial interpolation for derivative estimation to efficiently learn the Pauli coefficients of an unknown Hamiltonian $H$. Here, we assume query access to the unitary that implements time evolution along $H$ for (different) short times, all on the order of $\mathcal{O}(\nicefrac{1}{\norm{H}})$:

\begin{theorem}[Learning arbitrary Hamiltonians with quantum memory]\label{theorem:main-result-hamiltonian-learning}
    There is a learning algorithm with quantum memory that uses $\tilde{\mathcal{O}}\left( \nicefrac{n \norm{H}^4}{\varepsilon^4} \right)$ parallel queries to $\mathcal{O}(\nicefrac{1}{\norm{H}})$-time evolutions under an unknown $n$-qubit Hamiltonian $H$ with Pauli expansion $H = \sum_{A\in\{0,1,2,3\}^n} \alpha(A) \sigma_A$ with a total evolution time of $\tilde{\mathcal{O}}\left( \nicefrac{n \norm{H}^3}{\varepsilon^4} \right)$ to estimate all the $4^n -1$ non-trivial Pauli coefficients $\alpha(A)$, $A\in\{0,1,2,3\}^n\setminus\{0^n\}$, of $H$ to accuracy $\varepsilon$, with high success probability.
\end{theorem}

If the goal is to estimate only $M$ Pauli coefficients of $H$, then the estimation procedure uses classical computation time scaling linearly in $M$ and polynomially in $n$ and $\nicefrac{\norm{H}}{\varepsilon}$. Moreover, the required classical memory also scales polynomially in $n$ and $\nicefrac{\norm{H}}{\varepsilon}$.
Similarly to how \Cref{theorem:main-result-ptm-learning-separation} does not require any assumptions on the unknown quantum channel, \Cref{theorem:main-result-hamiltonian-learning} works query-efficiently for arbitrary Hamiltonians, assuming only a polynomial operator norm bound to be guaranteed in advance. In particular, the unknown Hamiltonian can contain non-local all-to-all interaction terms.

\subsection{Techniques and Proof Overview}

\paragraph{Upper bounds for learning with quantum memory}
To achieve the complexity upper bounds in \Cref{theorem:main-result-ptm-learning-separation} and \Cref{corollary:main-result-pauli-sparse-expectation-prediction}, we rephrase the task of Pauli transfer matrix learning as a shadow tomography task for the Choi state of the unknown quantum channel. Here, the (normalized) Choi state $\tfrac{1}{2^n}\Gamma^{\mathcal{N}}$ of an $n$-qubit channel $\mathcal{N}$ is obtained by applying $\mathcal{N}$ to one half of a maximally entangled state: $\tfrac{1}{2^n}\Gamma^{\mathcal{N}} = (\operatorname{id}\otimes\mathcal{N})(\ketbra{\Omega}{\Omega})$, with $\ket{\Omega}$ maximally entangled. 
Concretely, we observe that Pauli transfer matrix entries can be rewritten in terms of expectation values of Pauli observables on the Choi state as $\frac{1}{2^n}\tr [\sigma_A \mathcal{N}(\sigma_B)]= \tr[(\sigma_B^{\top} \otimes \sigma_A) \frac{1}{2^n}\Gamma^\mathcal{N}]$. 
This now allows us to appeal to an existing procedure for state shadow tomography with Pauli observables from \cite{huang2021information-theoretic}, which is based on $2$-copy Bell measurements. 

We point out that channel learning tasks cannot generically be reduced to state learning tasks in a query-efficient manner. One important reason for this is the dimension factor of $2^n$ incurred when translating from Choi state expectation values to input-output expectation values of the associated channel. 
However, with our chosen learning task, the normalization of the Pauli transfer matrix entries exactly cancels the dimension factor. This allows us to rely on Choi expectation value estimates without any loss in accuracy, so that we inherit the good complexity scaling of state tomography. 

We can go from estimates for Pauli transfer matrix entries to expectation value estimates under Pauli-sparsity assumptions on either the state and observable or on the channel by expanding the respective objects in the Pauli orthonormal basis. 
Here, the sparsity assumption guarantees that the transfer matrix learning guarantee, which can be viewed as achieving good $\infty$-norm error, also leads to a good error for other $p$-norms.
To obtain the analogous results (albeit without bounds on the classical memory and computation time) for more general orthonormal bases and associated transfer matrices as well as for suitably sparse objects, we follow the same line of reasoning, but rely on the general state shadow tomography results of \cite{aaronson2018shadow, badescu2021improved}. Consequently, any improvements to state shadow tomography will also directly improve our results.

\paragraph{Lower bounds for learning without quantum memory}
For our exponential query complexity lower bounds, we extend a technique introduced in \cite{chen2022exponential} to prove query complexity lower bounds for state shadow tomography without quantum memory to a channel learning task of shadow tomography type. 
Our proof consists of three steps. First, following the framework proposed in \cite{chen2022exponential}, we model an adaptive learning procedure that uses only a single copy of an unknown quantum channel at a time in terms of a so-called learning tree. Here, branchings in the tree correspond to different adaptive choices of the learner, depending on observed measurement outcomes.
Second, we consider the distributions over the leaves of the learning tree that arise when the learner acts on different (from their perspective unknown) channels, and we upper bound how distinguishable those distributions are. Here, distinguishability is to be understood in terms of uniform one-sided likelihood ratio lower bounds for the leaf probability distributions. 
Third, via Le Cam's two-point method, such likelihood bounds gives rise to query complexity lower bounds for a corresponding distinguishing task. We then argue that any learner would in particular be able to solve that distinguishing task, thus inheriting the lower bound.

Even with the general learning tree framework in place, analyzing the distinguishability of the induced distributions over learning trees remains a formidable technical challenge, which has so far not been resolved for general quantum channels. 
For state shadow tomography without quantum memory, \cite{chen2022exponential} was able to overcome this technical hurdle and to establish exponential query complexity lower bounds.
Our first observation in trying to instantiate the learning tree framework for lower bounds in Pauli transfer matrix learning is that a small modification to the reasoning of \cite{chen2022exponential} already suffices to obtain analogous lower bounds for predicting Pauli transfer matrix entries of an unknown quantum channel if the learner only has access to copies of its Choi state. Concretely, we obtain the lower bound via the intermediate task of distinguishing random Choi states of the form $\tfrac{1}{2^n}(\mathbbm{1}\pm 3\varepsilon P)$, where $\pm P$ is a signed $(2n)$-qubit Pauli, from the maximally mixed state.
This modification involves determining the second moment for a certain sub-ensemble of all $n$-qubit Paulis.
However, the $\Omega \left(\nicefrac{4^n}{\varepsilon^2}\right)$ query complexity lower bound obtained in this way does not yet apply to general adaptive quantum channel learning procedures without quantum memory, which can also adaptively choose inputs to the unknown channel.

Our main technical contribution is to prove the one-sided likelihood bounds required for LeCam's two-point method even in this fully adaptive setting.
To achieve this, we identify a quantity that, given a set of Pauli transfer matrix entries, characterizes the hardness of simultaneously predicting those entries without using a quantum memory. 
We analyze this quantity for a suitable set of transfer matrix entries and for a suitable ensemble of unknown quantum channels, whose Choi states are either $\tfrac{1}{2^n}(\mathbbm{1}\pm 3\varepsilon P)$, with a random signed Pauli string $\pm P$, or maximally mixed.
Our analysis, which combines a second moment calculation for Pauli sub-ensembles with matrix-analytic tools, provides a new way of utilizing the framework of \cite{chen2022exponential} and establishes the $\Omega \left(\nicefrac{4^n}{\varepsilon^2}\right)$ query complexity lower bound for general channel learning procedures without quantum memory. This analysis may serve as a template for proving query complexity lower bounds for further channel learning tasks without access to quantum-enhanced experiments.

\paragraph{Hamiltonian learning from Pauli transfer matrix learning}
We prove that our Pauli transfer matrix learning protocol gives rise to the Hamiltonian learning algorithm of \Cref{theorem:main-result-hamiltonian-learning} in three steps. 
First, reinterpreting observations made in \cite{haah2021optimal, stilck-franca2022efficient, gu2022practical}, we note that any single Pauli expansion coefficient $\alpha(A)$ of the unknown Hamiltonian $H = \sum_{A\in\{0,1,2,3\}^n} \alpha(A) \sigma_A$ can be related to the first-order time derivative of a specific time-dependent Pauli transfer matrix entry of the associated unitary time evolution, evaluated at time $0$. Importantly, this works for Pauli coefficients of arbitrarily high weight. 
Second, relying on tools from polynomial approximation \cite{howell1991derivative}, we can estimate the first-order derivative of a function via that of its Chebyshev interpolating polynomial. This requires us to prove bounds on higher-order derivatives of the function of interest. 
Thus, the third step in our proof is to control higher order derivatives of functions of the form $t\mapsto \tfrac{1}{2}\tr[\sigma_{B} e^{-itH}\rho e^{itH}]$ for an arbitrary Pauli $\sigma_B$ and an arbitrary state $\rho$. We achieve this via a rewriting in terms of iterated commutators with $H$. The concrete implementation of this last step deviates from prior work, who used structural assumptions on the Hamiltonian to obtain derivative bounds. In our case, since we do not make any prior assumptions on $H$, we work with derivative bounds that depend on the spectral norm $\norm{H}$.
In summary, we use our Pauli transfer matrix learning procedure applied to the evolution at different times to produce data points for polynomial interpolation, and then estimate the Pauli coefficients of the Hamiltonian via first-order time derivatives from the interpolating polynomial.

\subsection{Related Work}

\paragraph{Shadow tomography for quantum states}
Since \cite{aaronson02007learnability} considered a ``pretty good'' version of state tomography motivated by the framework of probably approximately correct learning, learning quantum states in different models has received significant attention.
For example, \cite{aaronson2019online, chen2022adaptive} investigated problems of online learning quantum states.
Most closely related to our work is shadow tomography for quantum states, originally proposed in \cite{aaronson2018shadow} and improved upon in \cite{badescu2021improved, huang2021information-theoretic}.
In shadow tomography, the task is to predict many expectation values of an unknown quantum from few copies.
Our work can be viewed as partially lifting shadow tomography from states to channels. Here, the ``partially'' is meant to express that, whereas shadow tomography for states allows to predict expectation values of arbitrary observables, in our channel shadow tomography we have to make some additional assumptions on the input states and output observables.

\paragraph{Classical shadows for quantum states and quantum processes}

While shadow tomography methods are efficient in terms of sample complexity, they rely on measurements on multiple copies of the unknown state and can be inefficient in terms of classical memory and computation time.
Classical shadows of quantum states, going back to \cite{paini2019approximate, huang2020predicting}, constitute a more practical variant of shadow tomography, since they use single-copy measurements and come with guarantees on classical memory and computation time. 
However, whereas shadow tomography can predict expectation values for arbitrary observables, classical shadows only work reliably for a suitably restricted set of observables, 
which depends on the random unitaries used in the shadow protocol.
\cite{levy2021classical-shadows, kunjummen2023shadow-process} lifted the classical shadow formalism from states to channels by applying the classical shadow protocol of \cite{huang2020predicting} based on randomized local Paulis to the Choi state of an unknown channel (see also \cite[Appendix H]{stilck-franca2022efficient}), which is similar to how we apply shadow tomography to the Choi state to learn transfer matrix entries of the channel from Choi access. As the local classical shadow protocol of \cite{huang2020predicting} comes with prediction guarantees only for local observables, the protocol of \cite{levy2021classical-shadows, kunjummen2023shadow-process} can only learn local reduced density matrices of the Choi state. This is not sufficient to learn the whole transfer matrix of the unknown channel, which our quantum-enhanced protocol achieves. 

\paragraph{Learning quantum channels}
In the wake of the many insights into learning quantum states, also questions of learning quantum channels have begun to attract attention.
For instance, there is a growing literature in variational quantum machine learning exploring how different kinds of complexity bounds for the quantum model, which plays the role of an unknown quantum channel to be learned, lead to sample complexity bounds \cite{caro2020pseudo, abbas2021power, bu2021effects, bu2021rademacher, gyurik2021structural, caro2021encodingdependent, chen2021expressibility, popescu2021learning, bu2022onthestatistical, cai2022sample, du2022efficient, caro2022generalization, caro2022outofdistribution}. Viewed from the perspective of channel learning, these results make assumptions on the complexity of the unknown channel and use this to bound the information-theoretic complexity of the learning task. These approaches can therefore typically not be applied to arbitrary quantum channels of high complexity.

Also in the broader quantum information theory community, channel learning tasks are gaining traction.
\cite{chung2021sample, caro2021binary, fanizza2022learning} prove sample complexity guarantees for learning channels with classical input and quantum output in a probably approximately correct learning setting.
When considering channels with quantum input and quantum output, the case of Pauli channels, relevant for modeling quantum noise, is already well studied: \cite{flammia2020efficient, harper2020efficient, harper2021fast, flammia2021pauli} give efficient procedures for learning the Pauli error rates of an unknown Pauli channel. 
In work closely related to the present paper, \cite{chen2022quantum} proved query complexity separations for learning the Pauli eigenvalues of an unknown Pauli channel, which in our language are its diagonal Pauli transfer matrix elements. They distinguish between learners with quantum memory and learners restricted either by not having access to a quantum memory or only to auxiliary systems of limited size.
Our query complexity lower bounds for learning without quantum memory show that the lower bounds of \cite[Theorem 2, (i) and (ii)]{chen2022quantum} can be strengthened significantly, achieving an exponential lower bound independently of the size of the allowed auxiliary system, if we go beyond Pauli channels with diagonal Pauli transfer matrices to general quantum channels with general Pauli transfer matrices.

Taking a broader perspective on channel learning, \cite{huang2022foundations} established fundamental results about the conditions under which quantum processes can or cannot be learned from experiments with noise.
And in very recent work, \cite{huang2022learning} extended the classical shadow formalism from states to channels, giving an efficient procedure without quantum memory for learning to predict an arbitrary unknown quantum channel w.r.t.~bounded degree $k$-local observables when averaged over any locally flat distribution of input states. In contrast, our learning algorithm relies on a quantum memory, but can predict w.r.t.~general Pauli-sparse observables and worst-case w.r.t.~Pauli-sparse input states. 

\paragraph{Hamiltonian learning}
The task of learning a Hamiltonian underlying the evolution of a quantum physical evolution has seen significant progress in recent years. 
Works such as \cite{garrison2018does, qi2019determining, bairey2019learning, li2020hamiltonian} showed that, at least in principle, generic geometrically local Hamiltonians can be learned from ``little information'', such as from a single eigenstate or from 
few generic pairs of input and output states.
\cite{anshu2021sample, rouze2021learning} gave efficient procedures for learning a geometrically local Hamiltonian from copies of its high-temperature Gibbs state, this was extended to more general low-intersection Hamiltonians while at the same time improving the sample and time complexity in \cite{haah2021optimal}.
The recent works \cite{stilck-franca2022efficient, gu2022practical, wilde2022scalably,yu2023robust} have highlighted the practical importance of learning an unknown Hamiltonian from access to its dynamics, deriving provable guarantees for geometrically local or low-intersection Hamiltonians and demonstrating the applicability of their Hamiltonian learning procedures for relatively large quantum systems.
Recently, \cite{huang2022heisenberg-scaling} proved that even Heisenberg-limited scaling for the total evolution time is achievable in the task of learning low-intersection Hamiltonians from dynamics.
Our addition to the Hamiltonian learning literature is the to our knowledge first efficient protocol for learning an arbitrary Hamiltonian with an a priori unknown structure from parallel access to its short-time dynamics.

\paragraph{Learning from quantum experiments}

\cite{huang2021information-theoretic, aharonov2022quantum} established exponential query complexity separations between certain tasks of learning from quantum experiments with and without quantum memory.
Additionally, \cite{huang2021information-theoretic} demonstrated that a class of average-case channel learning tasks does not admit a large quantum query complexity advantage, thereby emphasizing the importance of worst-case considerations.
Building on these works, \cite{chen2022exponential} introduced the learning tree framework as a powerful tool for proving query complexity lower bounds for learning without quantum memory, or with limited quantum memory \cite{chen2021hierarchy}.
This mathematical framework has been used in \cite{huang2022quantum} to prove an experimentally demonstrable quantum advantage in learning from quantum physics experiments when using a quantum memory.
Moreover, with \cite{chen2022complexity}, the learning tree formalism has recently found an application in a complexity-theoretic study of the capabilities of NISQ devices.
We use the framework of \cite{chen2022exponential} and develop new technical tools to extend their proof strategy and to derive query complexity lower bounds for Pauli transfer matrix learning without quantum memory, thereby extending their lower bounds for state shadow tomography to (Pauli-sparse) channel shadow tomography. 
Where \cite{aharonov2022quantum, chen2022exponential} considered ``channel distinguishing'' tasks, we focus more on ``channel learning'' tasks. (Note, though, that our query complexity lower bounds are proved via a channel distinguishing task.) And where \cite[Theorem 3]{huang2022quantum} established a quantum-versus-classical query complexity separation for learning a polynomial-time quantum channel on average over a distribution of states, we consider learning arbitrarily complex quantum channels in a worst-case framework over (certain) input states.

\subsection{Summary and Directions for Future Work}

In this work, we have given a provably sample-efficient procedure for learning transfer matrices with quantum memory from copies of the Choi state of an arbitrary unknown quantum channel. In the case of the Pauli transfer matrix, we obtained efficiency guarantees for classical memory and classical computation time in addition to the copy complexity bounds. This allows to efficiently learn an arbitrary unknown Hamiltonian from dynamics.
Moreover, we demonstrated that (Pauli) transfer matrix learning can be used as a subroutine for efficiently learning to predict expectation values for (Pauli-) sparse input states and output observables, a task of channel shadow tomography type.
In particular, solving these learning problems remains information-theoretically feasible even for arbitrarily complex quantum processes, which is in contrast to the idea of Occam's razor embodied in the established classical statistical learning theory literature based on complexity measures \cite{vapnik1971uniform, valiant1984theory, pollard1984convergence, kearns1994efficient, dudley1999uniform, bartlett2002rademacher}

While the tasks of transfer matrix learning and sparse expectation value prediction can be solved information-theoretically efficiently with quantum memory, we proved that no learner without quantum memory can succeed at these tasks with a subexponential number of queries to the unknown quantum channel.
Thus, these channel learning tasks show an exponential query complexity separation between learners with and without quantum memory.
Our $\Omega(\nicefrac{4^n}{\varepsilon^2})$ query complexity lower bound for learners without quantum memory is significantly stronger than the $\Omega(\nicefrac{2^n}{\varepsilon^2})$ obtained by directly applying the existing state shadow tomography lower bound to state preparation channels, showing that the quantum advantage over conventional experiments is even more distinctive in the channel learning case than in the state learning case.

We conclude by outlining some open problems raised by our work. 
Here, we begin by recalling the two questions underlying this work and aspects thereof that remain open.
\begin{itemize}
    \item Transfer matrix learning and Hamiltonian learning are efficiently solvable channel learning tasks without structural assumptions, thus serving as positive examples in answer to Question (1). Can we embed them into a broader class of feasible learning problems with arbitrary unknown quantum processes?
    \item Our exponential query complexity separation between learners with and without quantum memory answers Question (2) for transfer matrix learning. However, it breaks down for certain variants of the task (compare \cref{remark:no-exponential-separation-with-known-sparsity-structure,remark:no-exponential-separation-with-local-states-and-observables,remark:no-exponential-separation-pauli-channels,remark:no-exponential-separation-average-case}). Can we develop a more general understanding of the conditions conducive to a quantum query complexity advantage in channel learning?
\end{itemize}
Finally, we mention three questions regarding potential improvements and extensions of our results.
\begin{itemize}
    \item As our transfer matrix learning protocols with quantum memory require only Choi access to the unknown channel, it is natural to ask whether the query complexity can be further improved by allowing for sequential channel access. However, the linear-in-$n$ dependence cannot be improved, even for Pauli channels \cite[Theorem 2, (iv)]{chen2022quantum}. 
    Does sequential access allow to improve upon the $\tfrac{1}{\varepsilon^4}$-scaling?
    \item Our Hamiltonian learning complexity bounds in \Cref{theorem:hamiltonian-learning} exhibit a quartic scaling with $\nicefrac{1}{\varepsilon}$ and a linear or quadratic scaling with $n$. Can these dependencies be improved to come closer to the guarantees found for restricted Hamiltonians in prior works? Are quantum-enhanced experiments even necessary for learning arbitrary Hamiltonians from dynamics?
    \item Using \cite[Appendix A]{stilck-franca2022efficient}, we can employ our Pauli transfer matrix learning procedure to learn Lindblad generators with arbitrary Hamiltonian part and arbitrary single-site dissipation terms. Does this extend to more general Lindblad generators with arbitrary dissipative part?
\end{itemize}
Answering these questions will help understand the role of quantum memory in learning processes in quantum physics from a query complexity perspective motivated by theoretical computer science. Thereby, we can gain a better grasp of quantum learning problems solvable on near-term quantum hardware and of the experimental capabilities that need to be developed for addressing the most challenging quantum learning tasks.

\subsection{Structure of the Paper}

The remainder of the paper is structured as follows:
\Cref{section:preliminaries} introduces basic notions from quantum information as well as different models of learning from an unknown quantum channel. (See \Cref{appendix:parallel-access} for two additional learning models.)
In \Cref{section:tm-estimates-give-expectation-estimates}, we explain how approximations on the level of the transfer matrix of a channel translate to expectation value estimates under sparsity assumptions.
\Cref{section:ptm-learning} contains our efficient procedures with quantum memory for learning quantum channels via their Pauli transfer matrix, with a corresponding query complexity lower bound presented in \Cref{appendix:query-lower-bound-ptm-learning-with-quantum-memory-choi-access}. (More general transfer matrices are considered in \Cref{section:tm-learning}.)
This is then contrasted with exponential query complexity lower bounds for learning without quantum memory in \Cref{section:query-complexity-lower-bounds}, the proofs for which can be found in \Cref{appendix:proofs} and are based on the learning tree formalism \cite{chen2022exponential} shortly reviewed in \Cref{appendix:learning-tree-formalism}.
To showcase an application of transfer matrix learning, we demonstrate its applicability to Hamiltonian learning in \Cref{section:hamiltonian-learning}, with the main proof again deferred to \Cref{appendix:proofs}.

\section{Preliminaries}\label{section:preliminaries}

\subsection{Quantum Channels, the Choi-Jamiolkowski Isomorphism, and the Transfer Matrix Representation}

Here, we review basic notions from quantum information and computation that will appear throughout the rest of the paper. The presentation here is brief, we recommend textbooks such as \cite{nielsen2000quantum, heinosaari2011mathematical, watrous2018theory} or lecture notes such as \cite{wolf2012quantumchannels, preskill2020quantumcomputation}.

Throughout the paper, we work with finite-dimensional Hilbert spaces $\mathbb{C}^{d}$, with dimension $d\in\mathbb{N}_{\geq 1}= \{1,2,\ldots\}$. We employ Dirac bra-ket notation to denote elements of $\mathbb{C}^d$ by kets and elements of the dual space $(\mathbb{C}^d)^\ast$ by bras.
We use $\mathcal{B}(\mathbb{C}^{d})$ to denote the set of bounded linear operators on $\mathbb{C}^d$. 
Two norms on $\mathcal{B}(\mathbb{C}^{d})$ will be relevant for our purposes: We denote the operator norm (w.r.t.~the Euclidean norm on $\mathbb{C}^d$), also called the spectral norm or Schatten $\infty$-norm, of $X\in\mathcal{B}(\mathbb{C}^{d})$ by $\norm{X} = \sup_{\ket{\psi}\in\mathbb{C}^d: \braket{\psi}=1} \bra{\psi}X\ket{\psi}$. We denote the norm induced by the Hilbert-Schmidt inner product, also called the Frobenius norm or Schatten $2$-norm, of $X\in\mathcal{B}(\mathbb{C}^{d})$ by $\norm{X}_2 = \sqrt{\tr[X^\dagger X]}$.

The set of $d$-dimensional density matrices, the mathematical description for quantum states, is denoted by $\mathcal{S}(\mathbb{C}^d) = \left\{\rho\in\mathcal{B}(\mathbb{C}^d)~|~\rho\geq 0~\wedge~\tr[\rho]=1\right\}$. 
The rank-$1$ projections in $\mathcal{S}(\mathbb{C}^d)$ are called pure states. We typically identify a pure state $\ketbra{\psi}{\psi}\in \mathcal{S}(\mathbb{C}^d)$ with the corresponding normalized vector $\ket{\psi}\in\mathbb{C}^d$. If $\rho \in\mathcal{S}(\mathbb{C}^d\otimes \mathbb{C}^{d'})$ is a state on a composite system, we say that $\rho$ is separable if it can be written as a convex combination of tensor products of states on the two tensor factors, otherwise we call $\rho$ entangled.

We describe measurements in terms of positive operator-valued measures (POVMs). Here, a $d$-dimensional $M$-outcome POVM is a set $\{E_i\}_{i=1}^M$ of bounded linear operators $E_i\in \mathcal{B}(\mathbb{C}^{d})$ satisfying $0\leq E_i\leq \mathbbm{1}_d$ for all $1\leq i\leq M$ as well as $\sum_{i=1}^M E_i = \mathbbm{1}_d$.
According to Born's rule, the probability of observing outcome $i$ when measuring the POVM $\{E_i\}_{i=1}^M$ on the state $\rho$ is given by $\tr[E_i \rho]$.

We now turn our attention to evolutions of quantum systems. These are described mathematically by quantum channels:

\begin{definition}[Quantum channels]
    A linear superoperator $\mathcal{N}:\mathcal{B}(\mathbb{C}^{d_{\mathrm{in}}})\to \mathcal{B}(\mathbb{C}^{d_{\mathrm{out}}})$ is a quantum channel (in the Schr\"odinger picture) if $\mathcal{N}$ is completely positive (i.e., $\operatorname{id}_{\mathrm{aux}}\otimes \mathcal{N}$ is positivity-preserving for any auxiliary system) and trace-preserving.
    If $\mathcal{N}$ is also unital (i.e., satisfies $\mathcal{N}(\mathbbm{1}_{d_{\mathrm{in}}}) = \mathbbm{1}_{d_{\mathrm{out}}}$), then we call $\mathcal{N}$ a doubly-stochastic quantum channel (see \cite{mendl2009unital}).
\end{definition}

We call a quantum channel $\mathcal{N}$ entanglement-breaking if $(\operatorname{id}_{\mathrm{aux}}\otimes \mathcal{N})(\rho)$ is separable for any state $\rho\in\mathcal{S}(\mathbb{C}^{d_{\mathrm{aux}}}\otimes \mathbb{C}^{d_{\mathrm{in}}})$, see \cite{horodecki2003entanglement}.
Our main focus will be on $n$-qubit channels, for which $d_{\mathrm{in}} = d_{\mathrm{out}} =2^n$.

Linear superoperators are isomorphically related to bounded linear operators on a tensor product space. A well known isomorphism particularly useful for quantum information purposes is given in the following:

\begin{proposition}[Choi-Jamiolkowski isomorphism \cite{jamiolkowski1972linear, choi1975completely}]
    Given an orthonormal basis $\mathcal{B}=\{\ket{i}\}_{i=1}^{d_{\mathrm{in}}}\subseteq\mathbb{C}^{d_{\mathrm{in}}}$, the linear map that takes as input a linear superoperator $\mathcal{N}:\mathcal{B}(\mathbb{C}^{d_{\mathrm{in}}})\to \mathcal{B}(\mathbb{C}^{d_{\mathrm{out}}})$ and maps it to the bounded operator
    \begin{equation}\label{eq:definition-choi-state}
        \frac{1}{d_{\mathrm{in}}}\Gamma^\mathcal{N}
        \coloneqq (\operatorname{id}_{\mathrm{in}}\otimes\mathcal{N})(\Omega) \in \mathcal{B}(\mathbb{C}^{d_{\mathrm{in}}} \otimes \mathbb{C}^{d_{\mathrm{out}}}) ,
    \end{equation}
    where $\Omega = \ketbra{\Omega}{\Omega}$ with $\ket{\Omega} = \tfrac{1}{\sqrt{d_{\mathrm{in}}}}\sum_{i=1}^{d_{\mathrm{in}}} \ket{i i}$ is the canonical maximally entangled state w.r.t.~$\mathcal{B}$, is a linear isomorphism and is called the Choi-Jamiolkowski isomorphism (w.r.t.~the basis $\mathcal{B}$).
    
    Given the Choi state $\tfrac{1}{d_{\mathrm{in}}}\Gamma^\mathcal{N}$ of a quantum channel, we can describe the action of $\mathcal{N}$ on an input operator $X\in\mathcal{B}(\mathbb{C}^{d_{\mathrm{in}}})$ via $\mathcal{N}(X) = d_{\mathrm{in}} \tr_{\mathrm{in}}[ (X_{\mathrm{in}}^{\top}\otimes\mathbbm{1}_{\mathrm{out}}) \tfrac{1}{d_{\mathrm{in}}}\Gamma^\mathcal{N}]$. 
    More generally, the action of $\operatorname{id}_{\mathrm{aux}}\otimes\mathcal{N}$ on a bipartite input operator $X\in\mathcal{B}(\mathbb{C}^{d_{\mathrm{aux}}}\otimes\mathbb{C}^{d_{\mathrm{in}}})$ is given by $(\operatorname{id}_{\mathrm{aux}}\otimes\mathcal{N})(X) = d_{\mathrm{in}} \tr_{\mathrm{in}}[(X^{\top_{\mathrm{in}}}\otimes\mathbbm{1}_{\mathrm{out}}) (\mathbbm{1}_{\mathrm{aux}}\otimes \tfrac{1}{d_{\mathrm{in}}}\Gamma^\mathcal{N})]$.
\end{proposition}

From now on, we always imagine that a basis $\mathcal{B}$ of $\mathbb{C}^{d_{\mathrm{in}}}$ has been fixed, so we will not mention the choice of basis for the canonical maximally entangled state, for the Choi-Jamiolkowski isomorphism, or for matrix transpositions explicitly anymore.
We collect useful properties of the Choi-Jamiolkowski isomorphism in the next proposition:

\begin{proposition}[Properties of the Choi-Jamiolkowski isomorphism (see, e.g., {\cite[Proposition 2.1]{wolf2012quantumchannels}})]
    Let $\mathcal{N}:\mathcal{B}(\mathbb{C}^{d_{\mathrm{in}}})\to \mathcal{B}(\mathbb{C}^{d_{\mathrm{out}}})$ be a linear superoperator and let $\tfrac{1}{2^n}\Gamma^\mathcal{N}$ be its normalized Choi-Jamiolkowski operator as in \Cref{eq:definition-choi-state}.
    Then:
    \begin{enumerate}
        \item $\mathcal{N}$ is Hermiticity-preserving if and only if $\tfrac{1}{2^n}\Gamma^\mathcal{N}$ is Hermitian.
        \item $\mathcal{N}$ is completely positive if and only if $\tfrac{1}{2^n}\Gamma^\mathcal{N}$ is positive semidefinite.
        \item $\mathcal{N}$ is trace-preserving if and only if $\tr_{\mathrm{out}}[\tfrac{1}{2^n}\Gamma^\mathcal{N}] = \tfrac{1}{2^{d_{\mathrm{in}}}}\mathbbm{1}_{\mathrm{in}}$.
        \item $\mathcal{N}$ is unital if and only if $\tr_{\mathrm{in}}[\tfrac{1}{2^n}\Gamma^\mathcal{N}] = \tfrac{1}{2^{d_{\mathrm{out}}}}\mathbbm{1}_{\mathrm{out}}$.
    \end{enumerate}
\end{proposition}

In particular, we see that $\mathcal{N}$ is a quantum channel if and only if $\tfrac{1}{2^n}\Gamma^\mathcal{N}$ is a quantum state with maximally mixed first marginal. Moreover, recalling that $\mathcal{N}$ is entanglement breaking if and only if $\tfrac{1}{2^n}\Gamma^\mathcal{N}$ is separable \cite{horodecki2003entanglement}, we see that $\mathcal{N}$ is a doubly-stochastic and entanglement-breaking quantum channel if and only if $\tfrac{1}{2^n}\Gamma^\mathcal{N}$ is a separable quantum state whose first and second marginal both are maximally mixed.

The Choi-Jamiolkowski isomorphism already provides us with a useful representation for a quantum channel.
The next definition introduces an alternative representation obtained by choosing an orthonormal basis for the space of bounded operators and then considering the matrix representation of a quantum channel w.r.t.~that basis. For convenience, we focus on the case of multi-qubit channels:

\begin{definition}[General transfer matrices]\label{definition:transfer-matrices}
    Let $\mathcal{Q}=\{Q_i\}_{i=1}^{4^n}$ be an orthonormal basis (ONB) for $\mathcal{B}((\mathbb{C}^2)^{\otimes n})$ w.r.t.~the Hilbert-Schmidt inner product, consisting of Hermitian bounded linear operators $Q_i=Q_i^\dagger\in \mathcal{B}((\mathbb{C}^2)^{\otimes n})$.
    Let $\mathcal{N}:\mathcal{B}(\mathbb{C}^{2^n})\to \mathcal{B}(\mathbb{C}^{2^n})$ be an $n$-qubit quantum channel.
    The \emph{transfer matrix (TM)} of $\mathcal{N}$ w.r.t.~$\mathcal{Q}$ is the matrix $R_\mathcal{N}^\mathcal{Q}\in\mathbb{C}^{4^n\times 4^n}$ with entries
    \begin{equation}
        \left( R_\mathcal{N}^\mathcal{Q}\right)_{i,j}
        \coloneqq \tr [Q_i \mathcal{N}(Q_j)],\quad\forall 1\leq i,j\leq 4^n.
    \end{equation}
\end{definition}

As we are interested in channels acting on systems of multiple qubits, we can consider a relevant special case of \Cref{definition:local-transfer-matrices} by choosing the overall ONB to consist of tensor products of elements of a single-qubit ONB:

\begin{definition}[Local transfer matrices]\label{definition:local-transfer-matrices}
    Let $\mathcal{Q}=\{Q_i\}_{i=1}^{4}$ be a local orthonormal basis (ONB) for $\mathcal{B}(\mathbb{C}^2)$ w.r.t.~the Hilbert-Schmidt inner product, consisting of Hermitian bounded linear operators $Q_i=Q_i^\dagger\in \mathcal{B}(\mathbb{C}^2)$, and assume $Q_0=\tfrac{1}{\sqrt{2}}\mathbbm{1}_2$.
    Let $\mathcal{N}:\mathcal{B}(\mathbb{C}^{2^n})\to \mathcal{B}(\mathbb{C}^{2^n})$ be an $n$-qubit quantum channel.
    The \emph{local transfer matrix} of $\mathcal{N}$ w.r.t.~$\mathcal{Q}$ is the matrix $R_\mathcal{N}^\mathcal{Q}\in\mathbb{C}^{\{0,1,2,3\}^n \times \{0,1,2,3\}^n}$ with entries
    \begin{equation}
        \left( R_\mathcal{N}^\mathcal{Q}\right)_{A,B}
        \coloneqq \tr [Q_{A} \mathcal{N}(Q_{B})],\quad\forall A,B\in \{0,1,2,3\}^n,
    \end{equation}
    where we used the notation
    \begin{equation}
        Q_{A}
        = \bigotimes_{k=1}^n Q_{A_k},\quad\forall A=(A_1,\ldots,A_n)\in \{0,1,2,3\}^n.
    \end{equation}
\end{definition}

Note that we use $A,B,C,\ldots$ to denote elements of $\{0,1,2,3\}^n$ (or sometimes of $\{0,1,2,3\}^{2n}$), which leads to operators indexed by $A,B,C,\ldots$ to express tensor products. 
This is not to be confused with the indexing of operators by the subsystem they act on, for which sometimes a similar notation is used.
When an indexing by subsystems becomes necessary, we will do so by giving the subsystems names such as $\mathrm{in}$ (for ``input''), $\mathrm{out}$ (for ``output''), or $\mathrm{aux}$ (for ``auxiliary'').

We obtain the prime example for a (local) TM when considering the ONB of normalized Pauli strings:

\begin{example}[Pauli transfer matrix]\label{example:pauli-transfer-matrix}
    Consider the local ONB of normalized single-qubit Paulis, $\mathcal{P}=\{\tfrac{1}{\sqrt{2}}\sigma_a\}_{a\in\{0,1,2,3\}}$.
    For an $n$-qubit quantum channel $\mathcal{N}:\mathcal{B}(\mathbb{C}^{2^n})\to \mathcal{B}(\mathbb{C}^{2^n})$, this gives rise to the Pauli transfer matrix (PTM) $R_\mathcal{N}^\mathcal{P}\in\mathbb{C}^{\{0,1,2,3\}^n\times \{0,1,2,3\}^n}$, a local TM, with entries
    \begin{equation}
        \left( R_\mathcal{N}^\mathcal{P}\right)_{A,B}
        \coloneqq \frac{1}{2^n}\tr [\sigma_A \mathcal{N}(\sigma_B)],\quad\forall A,B \in \{0,1,2,3\}^n.
    \end{equation}
    Note that for a quantum channel $\mathcal{N}$, all PTM entries are real and lie in the interval $[-1,1]$, compare, e.g., \cite[Section 2.1.3]{greenbaum2015introduction}.
\end{example}

\subsection{Learning quantum channels with and without quantum memory}\label{section:different-learning-models}

In this section, we describe different models of learning quantum channels. These models differ in the type of access to the unknown quantum channel, for example depending on whether the learner has a quantum memory at their disposal or whether the learner can actively choose the input states on which the unknown quantum channel is queried. 
We first recall two definitions from \cite{chen2021exponential-arxiv}. The first considers general adaptive procedures for learning an unknown quantum channel if the learner does not have access to a quantum memory:

\begin{figure}
    \centering
    \includegraphics[width = 0.8\textwidth]{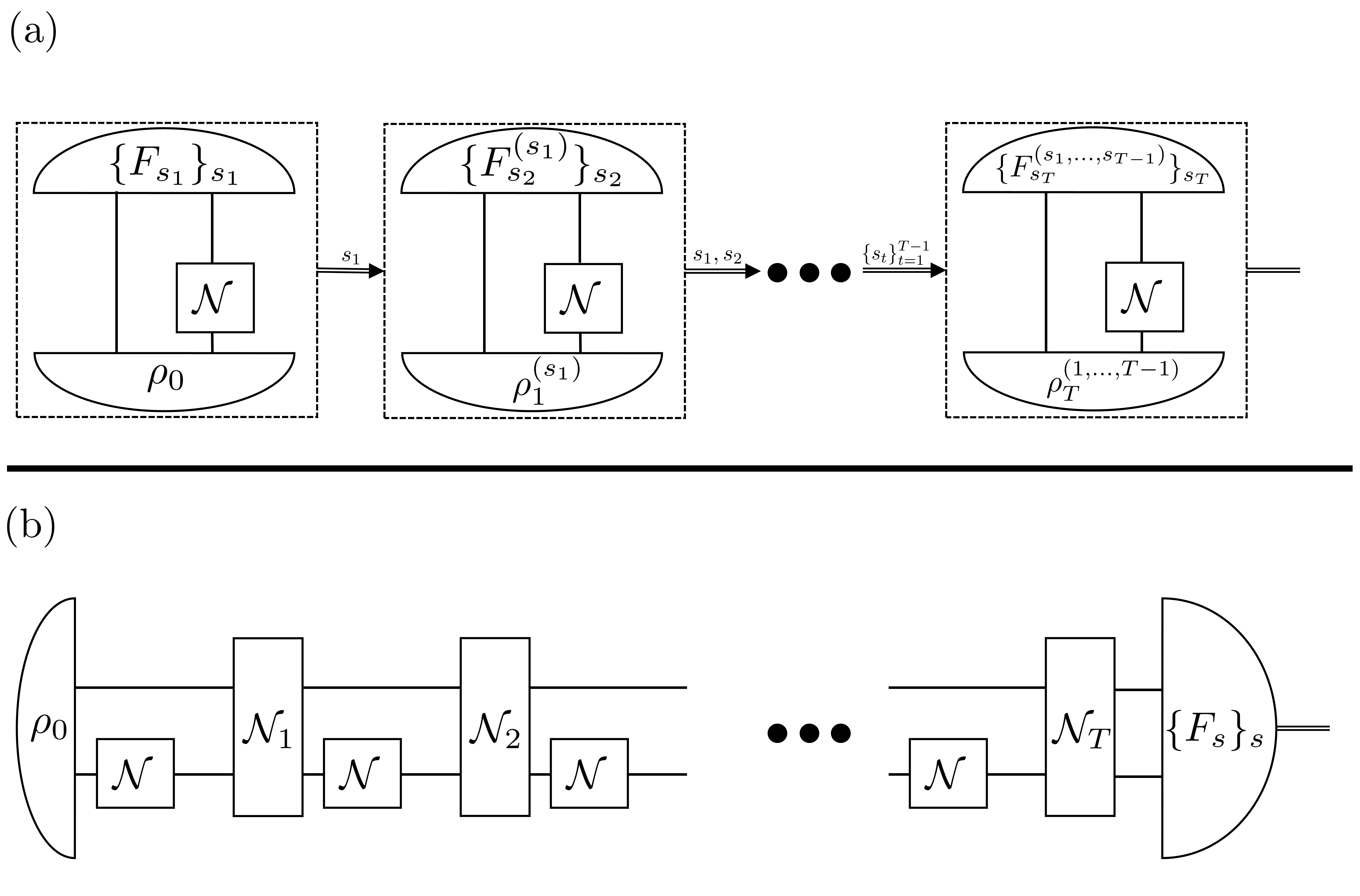}
    \caption{Illustration of learning quantum channels from general channel access, each panel to be read from left to right. Panel (a) depicts a learner without quantum memory, panel (b) depicts a learner with quantum memory.}
    \label{fig:general-access}
\end{figure}

\begin{definition}[Learning quantum channels without quantum memory~{\cite[Definition 4.17]{chen2021exponential-arxiv}}]\label{definition:qchannel-learning-without-qmemory}
    An algorithm for learning an unknown $n$-qubit quantum channel $\mathcal{N}$ without quantum memory can obtain classical data from the channel oracle by preparing an arbitrary input state $\rho\in\mathcal{B}((\mathbb{C}^2)^{\otimes n_{\mathrm{aux}}}\otimes (\mathbb{C}^2)^{\otimes n})$, evolving under $\mathcal{N}$ to yield the output state $(\operatorname{id}_{\mathrm{aux}}\otimes\mathcal{N})(\rho)$, and performing an arbitrary POVM measurement $\{F_s\}_s \subseteq\mathcal{B}((\mathbb{C}^2)^{\otimes n_{\mathrm{aux}}}\otimes (\mathbb{C}^2)^{\otimes n})$ on that output state to obtain outcome $s$ with probability $\Tr[F_s (\operatorname{id}_{\mathrm{aux}}\otimes\mathcal{N})(\rho) ]$.
    Here, the selection of the input state and the POVM can be adaptive, i.e., depend on previously chosen input states, previously chosen POVMs, and previously observed measurement outcomes.
    After $T$ oracle accesses, the algorithm uses the entirety of the obtained measurement outcomes to predict properties of $\mathcal{N}$.
\end{definition}

General learning procedures in which the learner does have a quantum memory at their disposal are encapsulated in the following definition:

\begin{definition}[Learning quantum channels with quantum memory~{\cite[Definition 4.18]{chen2021exponential-arxiv}}]\label{definition:qchannel-learning-with-qmemory}
    An algorithm for learning an unknown $n$-qubit quantum channel $\mathcal{N}$ with quantum memory can access the channel oracle as a quantum channel during a mixed-state quantum computation.
    The resulting state of the quantum memory after $T$ oracle accesses is
    \begin{equation}
        \rho^{\mathcal{N}}_T = \mathcal{N}_T (\operatorname{id}_{\mathrm{aux},T-1}\otimes\mathcal{N}) \ldots \mathcal{N}_2 (\operatorname{id}_{\mathrm{aux},1}\otimes\mathcal{N}) \mathcal{N}_1 (\operatorname{id}_{\mathrm{aux},0}\otimes\mathcal{N})(\rho_0),
    \end{equation}
    with some input state $\rho_0\in\mathcal{S}((\mathbb{C}^2)^{\otimes n_{\mathrm{aux},0}}\otimes (\mathbb{C}^2)^{\otimes n})$ and some $n_{t-1}$-to-$n_t$ qubit quantum channels $\mathcal{N}_t :\mathcal{B}((\mathbb{C}^2)^{\otimes n_{t-1}})\to\mathcal{B}((\mathbb{C}^2)^{\otimes n_{t}})$, for $t = 1, \ldots, T$, where $n_0=n$.
    After $T$ oracle accesses, the algorithm performs a POVM $\{F_s\}_s \subseteq\mathcal{B}((\mathbb{C}^2)^{\otimes n_{T}})$ on the quantum memory state $\rho^{\mathcal{N}}_T$ to predict properties of $\mathcal{N}$.
\end{definition}

\Cref{definition:qchannel-learning-without-qmemory,definition:qchannel-learning-with-qmemory} are illustrated in \Cref{fig:general-access}.
We will also consider variants of the previous two definitions. To obtain these variants, we suppose that the learner has access to the unknown quantum channel $\mathcal{N}$ only via its Choi state $\tfrac{1}{2^n}\Gamma^\mathcal{N}$. In analogy to \cite[Definitions 4.15 and 4.16]{chen2021exponential-arxiv}, this leads us to the following two definitions:

\begin{figure}
    \centering
    \includegraphics[width = 0.6\textwidth]{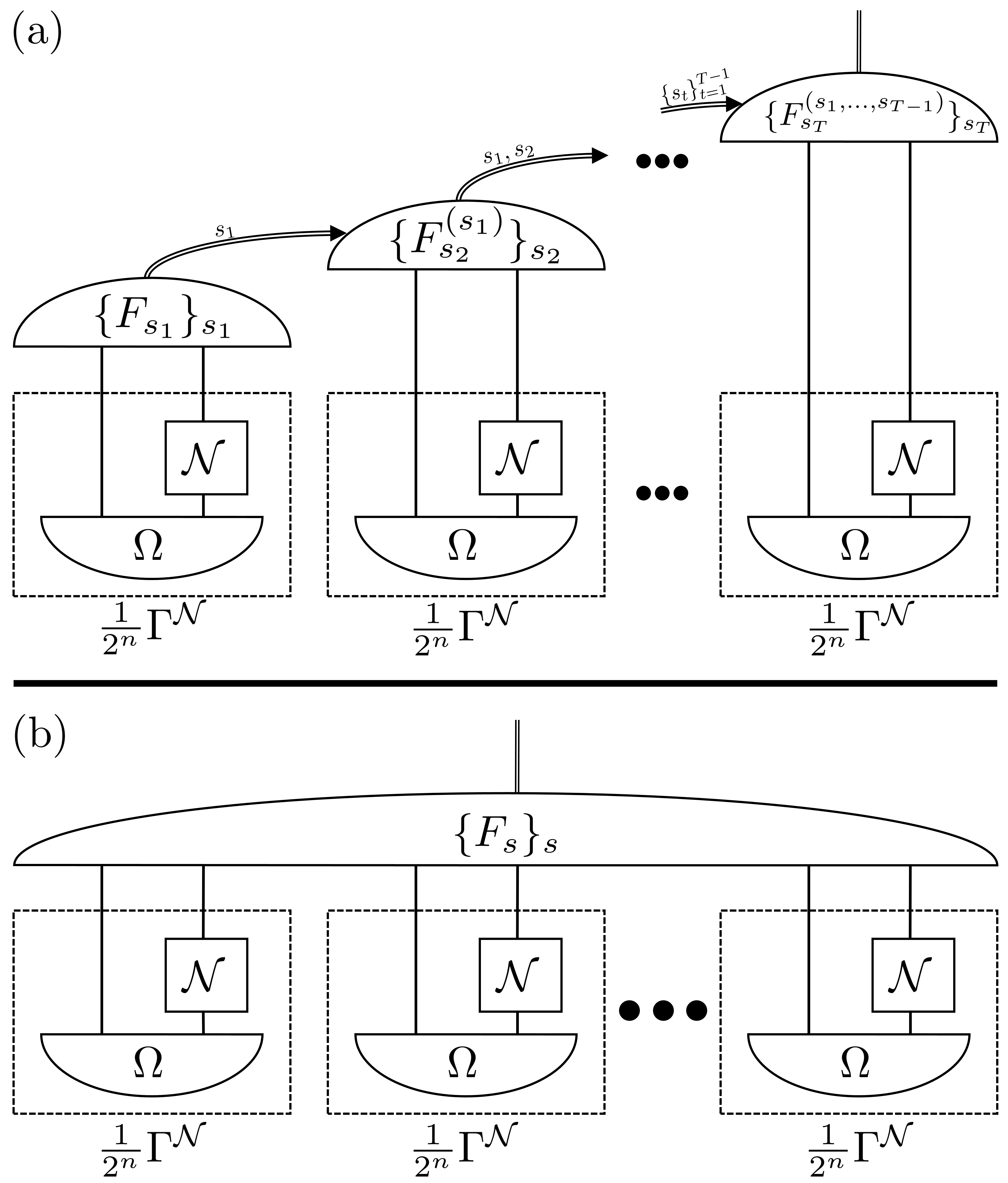}
    \caption{Illustration of learning quantum channels from Choi access, each panel to be read from bottom to top. Panel (a) depicts a learner without quantum memory, panel (b) depicts a learner with quantum memory.}
    \label{fig:choi-access}
\end{figure}

\begin{definition}[Learning quantum channels without quantum memory from Choi access]\label{definition:qchannel-learning-without-qmemory-choi}
    An algorithm for learning an unknown $n$-qubit quantum channel $\mathcal{N}$ without quantum memory from Choi access can obtain classical data from the Choi state oracle by performing arbitrary adaptively chosen POVM measurements on single copies of the Choi state $\tfrac{1}{2^n}\Gamma^\mathcal{N}$.
    That is, for each oracle access, the algorithm can select a POVM $\{F_s\}_s\subseteq \mathcal{B}((\mathbb{C}^2)^{\otimes n}\otimes (\mathbb{C}^2)^{\otimes n})$ and obtain the classical outcome $s$ with probability $\tr[F_s \tfrac{1}{2^n}\Gamma^\mathcal{N}]$. Here, the selection of the POVM can be adaptive, i.e., depend on previously chosen POVMs and previously observed measurement outcomes.
    After $T$ oracle accesses, the algorithm uses the entirety of the obtained measurement outcomes to predict properties of $\mathcal{N}$.
\end{definition}

\begin{definition}[Learning quantum channels with quantum memory from Choi access]\label{definition:qchannel-learning-with-qmemory-choi}
    An algorithm for learning an unknown $n$-qubit quantum channel $\mathcal{N}$ with quantum memory from Choi access can obtain copies of the Choi state $\tfrac{1}{2^n}\Gamma^\mathcal{N}$ from the Choi state oracle and store those copies in the quantum memory.
    After $T$ oracle accesses, the algorithm performs a joint POVM $\{F_s\}_s\subseteq\mathcal{B}(((\mathbb{C}^2)^{\otimes n}\otimes (\mathbb{C}^2)^{\otimes n})^{\otimes T})$ on $(\tfrac{1}{2^n}\Gamma^\mathcal{N})^{\otimes T}$ to predict properties of $\mathcal{N}$. 
\end{definition}

We illustrate learning channels from Choi access in \Cref{fig:choi-access}. 
Any procedure for learning a quantum channel from Choi access in particular is an instance of a general procedure for learning a quantum channel. To see this, simply use independent copies of a maximally entangled state as input in \Cref{definition:qchannel-learning-without-qmemory,definition:qchannel-learning-with-qmemory}.
However, as the implementation of an $n$-qubit quantum channel via teleportation of its Choi state in general has a success probability $4^{-n}$, not every general procedure for learning a quantum channel may be realizable via learning from Choi access. However, under additional assumptions on shared local unitary invariances of the Choi states of the unknown channels, the success probability of implementation via teleportation can be increased  (compare, e.g., \cite[Section 2.1]{wolf2012quantumchannels}), so reducing general channel learning to learning from Choi access may become feasible. 
Finally, in \Cref{appendix:parallel-access} we describe a model of channel learning from parallel access, which is intermediate between Choi access and general channel access.

\section{Predicting Expectation Values From Transfer Matrix Entries Under Sparsity Assumptions}\label{section:tm-estimates-give-expectation-estimates}

Before establishing our results on learning transfer matrices, we first discuss how approximations to entries of the transfer matrix of a channel $\mathcal{N}$ can be used to obtain estimates for expectation values of the form $\tr[O \mathcal{N}(\rho)]$. 
Namely, we first observe: The ability to simultaneously approximate entries $(R_\mathcal{N}^\mathcal{Q})_{i,j}$ of the transfer matrix of a channel $\mathcal{N}$ up to accuracy $\varepsilon$ allows to approximate the expectation value $\tr[O \mathcal{N}(\rho)]$ up to accuracy $\tilde{\varepsilon}$, where $\nicefrac{\tilde{\varepsilon}}{\varepsilon}$ can be controlled in terms of the $1$-norms of the $\mathcal{Q}$-basis coefficient vectors of $O$ and $\rho$. This is in particular useful if both $O$ and $\rho$ have sparse $\mathcal{Q}$-basis expansions.
Our second observation is similar: If the transfer matrix of $\mathcal{N}$ is sparse, then $\varepsilon$-accurate estimates of those entries translate to $\tilde{\varepsilon}$-accurate estimates of $\tr[O \mathcal{N}(\rho)]$ for any $O$ and $\rho$, where $\nicefrac{\tilde{\varepsilon}}{\varepsilon}$ depends on the sparsity of the transfer matrix. 
In summary: Under sparsity assumptions, accurate estimates of transfer matrix entries give rise to accurate estimates of expectation values.

For this section, we always take $\mathcal{Q}=\{Q_i\}_{i=1}^{4^n}$ to be a Hermitian ONB of $\mathcal{B}((\mathbb{C}^2)^{\otimes n})$. Accordingly, the transfer matrices in this subsection will be as in \Cref{definition:transfer-matrices}.

\subsection{Sparse Input States and Sparse Output Observables}

We begin by showing how estimates for the transfer matrix translate to expectation value estimates under $1$-norm assumptions on the coefficient vectors of state and observable:

\begin{lemma}\label{lemma:tm-entry-estimates-give-expectation-estimates-1-norm}
    Let $\rho\in\mathcal{S}(\mathbb{C}^{2^n})$ be an $n$-qubit quantum state whose coefficient vector in the $\mathcal{Q}$-basis expansion has $1$-norm $B_\rho>0$. 
    Let $O\in\mathcal{B}(\mathbb{C}^{2^n})$ be a self-adjoint $n$-qubit observable whose coefficient vector in the $\mathcal{Q}$-basis expansion has $1$-norm $B_O>0$. 
    Assume that we have $\varepsilon$-accurate estimates $\hat{r}_{i,j}^\mathcal{Q}$ for the entries of the $\mathcal{Q}$-TM of an $n$-qubit quantum channel $\mathcal{N}:\mathcal{B}(\mathbb{C}^{2^n})\to \mathcal{B}(\mathbb{C}^{2^n})$. Then, from these estimates we can obtain an $(\varepsilon B_\rho B_O)$-accurate estimate $\hat{\mu}$ of $\tr [O\mathcal{N}(\rho)]$.
\end{lemma}
\begin{proof}
    Consider the $\mathcal{Q}$-basis expansions
    \begin{align}
        \rho
        &=\sum_{i=1}^{4^n}\alpha_iQ_i,~\alpha_i=\Tr [\rho Q_i],\\
        O
        &=\sum_{i=1}^{4^n}\beta_iQ_i,~\beta_i=\Tr [OQ_i],
    \end{align}
    of $\rho$ and of $O$. 
    By expanding in the $\mathcal{Q}$-basis, we can rewrite our quantity of interest, $\Tr [O\mathcal{N}(\rho)]$, as
    \begin{equation}\label{eq:expectation-pauli-expansion}
        \Tr [O\mathcal{N}(\rho)]
        =\sum_{i,j=1}^{4^n}\alpha_i\beta_j(R_\mathcal{N}^\mathcal{Q})_{j,i}\, .
    \end{equation}
    Accordingly, we define the estimate $\hat{\mu}$ as
    \begin{equation}\label{eq:expectation-pauli-expansion-estimate}
        \hat{\mu}
        \coloneqq \sum_{i,j=1}^{4^n}\alpha_i\beta_j\hat{r}_{j,i}^\mathcal{Q}\, ,
    \end{equation}
    with the the estimates $\hat{r}_{i,j}^\mathcal{Q}$ for the entries of the $\mathcal{Q}$-TM satisfying $\lvert \hat{r}_{i,j}^\mathcal{Q} - (R_\mathcal{N}^\mathcal{Q})_{j,i}\rvert\leq\varepsilon$ for all $1\leq i,j\leq 4^n$. 
    Then, we can argue for the accuracy of $\hat{\mu}$ as follows: 
    \begin{align}
        \lvert \hat{\mu} - \Tr [O\mathcal{N}(\rho)]\rvert 
        &= \left\lvert \sum_{i,j=1}^{4^n}\alpha_i\beta_j (\hat{r}_{j,i}^\mathcal{Q} - (R_\mathcal{N}^\mathcal{Q})_{j,i}) \right\rvert\\
        &\leq \varepsilon\cdot \sum_{i,j=1}^{4^n}\lvert\alpha_i\rvert\cdot\lvert\beta_j\rvert\\
        &= \varepsilon \lVert\vec{\alpha}\rVert_1\cdot \lVert\vec{\beta}\rVert_1\\
        &= \varepsilon B_\rho B_O\, .
    \end{align}
    Here, the inequality is an application of H\"older's inequality (with $p=1$ and $q=\infty$).
\end{proof}

As sparsity assumptions lead to $1$-norm bounds, this has the following consequence:

\begin{corollary}\label{corollary:tm-estimates-give-sparse-expectation-estimates}
    Let $\rho\in\mathcal{S}(\mathbb{C}^{2^n})$ be an $n$-qubit quantum state whose coefficient vector in the $\mathcal{Q}$-basis expansion is $s_\rho$-sparse.
    Let $O\in\mathcal{B}(\mathbb{C}^{2^n})$ be a self-adjoint $n$-qubit observable whose coefficient vector in the $\mathcal{Q}$-basis expansion is $s_O$-sparse.
    Assume that we have $\varepsilon$-accurate estimates $\hat{r}_{i,j}^\mathcal{Q}$ for the entries of the $\mathcal{Q}$-TM of an $n$-qubit quantum channel $\mathcal{N}:\mathcal{B}(\mathbb{C}^{2^n})\to \mathcal{B}(\mathbb{C}^{2^n})$. Then, from these estimates we can obtain an $\left(\varepsilon s_\rho\sqrt{s_O} \norm{O}_2 \max_{1\leq i\leq 4^n} \lVert Q_i\rVert_\infty\right)$-accurate estimate of $\tr [O\mathcal{N}(\rho)]$.
\end{corollary}
\begin{proof}
    We use the notation from the previous proof for $\mathcal{Q}$-basis coefficients.
    First observe that, by orthonormality of the $Q_i$,  Parseval's equality yields
    \begin{equation}
        \lVert\rho\rVert_2^2
        =\sum_{i=1}^{4^n}\lvert\alpha_i\rvert^2 = \lVert\vec{\alpha}\rVert_2^2\, ,
        \quad 
        \lVert O\rVert_2^2=\sum_{i=1}^{4^n}\lvert\beta_i\rvert^2 = \lVert\vec{\beta}\rVert_2^2 \, .
    \end{equation}
    Now, if $\vec{\alpha}$ is $s_\rho$-sparse and $\vec{\beta}$ is $s_O$-sparse, Hölder's inequality tells us that
    \begin{equation}\label{eq-sparse_Holder}
        \lVert\vec{\alpha}\rVert_1
        \leq s_\rho\lVert\vec{\alpha}\rVert_\infty 
        = s_\rho \max_{1\leq i\leq 4^n}\lvert\Tr[\rho Q_i]\rvert
        \leq s_\rho \max_{1\leq i\leq 4^n}\lVert\rho\rVert_1\cdot \lVert Q_i\rVert_\infty
        =  s_\rho \max_{1\leq i\leq 4^n} \lVert Q_i\rVert_\infty\, ,
    \end{equation}
    as well as 
    \begin{equation}\label{eq-sparse_Frobenius}
        \lVert\vec{\beta}\rVert_1\leq \sqrt{s_O}\lVert\vec{\beta}\rVert_2 = \sqrt{s_O}\lVert O\rVert_2\, .
    \end{equation}
    Plugging these bounds into the statement of the previous lemma finishes the proof.
\end{proof}

\begin{remark}\label{remark:alternative-normalization}
    In the above proofs, we have chosen to expand both $\rho$ and $O$ w.r.t.~$\mathcal{Q}$. This leads to the norm bounds $\norm{\vec{\alpha}}_\infty\leq \max_{1\leq i\leq 4^n} \lVert Q_i\rVert_\infty$ and $\lVert \vec{\beta}\rVert_\infty\leq \norm{O}_2$. 
    From a quantum information perspective, one may instead consider the alternative expansions $\rho = \sum_{i=1}^{4^n} \tilde{\alpha}_i \max_{1\leq j\leq 4^n} \lVert Q_j\rVert_\infty Q_j$ with $\tilde{\alpha}_i=\tfrac{1}{\max_{1\leq j\leq 4^n} \lVert Q_i\rVert_\infty} \Tr [\rho Q_i]$, as well as $O= \sum_{i=1}^{4^n}\tilde{\beta}_i \tfrac{1}{\max_{1\leq j\leq 4^n} \lVert Q_j\rVert_\infty} Q_i$ with $\tilde{\beta}_i = \max_{1\leq j\leq 4^n} \lVert Q_j\rVert_\infty \Tr [OQ_i]$.
    These lead to the norm bounds $\lVert \vec{\tilde{\alpha}} \rVert_\infty \leq 1$ and $\lVert \vec{\tilde{\beta}}\rVert_\infty\leq \max_{1\leq j\leq 4^n} \lVert Q_i\rVert_\infty\norm{O}_2$.
    In particular, if $\mathcal{Q}=\mathcal{P}$ is the Pauli ONB, we have the operator norm bound $\max_{1\leq j\leq 4^n} \lVert Q_i\rVert_\infty = \tfrac{1}{\sqrt{2^n}}$, and if $O$ has bounded operator norm $\norm{O}\leq B$, we get $\lVert \vec{\tilde{\alpha}} \rVert_\infty \leq 1$ and $\lVert \vec{\tilde{\beta}}\rVert_\infty\leq B$.
    Of course, this change in normalization does not change the sparsity of the respective expansions and also leaves the product of the $1$-norms of the coefficient vectors invariant. Thus, the results of \Cref{lemma:tm-entry-estimates-give-expectation-estimates-1-norm} and \Cref{corollary:tm-estimates-give-sparse-expectation-estimates} remain untouched by such a simultaneous change in normalization.
\end{remark}

In the remainder of this paper, we will sometimes for brevity  call $O$ (or $\rho$) $s_O$-sparse (or $s_\rho$-sparse) w.r.t.~$\mathcal{Q}$ if the respective coefficient vector $(\tr[O Q_i])_{i=1}^{4^n}$ (or $(\tr[\rho Q_i])_{i=1}^{4^n}$) in the $\mathcal{Q}$-basis is $s_O$-sparse (or $s_\rho$-sparse). If $\mathcal{Q}=\mathcal{P}$ is the Pauli ONB, we will speak of Pauli-sparse observables and states.
To illustrate potential applications, we next discuss some examples of observables and states with a sparse ONB expansion. We do so for the relevant special case of the Pauli ONB.

\begin{example}[Pauli-sparse observables]\label{example:pauli-sparse-observables}
    The following constitute examples for observables with a sparse Pauli-ONB expansion:
    \begin{itemize}
        \item All $n$-qubit Pauli observables of arbitrarily high weight have Pauli-sparsity $s_O = 1$, independently of $n$.
        \item All geometrically $k$-local $n$-qubit observables have Pauli-sparsity $s_O = \mathcal{O}(4^k n^D)$, assuming the underlying lattice is a subset of $\mathbb{Z}^D$.
        \item All $k$-local $n$-qubit observables have Pauli-sparsity $s_O \leq \binom{n}{k}\cdot 4^k \leq \mathcal{O}((4n)^k)$.
        \item All $k$-local degree-$d$ $n$-qubit observables, that is, $k$-local $n$-qubit observables made up of summands such that each qubit is acted on by at most $d$ of them (see \cite{huang2022learning} for a formal definition), have Pauli-sparsity $s_O \leq \mathcal{O}(n d 4^k)$.
    \end{itemize}
\end{example}

\begin{remark}\label{remark:improved-1-norm-bound-from-huang-et-al}
    For these last two classes of observables, one can alternatively use \cite[Corollaries 3 and 4]{huang2022learning} to get better bounds on the $1$-norm of the vector of Pauli coefficients of $O$.
    Note, however, that \cite{huang2022learning} considers the unnormalized Pauli expansion of an observable. 
    If we account for this (compare also the discussion in \Cref{remark:alternative-normalization}), then in the notation of \Cref{lemma:tm-entry-estimates-give-expectation-estimates-1-norm}, \cite[Corollary 3]{huang2022learning} combined with Hölder's inequality implies that for $k$-local $O$, we have:
    \begin{equation}
        \lVert \vec{\beta}\rVert_{1}
        \leq \lVert\vec{\beta}\rVert_{\nicefrac{2k}{(k+1)}} s_O^{\nicefrac{(k-1)}{2k}}
        \leq \frac{3\sqrt{2^n}}{C(k)}\norm{O}\cdot \mathcal{O}\left( (4n)^{\tfrac{k-1}{2}}\right),
    \end{equation}
    where $C(k)=\left(\exp\left(\Theta (k \log k) \right)\right)^{-1}$.
    Even more directly, \cite[Corollary 4]{huang2022learning} implies for a $k$-local degree-$d$ observable $O$:
    \begin{equation}
        \lVert \vec{\beta}\rVert_{1}
        \leq \frac{3\sqrt{2^n}}{C(k,d)}\norm{O}, 
    \end{equation}
    with $C(k,d)=\left(\sqrt{d}\exp\left(\Theta (k \log k) \right)\right)^{-1}$.
    Both of these $1$-norm bounds improve in terms of $n$-dependence upon the more naive $\lVert \vec{\beta}\rvert_1\leq \sqrt{s_O}\norm{O}_2\leq \sqrt{s_O 2^n}\norm{O}$ that we would obtain by combining the respective sparsity observed in \Cref{example:pauli-sparse-observables} with the reasoning that led from \Cref{lemma:tm-entry-estimates-give-expectation-estimates-1-norm} to \Cref{corollary:tm-estimates-give-sparse-expectation-estimates}.
    This demonstrates that the bound of \Cref{corollary:tm-estimates-give-sparse-expectation-estimates}, based on exact sparsity, can be looser than the $1$-norm-based bound of \Cref{lemma:tm-entry-estimates-give-expectation-estimates-1-norm}.
\end{remark}

\begin{example}[Pauli-sparse states]
    The following constitute examples for states with a sparse Pauli-ONB expansion:
    \begin{itemize}
        \item For any non-identity $n$-qubit Pauli $\sigma_A$, indexed by a string $A\in\{0,1,2,3\}^n$ with $A\neq 0^n$, the mixed state $\rho_\pm= \tfrac{1}{2^n} (\mathbbm{1}_{2^n} \pm \sigma_A)$ proportional to the projector on the $\pm 1$ eigenspace of $\sigma_A$ has Pauli-sparsity $s_{\rho_\pm} = 2$, independently of $n$.
        \item Let $\sigma_{A_1},\ldots,\sigma_{A_k}$ be pairwise distinct $n$-qubit Paulis. Then, the mixed state $\rho_\pm = \tfrac{1}{2^{n}}\prod_{\ell =1}^k (\mathbbm{1}_{2^n} \pm \sigma_{A_\ell}) = \tfrac{1}{2^{n}}\sum_{b\in\{0,1\}^k} (-1)^{\lvert b\rvert} \sigma_{A_1}^{b_1}\ldots\sigma_{A_k}^{b_k}$ proportional to the projector on the joint $\pm 1$ eigenspace of $\sigma_{A_1},\ldots,\sigma_{A_k}$ has Pauli-sparsity $s_{\rho_\pm} = \lvert \langle \sigma_{A_1},\ldots,\sigma_{A_k}\rangle \rvert = 2^k$.
    \end{itemize}
\end{example}

Having given these examples of observables and states that we can cover, we remark on an important limitation on the class of sparse states:

\begin{remark}\label{remark:sparse-states-limited-purity}
    The class of quantum states with an $s_\rho$-sparse $\mathcal{Q}$-basis expansion only contains states with a limited purity. With the notation used in the proof of \Cref{lemma:tm-entry-estimates-give-sparse-channel-expectation-estimates}, this can be seen as follows: If $\rho$ is $s_\rho$-sparse w.r.t.~$\mathcal{Q}$, then its purity can be bounded as
    \begin{align}
        \tr[\rho^2]
        = \norm{\rho}_2^2
        = \norm{\vec{\alpha}}_2^2
        \leq s_\rho \norm{\vec{\alpha}}_\infty^2
        \leq s_\rho \max_{1\leq i\leq 4^n} \lVert Q_i\rVert_\infty^2 .
    \end{align}
    As $1 = \tr[\rho] \leq \operatorname{rank}(\rho)\tr[\rho^2]$, this also implies that $\operatorname{rank}(\rho)\geq (s_\rho \max_{1\leq i\leq 4^n} \lVert Q_i\rVert_\infty^2)^{-1}$.
    For our results in later sections, we will assume $\max_{1\leq i\leq 4^n} \lVert Q_i\rVert_\infty = \tfrac{1}{\sqrt{2^n}}$ and $s_\rho\leq\mathcal{O}(\poly (n))$, with $n$ the number of qubits. 
    Under these assumptions, we see that any $s_\rho$-sparse state $\rho$ has purity at most $\mathcal{O}(\tfrac{\poly (n)}{2^n})$ and rank at least $\Omega (\tfrac{2^n}{\poly (n)})$. 
    Thus, such states are highly mixed and high-rank.
    Note, however, that sparse states can be far from maximally mixed in trace distance. 
    For instance, if $A\in\{0,1,2,3\}^n$ with $A\neq 0^n$, then $\rho_\pm = \tfrac{\mathbbm{1}_{2^n} \pm \sigma_A}{2^n}$ has Pauli-sparsity $2$ and $\norm{\rho_\pm - \tfrac{1}{2^n}\mathbbm{1}_{2^n}}_1 = 1$.
\end{remark}

We also note that the results of this subsection as well as the results of later sections all straightforwardly extend to observables $O$ that are close to sparse in operator norm and to states that are close to sparse in $1$-norm, since the true expectation value is then well approximated by the one of the sparse approximations. This, however, does not add conceptual insight beyond the case of exact sparsity, so we do not explicitly mention the approximate sparsity variant of the result each time.

\subsection{Sparse Quantum Channel}

In the previous subsection, we made assumptions on the states and observables but no assumptions on the unknown quantum channel. Here, we consider the converse, by making a sparsity assumption on the channel but allowing for arbitrary states and observables.

\begin{lemma}\label{lemma:tm-entry-estimates-give-sparse-channel-expectation-estimates}
    Let $\rho\in\mathcal{S}(\mathbb{C}^{2^n})$ be an $n$-qubit quantum state and let $O\in\mathcal{B}(\mathbb{C}^{2^n})$ be a self-adjoint $n$-qubit observable.
    Let $\mathcal{N}:\mathcal{B}(\mathbb{C}^{2^n})\to \mathcal{B}(\mathbb{C}^{2^n})$ be an $n$-qubit quantum channel with an $s_{\mathcal{N}}$-sparse transfer matrix w.r.t.~$\mathcal{Q}$. 
    Assume that we have $\varepsilon$-accurate estimates $\hat{r}_{i,j}^\mathcal{Q}$ for the non-zero entries of the $\mathcal{Q}$-TM of $\mathcal{N}$.
    Then, from these estimates we can obtain an $\left(\varepsilon s_{\mathcal{N}}\norm{O}_2 \max_{1\leq i\leq 4^n} \lVert Q_i\rVert_\infty\right)$-accurate estimate of $\tr [O\mathcal{N}(\rho)]$.
\end{lemma}
\begin{proof} 
    Let us denote the pairs of indices corresponding to non-zero entries of the $\mathcal{Q}$-TM of $\mathcal{N}$ by $(i_1,j_1),\ldots, (i_{s_{\mathcal{N}}},j_{s_{\mathcal{N}}})\in \{1,\ldots,4^n\}^2$. 
    Expanding in the $\mathcal{Q}$-basis, we again start from \Cref{eq:expectation-pauli-expansion,eq:expectation-pauli-expansion-estimate}.
    We can make use of our sparsity assumption to reduce the number of summands and then again apply H\"older to obtain:
    \begin{align}
        \lvert \hat{\mu} - \Tr [O\mathcal{N}(\rho)]\rvert 
        &= \left\lvert \sum_{i,j=1}^{4^n}\alpha_i\beta_j (\hat{r}_{j,i}^\mathcal{Q} - (R_\mathcal{N}^\mathcal{Q})_{j,i}) \right\rvert\\
        &= \left\lvert \sum_{s=1}^{s_{\mathcal{N}}}\alpha_{j_s}\beta_{i_s} \left(\hat{r}^\mathcal{Q}_{i_s,j_s} - (R_\mathcal{N}^\mathcal{Q})_{i_s,j_s}\right) \right\rvert\\
    	&\leq \varepsilon\cdot \sum_{s=1}^{s_{\mathcal{N}}}\lvert\alpha_{j_s}\rvert \cdot\lvert\beta_{i_s}\rvert \, ,
    \end{align}
    where we used our assumption that $\lvert \hat{r}^\mathcal{Q}_{i_s,j_s} - (R_\mathcal{N}^\mathcal{Q})_{i_s,j_s}\rvert\leq\varepsilon$ for all $1\leq s\leq s_{\mathcal{N}}$.
    Yet another application of H\"older gives us
    \begin{align}
         \sum_{s=1}^{s_{\mathcal{N}}}\lvert\alpha_{j_s}\rvert \cdot\lvert\beta_{i_s}\rvert
         &\leq \left(\sum_{s=1}^{s_{\mathcal{N}}}\lvert\alpha_{j_s}\rvert\right)\cdot \max_{1\leq s\leq s_{\mathcal{N}}}\lvert\beta_{i_s}\rvert\\
         &\leq s_{\mathcal{N}}\lVert\vec{\alpha}\rVert_\infty \cdot \lVert\vec{\beta}\rVert_\infty \\
         &\leq s_{\mathcal{N}}\lVert\vec{\alpha}\rVert_\infty \cdot \lVert\vec{\beta}\rVert_2 \\
         &\leq s_{\mathcal{N}}\cdot \left(\max_{1\leq i\leq 4^n} \lVert Q_i\rVert_\infty\right)\cdot \lVert O\rVert_2\, ,
    \end{align}
    where the final step uses the inequality $\lVert\vec{\alpha}\rVert_\infty\leq \max_{1\leq i\leq 4^n} \lVert Q_i\rVert_\infty$ and the equality $\lVert\vec{\beta}\rVert_2 = \lVert O\rVert_2$, both of which were established in the proof of \Cref{corollary:tm-estimates-give-sparse-expectation-estimates}.
    Combining the two chains of inequalities gives the claimed approximation guarantee.
\end{proof}

We emphasize that applying \Cref{lemma:tm-entry-estimates-give-sparse-channel-expectation-estimates} requires estimates for the non-zero TM entries as well as the knowledge that all other TM entries are equal to zero.
Also, similarly to how we could allow for approximate sparsity in the states and observables in the previous subsection, it is straightforward to extend \Cref{lemma:tm-entry-estimates-give-sparse-channel-expectation-estimates} to channels that are close to a channel with sparse TM in $1$-to-$1$ norm. We conclude this subsection with some examples of quantum channels with a sparse (Pauli) transfer matrix.

\begin{example}[Pauli-sparse channels]\label{example:sparse-channels}
    The following constitute examples for channels with a sparse Pauli transfer matrix:
    \begin{itemize}
        \item Let $\mathcal{N}(\cdot) = \sum_{C\in\{0,1,2,3\}^n} \gamma_C \sigma_C (\cdot )\sigma_C$ be an $n$-qubit Pauli channel with Pauli error rates satisfying $\gamma_C\geq 0$ and $\sum_{C\in\{0,1,2,3\}^n} \gamma_C = 1$. 
        Then, for any $A,B\in\{0,1,2,3\}^n$, we have 
        \begin{align}
            \frac{1}{2^n}\tr[\sigma_A \mathcal{N}(\sigma_B)]
            &= \frac{1}{2^n}\sum_{C\in\{0,1,2,3\}^n} \gamma_C \tr[\sigma_A \sigma_C \sigma_B \sigma_C]\\
            &= \left( \sum_{\substack{C\in\{0,1,2,3\}^n\\ [\sigma_B,\sigma_C]=0}} \gamma_C - \sum_{\substack{D\in\{0,1,2,3\}^n\\ \{\sigma_B,\sigma_D\}=0}} \gamma_D\right) \delta_{A, B} .
        \end{align}
        That is, for any $B\in\{0,1,2,3\}^n$, the normalized Pauli $\tfrac{\sigma_B}{\sqrt{2^n}}$ is an eigenoperator of $\mathcal{N}$ with associated Pauli eigenvalue given by 
        \begin{equation}
            \lambda_B = \sum_{\substack{C\in\{0,1,2,3\}^n\\ [\sigma_B,\sigma_C]=0}} \gamma_C - \sum_{\substack{D\in\{0,1,2,3\}^n\\ \{\sigma_B,\sigma_D\}=0}} \gamma_D .
        \end{equation}
        In other words, the PTM of a Pauli channel is a diagonal matrix with its eigenvalues on the diagonal. Therefore, for Pauli channels, PTM learning is equivalent to the task of Pauli eigenvalue learning considered in \cite{chen2022quantum}.
        Accordingly, a Pauli channel has an $s$-sparse PTM if and only if its spectrum contains (at most) $s$ non-zero eigenvalues.
        \item An $n$-qubit channel $\mathcal{N}$ has an $s$-sparse PTM iff its Choi state $\tfrac{1}{2^n}\Gamma^{\mathcal{N}}$ has an $s$-sparse Pauli ONB expansion. (This can be seen as a direct consequence of \Cref{eq:ptm-as-pauli-choi-expectation-value} below.)
        Thus, any example of a Pauli-sparse $2n$-qubit state $\rho$ with $\tr_{n+1,\ldots,2n}[\rho]=\tfrac{1}{2^n}\mathbbm{1}_{2^n}$ gives rise to quantum channel with $s$-sparse PTM. 
        Concretely, if the Pauli expansion of the Choi state is $\tfrac{1}{2^n}\Gamma^{\mathcal{N}} = \sum_{C\in\{0,1,2,3\}^{2n}} \gamma(C)\tfrac{\sigma_C}{\sqrt{2^{2n}}}$, with $\gamma(B 0^n)=\tfrac{1}{\sqrt{2^n}}\delta_{B, 0^n}$ to ensure that $\mathcal{N}$ is trace-preserving, then the corresponding channel $\mathcal{N}$ acts on an input state $\rho = \sum_{A\in\{0,1,2,3\}^n} \alpha(A) \tfrac{\sigma_A}{\sqrt{2^n}}$ via
        \begin{equation}
            \mathcal{N}(\rho)
            = 2^n \sum_{A,B\in\{0,1,2,3\}^n} (-1)^{\lvert\{1\leq i\leq n~|~A_i=2\}\rvert} \alpha(A) \gamma(AB)  \frac{\sigma_B}{\sqrt{2^n}}\, .
        \end{equation}
        \item To make the abstract construction of the previous bullet point more concrete, consider the Choi state $\tfrac{1}{2^n}\Gamma^{\mathcal{N}} = \tfrac{\mathbbm{1}_{2^{2n}} + \sigma_A}{2^{2n}}$, for some $A\in\{0,1,2,3\}^{2n}$ with $A_{n+1}\ldots A_{2n}\neq 0^n$. 
        The action of the associated $n$-qubit quantum channel is given by
        \begin{equation}
            \mathcal{N}(\cdot)
            = 2^n \tr_{\mathrm{in}}\left[((\cdot)^{\top}_{\mathrm{in}}\otimes\mathbbm{1}_{\mathrm{out}})\frac{1}{2^n} \Gamma^{\mathcal{N}}\right]
            = \frac{\tr[(\cdot)] \mathbbm{1}_{2^n} + s_A \tr[(\cdot)\sigma_{A_{1}\ldots A_n}] \sigma_{A_{n+1}\ldots A_{2n}}}{2^n} ,
        \end{equation}
        where the sign $s_A$ is given by $s_A= (-1)^{\lvert\{1\leq i\leq n~|~A_i=2\}\rvert}$.
        As $\tfrac{1}{2^n}\Gamma^{\mathcal{N}}$ has a $2$-sparse Pauli expansion, this channel $\mathcal{N}$ has a $2$-sparse PTM.
    \end{itemize}
\end{example}

We conclude this section with a remark about limitations of the class of sparse channels:
    
\begin{remark}\label{remark:sparse-channel-limited-rank}
    If $\mathcal{N}$ has an $s_\mathcal{N}$-sparse $\mathcal{Q}$-TM, then that directly implies that this TM has rank at most $s_\mathcal{N}$. As the $\mathcal{Q}$-TM of $\mathcal{N}$ is simply its matrix representation w.r.t.~the ONB $\mathcal{Q}$, that tells us: An $n$-qubit quantum channel with $s_\mathcal{N}$-sparse TM has rank at most $s_\mathcal{N}$ when viewed as a linear map from $\mathcal{B}((\mathbb{C}^2)^{\otimes n})$ to itself.
    Thus, our notion of sparse channels can only capture relatively low-rank channels, i.e., channels with a relatively low-dimensional image space.
    
    We can also consider the Choi/Kraus rank of $\mathcal{N}$, instead of its rank as a linear map. For simplicity, let us focus on the case of the Pauli ONB.
    As mentioned above, by \Cref{eq:ptm-as-pauli-choi-expectation-value},  $\mathcal{N}$ having an $s$-sparse PTM is equivalent to its Choi state $\tfrac{1}{2^n}\Gamma^{\mathcal{N}}$ having an $s$-sparse Pauli expansion. Following the reasoning in \Cref{remark:sparse-states-limited-purity}, we thus see that an $s$-sparse channel $\mathcal{N}$ has a Choi/Kraus rank of at least $\Omega (\tfrac{4^n}{s})$.
    For general ONBs and TMs, we can obtain a Choi/Kraus rank lower bound with a similar reasoning, replacing \Cref{eq:ptm-as-pauli-choi-expectation-value} by \Cref{eq:general-tm-via-choi-expectation}.
\end{remark}

\section{Efficient Pauli Transfer Matrix Learning With Quantum Memory}\label{section:ptm-learning}

In this section, we show that a learner with access to a quantum memory can efficiently estimate all the PTM entries of an unknown channels from Choi access. Combining this with the results of \Cref{section:tm-estimates-give-expectation-estimates}, this gives a procedure for estimating Pauli-sparse expectation values of an unknown channel. Alternatively, we can use it to estimate arbitrary expectation values for an unknown channel promised to have a sparse PTM.

First, we show that the Pauli-specific shadow tomography result of \cite{huang2021information-theoretic} allows us to learn the PTM of an arbitrary unknown channel from a small number of Choi state copies, assuming access to a quantum memory:

\begin{theorem}\label{theorem:ptm-entries-from-choi-shadow-tomography}
    There is a learning procedure with quantum memory from Choi access that, using $m=\mathcal{O}\left(\tfrac{n + \log(1/\delta)}{\varepsilon^4}\right)$ copies of the Choi state $\tfrac{1}{2^n}\Gamma^{\mathcal{N}}$ of an arbitrary unknown $n$-qubit channel $\mathcal{N}$, outputs, with success probability $\geq 1-\delta$, numbers $\hat{r}_{A,B}^\mathcal{P}$ for $A,B\in\{0,1,2,3\}^n$ such that $\lvert  \hat{r}_{A,B}^\mathcal{P} - (R_\mathcal{N}^\mathcal{P})_{A,B}\rvert\leq\varepsilon$ holds simultaneously for all $A,B\in\{0,1,2,3\}^n$.
\end{theorem}
\begin{proof}
    For every $A,B\in\{0,1,2,3\}^n$, we can write the corresponding PTM entry as
    \small
    \begin{equation}\label{eq:ptm-as-pauli-choi-expectation-value}
        \left( R_\mathcal{N}^\mathcal{P}\right)_{A,B}
        = \frac{1}{2^n}\tr [\sigma_A \mathcal{N}(\sigma_B)]
        = \tr[(\sigma_B^{\top} \otimes \sigma_A) \frac{1}{2^n}\Gamma^\mathcal{N}]
        = (-1)^{\lvert\{1\leq i\leq n| B_i=2\}\rvert} \tr[(\sigma_B \otimes \sigma_A) \frac{1}{2^n}\Gamma^\mathcal{N}] \, .
    \end{equation}
    \normalsize
    Here, the last equality used that $\sigma_0^{\top}=\sigma_0$, $\sigma_1^{\top}=\sigma_1$, $\sigma_2^{\top}=- \sigma_2$, and $\sigma_3^{\top}=\sigma_3$.
    Thus, an $\varepsilon$-accurate estimate of the Pauli expectation value $\tr[(\sigma_B \otimes \sigma_A) \tfrac{1}{2^n}\Gamma^\mathcal{N}]$ immediately translates to a $\varepsilon$-accurate estimate of the PTM entry $\left( R_\mathcal{N}^\mathcal{P}\right)_{A,B}$.
    
    To learn all these $M=16^{n}$ Pauli expectation values simultaneously, we employ the Pauli shadow tomography procedure of \cite[Supplementary Material V.B]{huang2021information-theoretic}.
    According to \cite[Theorem 2]{huang2021information-theoretic}, this procedure uses at most
    \begin{equation}
        \mathcal{O}\left(\frac{\log\left(\tfrac{16^n}{\delta}\right)}{\varepsilon^4}\right)
        = \mathcal{O}\left(\frac{n + \log\left(\tfrac{1}{\delta}\right)}{\varepsilon^4}\right)
    \end{equation}
    copies of the Choi state $\tfrac{1}{2^n}\Gamma^\mathcal{N}$ to produce, with success probability $\geq 1-\delta$, simultaneously $\varepsilon$-accurate estimates for all Pauli expectation values.
\end{proof}

In \Cref{appendix:query-lower-bound-ptm-learning-with-quantum-memory-choi-access}, we give a simple proof for a corresponding lower bound of $\Omega (\tfrac{n}{\varepsilon^2})$, thus establishing the optimality of the linear-in-$n$ scaling for learning the PTM with quantum memory from Choi access.
In fact, \cite[Theorem 2, (iv)]{chen2022quantum} uses teleportation stretching for Pauli channels to show: Even given general channel access and even when the unknown channel is promised to be a Pauli channel and thus to have a diagonal PTM, the linear-in-$n$ scaling cannot be improved.

When focusing on the task of predicting sparse expectation values, we not only have efficient query complexity, we can also translate guarantees on the efficiency of the Pauli shadow tomography result of \cite{huang2021information-theoretic} in terms of both classical memory and classical post-processing time to our setting.

\begin{corollary}\label{corollary:pauli-sparse-expectations-from-choi-shadow-tomography}
    There is a learning procedure with quantum memory from Choi access that, using $m_1=\mathcal{O}\left(\tfrac{n + \log(1/\delta)}{\varepsilon^4}\cdot B^4 s_\rho^4 s_O^2\right)$ copies of the Choi state $\tfrac{1}{2^n}\Gamma^{\mathcal{N}}$ of an unknown $n$-qubit channel $\mathcal{N}$, produces a classical description $\hat{\mathcal{N}}$ of $\mathcal{N}$ from which any $M$ expectation values of the form $\operatorname{tr}[O_i\mathcal{N}(\rho_i)]$, $1\leq i\leq M$, with $O_i$ some $n$-qubit observables with $s_O$-sparse Pauli basis expansions satisfying $\lVert O_i\rVert_2\leq B\sqrt{2^n}$ and $\rho_i$ some $n$-qubit states with $s_\rho$-sparse Pauli basis expansions, can be predicted to accuracy $\varepsilon$ with $m_2 = \mathcal{O}\left(\tfrac{ \log (\min \{M s_\rho s_O, 16^n\}) + \log (1 /\delta)}{\varepsilon^2}\cdot B^2 s_\rho^2 s_O\right)$ additional copies of $\tfrac{1}{2^n}\Gamma^{\mathcal{N}}$ and classical computation time $\mathcal{O}\left(n m_1 M\right)=\mathcal{O}\left(\tfrac{n^2 + n\log (1/\delta)}{\varepsilon^4 }\cdot B^4 s_\rho^4 s_O^2 M\right)$, with success probability $\geq 1-\delta$.
    The classical representation $\hat{\mathcal{N}}$ of $\mathcal{N}$ consists of $\mathcal{O} (n m_1)=\mathcal{O}\left(\tfrac{n^2 + n \log(1/\delta)}{\varepsilon^4}\cdot B^4 s_\rho^4 s_O^2\right)$ real numbers stored in classical memory.
\end{corollary}
\begin{proof}
    This follows from \cite[Theorem 2]{huang2021information-theoretic} and the bounds on classical memory and classical computation time given in the proof thereof (see \cite[Supplemental material V.B]{huang2021information-theoretic}) when combined with \Cref{corollary:tm-estimates-give-sparse-expectation-estimates}.
    To see this, note that, by \Cref{corollary:tm-estimates-give-sparse-expectation-estimates}, it suffices to use the protocol of \cite[Supplemental material V.B]{huang2021information-theoretic} to predict the at most $\min\{M s_\rho s_O, 16^n\}$ relevant non-zero PTM entries, which we can relate to expectation values of Pauli observables on the Choi state by \Cref{eq:ptm-as-pauli-choi-expectation-value}, each to accuracy $\tilde{\varepsilon}=\tfrac{\varepsilon}{s_\rho \sqrt{s_O} B}$.
\end{proof}

When talking about computational efficiency in \Cref{corollary:pauli-sparse-expectations-from-choi-shadow-tomography}, we implicitly assume that the Pauli-sparse observables $O_i$ and the Pauli-sparse states $\rho_i$ are given to the learner in the natural efficient classical representations, $\hat{O}_i$ and $\hat{\rho}_i$, namely as lists of the respective Pauli coefficients. For example, given the freedom in normalization discussed in \Cref{remark:alternative-normalization}, we can work with the classical representations $\hat{\rho}_i = \{(A_s, \tr[\rho_i \sigma_{A_s}])\}_{s=1}^{s_\rho}$ and $\hat{O}_i = \{(B_s, \tfrac{1}{2^n}\tr[O_i \sigma_{B_s}])\}_{s=1}^{s_O}$.
Given these Pauli representations of $O_i$ and $\rho_i$, our learning procedure can then efficiently predict $\operatorname{tr}[O_i\mathcal{N}(\rho_i)]$.
\Cref{corollary:pauli-sparse-expectations-from-choi-shadow-tomography} thus tells us that we can predict expectation values of the form $\operatorname{tr}[O\mathcal{N}(\rho)]$ to an inverse-polynomial accuracy from an efficient number of copies of $\tfrac{1}{2^n}\Gamma^{\mathcal{N}}$ using efficient classical computation time and efficient classical memory, as long as we assume that both $O$ and $\rho$ have a $\mathcal{O}(\poly (n))$-sparse Pauli basis expansions and that $\norm{O}\leq \mathcal{O}(\poly (n))$.

While we formulate \Cref{corollary:pauli-sparse-expectations-from-choi-shadow-tomography} for sparse observables and states, we note that a similar copy complexity bound holds when replacing sparsity by $1$-norm assumptions. This can be proven in the same way, merely replacing \Cref{corollary:tm-estimates-give-sparse-expectation-estimates} by \Cref{lemma:tm-entry-estimates-give-expectation-estimates-1-norm}. As we saw in \Cref{remark:improved-1-norm-bound-from-huang-et-al}, the latter can lead to tighter bounds than the former.
However, when starting from $1$-norm assumptions, it is not immediately clear how to translate the statements about computational efficiency, for example because there is no longer a natural efficient representation for the observables and states.

\begin{remark}
    We note some additional features of \Cref{corollary:pauli-sparse-expectations-from-choi-shadow-tomography}.
    First, if the classical representations are only approximate, in the sense that we obtain $\hat{\rho}_i = \{(A_s, \hat{\alpha}_s^{(i)})\}_{s=1}^{s_\rho}$ with $\lvert \hat{\alpha}_s^{(i)} - \tr[\rho_i \sigma_{A_s}]\rvert\leq \varepsilon_\rho$ and $\hat{O}_i = \{(B_s, \hat{\beta}_s^{(i)})\}_{s=1}^{s_O}$ with $\lvert \hat{\beta}_s^{(i)} - \tfrac{1}{2^n}\tr[O_i \sigma_{B_s}]\rvert\leq\varepsilon_O$, then our procedure still gives $(\varepsilon + s_\rho s_O \varepsilon_O + s_\rho \varepsilon_\rho (\sqrt{s_O}B + s_O \varepsilon_O))$-accurate estimates for $\operatorname{tr}[O_i\mathcal{N}(\rho_i)]$. (This can be seen by augmenting the proof of \Cref{lemma:tm-entry-estimates-give-sparse-channel-expectation-estimates} with additional triangle inequalities.)
    
    Second, we inherit two more interesting features of \cite[Theorem 2]{huang2021information-theoretic}:
    On the one hand, in \Cref{corollary:pauli-sparse-expectations-from-choi-shadow-tomography}, the $O_i$ and $\rho_i$ need not be known in advance, we can build our classical representation and get told only afterwards which expectation values we should predict.
    On the other hand, the part of the procedure in \Cref{corollary:pauli-sparse-expectations-from-choi-shadow-tomography} that produces the classical representation $\hat{\mathcal{N}}$ only ever performs measurements on at most two copies of the Choi state $\tfrac{1}{2^n}\Gamma^{\mathcal{N}}$, so a quantum memory of $2n$ qubits suffices to build $\hat{\mathcal{N}}$. However, using $\hat{\mathcal{N}}$ to predict $M$ expectation values uses a quantum memory of $2n(m_2-1)$ qubits, with $m_2$ as in \Cref{corollary:pauli-sparse-expectations-from-choi-shadow-tomography}.
\end{remark}

If we are promised in advance that the unknown quantum channel has a sparse PTM but the location of the non-zero entries is unknown, then a combination of \Cref{theorem:ptm-entries-from-choi-shadow-tomography} with a cut-off argument leads to the following variant of our result:

\begin{lemma}\label{lemma:sparse-channel-ptm-entries-from-choi-shadow-tomography}
    There is a learning procedure with quantum memory from Choi access that, using $m=\mathcal{O}\left(\tfrac{n + \log(1/\delta)}{\varepsilon^4}\right)$ copies of the Choi state $\tfrac{1}{2^n}\Gamma^{\mathcal{N}}$ of an unknown $n$-qubit channel $\mathcal{N}$ with $s_\mathcal{N}$-sparse PTM, outputs, with success probability $\geq 1-\delta$, a list $\{((i_{s},j_{s}), \hat{r}_{i_s,j_s}^\mathcal{P})\}_{s=1}^{S}$ of length $S$ such that
    \begin{enumerate}
        \item $(i,j)\not\in \{(i_{s},j_{s})\}_{s=1}^{S}$ $\Rightarrow$ $\lvert (R_\mathcal{N}^\mathcal{P})_{i,j}\rvert\leq\varepsilon$,
        \item $(i,j)\in \{(i_{s},j_{s})\}_{s=1}^{S}$ $\Rightarrow$ $\lvert (R_\mathcal{N}^\mathcal{P})_{i,j}\rvert > \nicefrac{\varepsilon}{6}$,
        \item $S\leq s_\mathcal{N}$, and
        \item $\lvert  \hat{r}_{i_s,j_s}^\mathcal{P} - (R_\mathcal{N}^\mathcal{P})_{i_s,j_s}\rvert\leq\varepsilon$ holds for all $1\leq s \leq S$.
    \end{enumerate}
\end{lemma}
\begin{proof}
    First apply the learning procedure of \Cref{theorem:ptm-entries-from-choi-shadow-tomography} with accuracy parameter $\tilde{\varepsilon} = \nicefrac{\varepsilon}{3}$.
    With success probability $\geq 1-\delta$, this produces estimates $\tilde{r}_{i,j}^\mathcal{P}$ satisfying $\lvert  \tilde{r}_{i,j}^\mathcal{P} - (R_\mathcal{N}^\mathcal{P})_{i,j}\rvert\leq\nicefrac{\varepsilon}{3}$ simultaneously for all $1\leq i,j\leq 4^n$.
    Now define 
    \begin{equation}
        \hat{r}_{i,j}^\mathcal{P}
        = \begin{cases} 0 \quad &\textrm{if } \lvert \tilde{r}_{i,j}^\mathcal{P}\rvert\leq\nicefrac{\varepsilon}{2}\\ \tilde{r}_{i,j}^\mathcal{P} &\textrm{else}\end{cases} .
    \end{equation}
    Let $\{((i_{s},j_{s}), \hat{r}_{i_s,j_s}^\mathcal{P})\}_{s=1}^{S}$ be the list of remaining non-zero numbers $\hat{r}_{i,j}^\mathcal{P}$ and the associated locations $(i,j)$ in the transfer matrix.
    It is now easy to see that, conditioned on the success event, this list has the desired properties:
    \begin{enumerate}
        \item Assume $(i,j)\not\in \{(i_{s},j_{s})\}_{s=1}^{S}$. By construction, this is the case iff $\hat{r}_{i,j}^\mathcal{P} =0$, which in turn is equivalent to $\lvert \tilde{r}_{i,j}^\mathcal{P}\rvert\leq\nicefrac{\varepsilon}{2}$. As $\lvert  \tilde{r}_{i,j}^\mathcal{P} - (R_\mathcal{N}^\mathcal{P})_{i,j}\rvert\leq\nicefrac{\varepsilon}{3}$ holds in the case of success, this implies $\lvert (R_\mathcal{N}^\mathcal{P})_{i,j}\rvert\leq \nicefrac{\varepsilon}{2} + \nicefrac{\varepsilon}{3}= \nicefrac{5\varepsilon}{6}\leq\varepsilon$ by the triangle inequality.
        \item Assume $(i,j)\in \{(i_{s},j_{s})\}_{s=1}^{S}$. By construction, this is the case iff $\hat{r}_{i,j}^\mathcal{P} \neq 0$ and $\hat{r}_{i,j}^\mathcal{P} = \tilde{r}_{i,j}^\mathcal{P}$, which in turn is equivalent to $\lvert \tilde{r}_{i,j}^\mathcal{P}\rvert >\nicefrac{\varepsilon}{2}$. As $\lvert  \tilde{r}_{i,j}^\mathcal{P} - (R_\mathcal{N}^\mathcal{P})_{i,j}\rvert\leq\nicefrac{\varepsilon}{3}$ holds in the case of success, this implies $\lvert (R_\mathcal{N}^\mathcal{P})_{i,j}\rvert > \nicefrac{\varepsilon}{2} - \nicefrac{\varepsilon}{3}= \nicefrac{\varepsilon}{6}$.
        \item As a direct consequence of 2., we see that $S\leq \{(i,j)~|~ (R_\mathcal{N}^\mathcal{P})_{i,j}\neq 0\} = s_\mathcal{N}$.
        \item For any $1\leq s\leq S$, we have $\hat{r}_{i_s,j_s}^\mathcal{P} = \tilde{r}_{i_s,j_s}^\mathcal{P}$ by construction. As $\lvert  \tilde{r}_{i_s,j_s}^\mathcal{P} - (R_\mathcal{N}^\mathcal{P})_{i_s,j_s}\rvert\leq\nicefrac{\varepsilon}{3}$ holds in the case of success, we have $\lvert  \hat{r}_{i_s,j_s}^\mathcal{P} - (R_\mathcal{N}^\mathcal{P})_{i_s,j_s}\rvert\leq\nicefrac{\varepsilon}{3}\leq \varepsilon$.
    \end{enumerate}
\end{proof}

Combining \Cref{lemma:sparse-channel-ptm-entries-from-choi-shadow-tomography} with \Cref{lemma:tm-entry-estimates-give-sparse-channel-expectation-estimates} gives a variant of \Cref{corollary:pauli-sparse-expectations-from-choi-shadow-tomography} in which we replace the sparsity assumption on states and observables by a sparsity assumption on the unknown channel:

\begin{corollary}\label{corollary:sample-complexity-predicting-pauli-sparse-channel-expectation-values}
    There is a learning procedure with quantum memory from Choi access that, using
    \begin{equation}
        m
        =\mathcal{O}\left(\frac{n + \log(1/\delta)}{\varepsilon^4}\cdot B^4 s_\mathcal{N}^4 \right)
    \end{equation}
    copies of the Choi state $\tfrac{1}{2^n}\Gamma^{\mathcal{N}}$ of an unknown $n$-qubit channel $\mathcal{N}$ with $s_\mathcal{N}$-sparse PTM, produces, with success probability $\geq 1-\delta$, a classical description $\hat{\mathcal{N}}$ of $\mathcal{N}$ from which all expectation values of the form $\operatorname{tr}[O\mathcal{N}(\rho)]$, where $O$ is an arbitrary $n$-qubit observable with $\lVert O\rVert_2\leq B\sqrt{2^n}$ and $\rho$ is an arbitrary $n$-qubit state, can be predicted to accuracy $\varepsilon$.
\end{corollary}
\begin{proof}
    By \Cref{lemma:tm-entry-estimates-give-sparse-channel-expectation-estimates}, we can get the desired $\varepsilon$-accurate expectation value estimates from $\tilde{\varepsilon}$-accurate PTM entry estimates if we set $\tilde{\varepsilon}=\tfrac{\varepsilon \sqrt{2^n}}{s_\mathcal{N} \norm{O}_2 }$.
    If we plug this accuracy $\tilde{\varepsilon}$ into the copy complexity bound of \Cref{lemma:sparse-channel-ptm-entries-from-choi-shadow-tomography} and insert the assumed inequality $\lVert O\rVert_2\leq B\sqrt{2^n}$, we obtain the stated copy complexity bound.
\end{proof}

\Cref{corollary:sample-complexity-predicting-pauli-sparse-channel-expectation-values} describes an information-theoretically efficient way of predicting expectation values for arbitrary input states and arbitrary bounded output observables if the unknown quantum channel is promised to have a polynomially-sparse PTM. Importantly, the learning procedure in \Cref{corollary:sample-complexity-predicting-pauli-sparse-channel-expectation-values} does not require knowledge about the sparsity structure of the PTM. It suffices to know only that the TM is sparse, the sparsity structure is then found as part of the learning procedure.
In fact, if additionally the sparsity structure of the (P)TM is known in advance, then there is a straightforward information-theoretically efficient learning procedure without quantum memory, compare \Cref{remark:no-exponential-separation-with-known-sparsity-structure}.

\begin{remark}
    Analogously to \Cref{corollary:pauli-sparse-expectations-from-choi-shadow-tomography}, one can show that the classical representation $\mathcal{N}$ in \Cref{corollary:sample-complexity-predicting-pauli-sparse-channel-expectation-values} can be chosen to consist of $\mathcal{O}(\tfrac{n^2 + n \log(1/\delta)}{\varepsilon^4}\cdot B^4 s_\mathcal{N}^4 )$ real numbers.
    However, it is not clear how to translate the computation time bound from \Cref{corollary:pauli-sparse-expectations-from-choi-shadow-tomography} to the case of sparse channels and arbitrary states and observables.
    In particular, there is no longer a natural efficient representation of the states and observables for which the expectation value is to be predicted.
\end{remark}

While we focus on the PTM in this section, we can extend most of our results to TMs w.r.t.~more general unitary orthonormal bases. We explain this in \Cref{section:tm-learning}.

\section{Exponential Query Complexity Lower Bounds For Pauli Transfer Matrix Learning Without Quantum Memory}\label{section:query-complexity-lower-bounds}

In the language introduced in \Cref{section:different-learning-models}, we can view \Cref{section:ptm-learning} (and \Cref{section:tm-learning}) as giving us algorithms with quantum memory for learning the PTM (and TM) and for predicting sparse expectation values of an unknown quantum channel from Choi access with a polynomial number of copies.
That is, if we have access to a quantum memory, then we can information-theoretically efficiently solve these two learning tasks even in the weakest of the three quantum access models defined in \Cref{section:different-learning-models}.
In this section, we explore how essential the quantum memory is to achieve the efficient information-theoretic complexity scaling.

\subsection{Query Complexity Lower Bounds For Pauli Transfer Matrix Learning Without Quantum Memory From Choi Access}

We first prove exponential query complexity lower bounds for the weakest of our learning models, namely that of learning without quantum memory from Choi access. While the result of this subsection is implied by that of the next subsection on learning without quantum memory from general channel access, we have chosen to present it separately because the proof in the case of Choi access 
is technically less involved and can serve as preparation for the proof in the general channel access case.
More precisely, the proof for Choi access is a variant of the proof strategy used in \cite{chen2022exponential} to establish exponential query complexity lower bounds for learners without quantum memory for the task of Pauli shadow tomography of quantum states.
For learning without quantum memory from Choi access, we have the following query complexity lower bound:

\begin{lemma}\label{lemma:lower-bound-ptm-learning-without-quantum-memory-from-choi-state-copies}
    Any algorithm for learning without quantum memory from Choi access requires $\Omega (\tfrac{4^n}{\varepsilon^2})$ copies of the Choi state of an unknown $n$-qubit channel $\mathcal{N}$ to estimate all entries of $R_\mathcal{N}^\mathcal{P}$ up to accuracy $\varepsilon$ with success probability $\geq \tfrac{2}{3}$.
\end{lemma}
\begin{proof}[Proof sketch]
    We give a detailed proof in \Cref{appendix:proofs}. The proof is a modification to that of \cite[Corollary 5.9]{chen2021exponential-arxiv} and proceeds as follows:
    As learning PTM entries is, by \Cref{eq:ptm-as-pauli-choi-expectation-value}, equivalent to predicting Pauli expectation values for the Choi state, PTM learning becomes a Pauli shadow tomography problem on the level of the Choi state. 
    For general states, \cite[Corollary 5.9]{chen2021exponential-arxiv} established a Pauli shadow tomography query complexity lower bound for learners without quantum memory by reducing to a many-versus-one state distinguishing task and constructing an ensemble of states for which the latter task is hard. 
    To adapt their proof to our setting, we effectively restrict their state ensemble to the sub-ensemble of valid Choi states and analyze the hardness of the Pauli shadow tomography task for that sub-ensemble.
\end{proof}

We can recover the same asymptotic copy complexity lower bound even if we restrict the unknown channel to be doubly-stochastic, entanglement-breaking, and to have a sparse PTM with unknown sparsity structure:

\begin{lemma}\label{lemma:lower-bound-ptm-learning-without-quantum-memory-from-doubly-stochastic-choi-state-copies}
    Any algorithm for learning without quantum memory from Choi access requires $\Omega (\tfrac{4^n}{\varepsilon^2})$ copies of the Choi state of an unknown doubly-stochastic and entanglement-breaking $n$-qubit channel $\mathcal{N}$ with $\mathcal{O}(1)$-sparse PTM to estimate all entries of $R_\mathcal{N}^\mathcal{P}$ up to accuracy $\varepsilon$ with success probability $\geq \tfrac{2}{3}$.
\end{lemma}
\begin{proof}
    We give a detailed proof in \Cref{appendix:proofs}. The proof requires only a small modification compared to that of \Cref{lemma:lower-bound-ptm-learning-without-quantum-memory-from-choi-state-copies}. Namely, we have to further restrict the state ensemble to those that are valid Choi states of doubly stochastic quantum channels. All states in the ensemble are already separable and $\mathcal{O}(1)$-sparse, so we automatically focus on entanglement-breaking channels with $\mathcal{O}(1)$-sparse PTM.
    Again, we have to perform a hardness analysis for the Pauli shadow tomography task of this further restricted ensemble.
\end{proof}

If $\mathcal{N}$ is doubly-stochastic, then we have for $A,B\in\{0,1,2,3\}^n$:
\begin{equation}
    \tr\left[\sigma_A \mathcal{N}\left(\frac{\mathbbm{1}_2^{\otimes n} + \sigma_B}{2^n}\right)\right]
    = \begin{cases} 1 + \left( R_\mathcal{N}^\mathcal{P}\right)_{A,B} \quad &\textrm{if } A = 0^n\\ \left( R_\mathcal{N}^\mathcal{P}\right)_{A,B} &\textrm{else} \end{cases} .
\end{equation}
Thus, the ability to estimate expectation values of the form $\tr\left[\sigma_A \mathcal{N}\left(\tfrac{\mathbbm{1}_2^{\otimes n} + \sigma_B}{2^n}\right)\right]$ up to a desired accuracy $\varepsilon$ immediately allows one to estimate the PTM entries $\left( R_\mathcal{N}^\mathcal{P}\right)_{A,B}$ to the same accuracy, assuming $\mathcal{N}$ is doubly-stochastic.
Trivially, the observables $\sigma_A$ have a $1$-sparse Pauli basis expansion, and the states $\tfrac{\mathbbm{1}_2^{\otimes n} + \sigma_B}{2^n}$ have a $2$-sparse Pauli basis expansion.
Also, $\lVert \sigma_A\rVert_2 = \sqrt{2^n}$.
Therefore, \Cref{lemma:lower-bound-ptm-learning-without-quantum-memory-from-doubly-stochastic-choi-state-copies} has the following consequence:

\begin{corollary}\label{corollary:lower-bound-sparse-expectation-values-without-quantum-memory-from-doubly-stochastic-choi-state-copies}
    Any learning algorithm without quantum memory requires $\Omega (\tfrac{4^n}{\varepsilon^2})$ copies of the Choi state of an unknown doubly-stochastic and entanglement-breaking $n$-qubit channel $\mathcal{N}$ with $\mathcal{O}(1)$-sparse PTM to produce a classical description $\hat{\mathcal{N}}$ of $\mathcal{N}$ from which any expectation value of the form $\operatorname{tr}[O\mathcal{N}(\rho)]$, with $O$ an $n$-qubit observable with $\mathcal{O}(1)$-sparse Pauli basis expansion satisfying $\lVert O\rVert_2\leq \sqrt{2^n}$ and $\rho$ an $n$-qubit state with $\mathcal{O}(1)$-sparse Pauli basis expansion, can be predicted to accuracy $\varepsilon$ with success probability $\geq \tfrac{2}{3}$.
\end{corollary}

This is to be contrasted with \Cref{corollary:pauli-sparse-expectations-from-choi-shadow-tomography,corollary:sample-complexity-predicting-pauli-sparse-channel-expectation-values} (as well as \Cref{corollary:sample-complexity-predicting-sparse-expectation-values,corollary:sample-complexity-predicting-sparse-channel-expectation-values}), which showed that using a quantum memory allows to solve even more general versions of the same expectation value prediction task using polynomially-in-$n$ or even linearly-in-$n$ many copies of the Choi state of an unknown $n$-qubit quantum channel.
Thus, we have established an exponential separation in query complexity between algorithms with and without quantum memory for the task of predicting Pauli-sparse expectation values of an unknown channel with sparse PTM when given access to copies of its Choi state.

\begin{remark}\label{remark:no-exponential-separation-with-known-sparsity-structure}
    The exponential query complexity lower bound for learners without quantum memory from Choi access may no longer apply if stronger assumptions on the unknown quantum channel are made.
    For example, if $\mathcal{N}$ is promised to be $s_\mathcal{N}$-sparse and the sparsity structure of $\mathcal{N}$ (i.e., the positions of non-zero PTM entries) is known in advance, then, by \Cref{lemma:tm-entry-estimates-give-sparse-channel-expectation-estimates}, it suffices to estimate the corresponding $s_\mathcal{N}$ many Pauli expectation values for the Choi state. 
    Thus, if $s_{\mathcal{N}}\leq\mathcal{O}(\poly (n))$, then this can easily be achieved without quantum memory using $\mathcal{O}(\tfrac{\poly (n) \cdot\log (1/\delta)}{\varepsilon^2})$ Choi state copies.
\end{remark}

\begin{remark}\label{remark:no-exponential-separation-with-local-states-and-observables}
    Complementary to the previous remark, we note that the exponential query complexity lower bound for learners without quantum memory from Choi access may also fail to hold if stronger assumptions on the input states and output observables are made.
    For instance, if we only care about ``low-weight'' observables and states, whose Pauli expansions consist of Paulis acting non-trivially on at most a constant number of qubits, then applying classical shadows to the Choi state suffices, so no quantum memory is needed \cite{levy2021classical-shadows, kunjummen2023shadow-process}.
\end{remark}

\subsection{Query Complexity Lower Bounds For Pauli Transfer Matrix Learning Without Quantum Memory From Sequential Channel Access}

Next, we extend the exponential query complexity lower bounds for the task of PTM learning without quantum memory also to the stronger model of sequential channel access (\Cref{definition:qchannel-learning-without-qmemory}). 

\begin{theorem}\label{theorem:lower-bound-ptm-learning-without-quantum-memory-from-channel-queries}
    Any learning algorithm without quantum memory requires $\Omega (\tfrac{4^n}{\varepsilon^2})$ copies of an unknown $n$-qubit channel $\mathcal{N}$ to estimate all entries of $R_\mathcal{N}^\mathcal{P}$ up to accuracy $\varepsilon$ with success probability $\geq \tfrac{2}{3}$.
\end{theorem}
\begin{proof}
    We give a detailed proof in \Cref{appendix:proofs}. The proof is based on the learning tree representation of \cite{chen2022exponential} and an extension of their analysis from state shadow tomography to our channel shadow tomography setting.
    While the underlying framework is the same as in the proof of the Choi state access case, the technical details become more involved due to the possibility of adaptively chosen input states.
\end{proof}

Similarly to the previous subsection, we also have a stronger doubly-stochastic, entanglement-breaking, and sparse version:

\begin{theorem}\label{theorem:lower-bound-ptm-learning-without-quantum-memory-from-doubly-stochastic-channel-queries}
    Any learning algorithm without quantum memory requires $\Omega (\tfrac{4^n}{\varepsilon^2})$ copies of an unknown doubly-stochastic and entanglement-breaking $n$-qubit channel $\mathcal{N}$ with $\mathcal{O}(1)$-sparse PTM to estimate all entries of $R_\mathcal{N}^\mathcal{P}$ up to accuracy $\varepsilon$ with success probability $\geq \tfrac{2}{3}$.
\end{theorem}
\begin{proof}
    We give a detailed proof in \Cref{appendix:proofs}. This proof requires only a small modifcation to that of \Cref{theorem:lower-bound-ptm-learning-without-quantum-memory-from-channel-queries}, which again consists in further restricting the ensemble of unknown channels to a sub-ensemble consisting only of doubly-stochastic entanglement-breaking channels with $\mathcal{O}(1)$-sparse PTM. This further restriction then requires a corresponding adaptation to the hardness analysis.
\end{proof}

In exactly the same way that \Cref{lemma:lower-bound-ptm-learning-without-quantum-memory-from-doubly-stochastic-choi-state-copies} gave rise to \Cref{corollary:lower-bound-sparse-expectation-values-without-quantum-memory-from-doubly-stochastic-choi-state-copies}, \Cref{theorem:lower-bound-ptm-learning-without-quantum-memory-from-doubly-stochastic-channel-queries} has the following immediate consequence:

\begin{corollary}\label{corollary:lower-bound-sparse-expectation-values-without-quantum-memory-from-doubly-stochastic-channel-queries}
    Any learning algorithm without quantum memory requires $\Omega (\tfrac{4^n}{\varepsilon^2})$ copies of an unknown doubly-stochastic and entanglement-breaking $n$-qubit channel $\mathcal{N}$ with $\mathcal{O}(1)$-sparse PTM to produce a classical description $\hat{\mathcal{N}}$ of $\mathcal{N}$ from which any expectation value of the form $\operatorname{tr}[O\mathcal{N}(\rho)]$, with $O$ an $n$-qubit observable with $\mathcal{O}(1)$-sparse Pauli basis expansion satisfying $\lVert O\rVert_2\leq \sqrt{2^n}$ and $\rho$ an $n$-qubit state with $\mathcal{O}(1)$-sparse Pauli basis expansion, can be predicted to accuracy $\varepsilon$ with success probability $\geq \tfrac{2}{3}$.
\end{corollary}

This shows that the exponential query complexity separation for the task of predicting Pauli-sparse expectation values between learning with quantum memory from Choi state copies and learning without quantum memory from Choi state copies persists even if we allow learners without quantum memory to access copies of the unknown sparse channel. Again, a learner with quantum memory can even solve more general tasks than that in \Cref{corollary:lower-bound-sparse-expectation-values-without-quantum-memory-from-doubly-stochastic-channel-queries}, see \Cref{corollary:pauli-sparse-expectations-from-choi-shadow-tomography,corollary:sample-complexity-predicting-pauli-sparse-channel-expectation-values}.)
Thereby, we have established an exponential separation between the weakest of our three learning models with quantum memory (\Cref{definition:qchannel-learning-with-qmemory-choi}) and the strongest of our three learning models without quantum memory (\Cref{definition:qchannel-learning-without-qmemory}).

\begin{remark}\label{remark:no-exponential-separation-pauli-channels}
    Also the exponential query complexity lower bounds in this subsection can fail under a priori promises on the unknown channel. 
    Concretely, \cite[Corollary 2]{chen2022quantum} gives a learning procedure without quantum memory that estimates all (non-trivial) PTM entries of an unknown $n$-qubit Pauli channel $\mathcal{N}$ using $\mathcal{O}(\tfrac{n\log(1/\delta)}{\varepsilon^2})$ queries to $\mathcal{N}$, where $n$-dimensional auxiliary quantum systems are used.
    Thus, if we are promised that the unknown channel is a Pauli channel, the exponential query complexity separation between PTM learning with and without quantum memory breaks down. In fact, the protocol of \cite{chen2022quantum} achieves the same $n$-dependence and a better $\varepsilon$-dependence than that of \Cref{theorem:ptm-entries-from-choi-shadow-tomography}, despite the latter using a quantum memory.
\end{remark}

\begin{remark}\label{remark:no-exponential-separation-average-case}
    As a further example of a change in framework can break the exponential query complexity lower bound for learners without quantum memory, we highlight the worst case nature of our task. Namely, we require the learner to make accurate predictions for all Pauli-sparse states and observables. 
    In contrast, \cite{huang2022learning} recently established that arbitrary quantum channels can be learned efficiently without quantum memory from parallel access in a worst-case sense over certain ($k$-local degree-$d$) observables and on-average over certain (locally flat) distributions of quantum states. 
\end{remark}

\begin{remark}\label{remark:lower-bound-no-auxiliary-system}
    To conclude our discussion of query complexity lower bounds, we explain the implications of \cite[Theorem 2, (iii)]{chen2022quantum} for our learning task. 
    Formulated in our language, \cite[Theorem 2, (iii)]{chen2022quantum} says: Without the possibility of querying the unknown channel on a subsystem of an entangled quantum state (as is for example required to prepare Choi state copies), the task of PTM learning requires $\Omega (2^{\nicefrac{n}{3}})$ many queries, even if the unknown channel is promised to be a Pauli channel (or, equivalently, to have a diagonal PTM).
    This teaches us about the importance of allowing for an auxiliary quantum system our definitions of learning with quantum memory. Combined with the query complexity lower bounds for learners without quantum memory established in this section, we see that the efficient query complexity of PTM learning with quantum memory established in \Cref{theorem:ptm-entries-from-choi-shadow-tomography} requires both the ability to query the unknown channel sequentially without intermediate measurements and the ability to let the unknown channel act on a subsystem of a larger composite quantum system.
\end{remark}

\section{Hamiltonian Learning With Quantum Memory}\label{section:hamiltonian-learning}

As an application of our positive results for Pauli transfer matrix learning with quantum memory, we show how these can be combined with polynomial interpolation techniques to give rise to a Hamiltonian learning protocol. The ideas presented in this section follow the approaches based on polynomial interpolation and derivative estimation put forward in \cite{haah2021optimal, stilck-franca2022efficient, gu2022practical}, but our PTM learner allows us to deal with arbitrary Hamiltonians, whereas these prior works require assumptions on the structure of the unknown Hamiltonian.

Our goal is to learn an unknown $n$-qubit Hamiltonian $H=H^\dagger\in\mathcal{B} ((\mathbb{C}^{2})^{\otimes n})$, without making prior assumptions on the structure of $H$.
Namely, if $H$ has the Pauli expansion $H = \sum_{A\in\{0,1,2,3\}^n} \alpha(A) \sigma_A$, with $\alpha(A) = \tfrac{1}{2^n} \tr[H\sigma_A] \in\mathbb{R}$ and w.l.o.g.~$\alpha(0^n)=0$, we aim to approximate the coefficient vector $\vec{\alpha} = (\alpha(A))_{A\in \{0,1,2,3\}^n}$ in $\ell_\infty$-norm over $\mathbb{R}^{4^n}$. 
(Note: Here we consider an unnormalized Pauli expansion.)

We will consider learning an unknown Hamiltonian from access to the associated unitary dynamics.
For any time $t\geq 0$, we will denote the unitary quantum channel describing the (Schr\"odinger picture) evolution under the Hamiltonian $H$ for time $t$ by
\begin{equation}
    \mathcal{U}_t :\mathcal{B} ((\mathbb{C}^{2})^{\otimes n})\to \mathcal{B} ((\mathbb{C}^{2})^{\otimes n}),~
    \mathcal{U}_t (\rho)
    = e^{-itH}\rho e^{itH}\, .
\end{equation}
By expanding the operator exponentials in their power series, we obtain:
\begin{equation}\label{eq:time-evolution-second-order-expansion}
    \mathcal{U}_t (\rho)
    = \rho - it [H, \rho] + \sum_{k=2}^\infty \frac{(-it)^k}{k!}C_k^{H}(\rho)\, ,
\end{equation}
where $C_k^{H}(\rho)$ denotes the $k$-fold iterated commutator defined via
\begin{equation}
    C_k^X (Y)\coloneqq \underbrace{[X,[X,[\ldots, [X}_{k\textrm{ times}},Y]\ldots]]]\, .
\end{equation}
Repeating an observation previously made in \cite{haah2021optimal, yu2023robust, stilck-franca2022efficient, gu2022practical}, we can isolate any single coefficient $\alpha (A)$ in the first-order term when acting on a suitable Pauli-sparse input state and measuring a suitable Pauli observable on the obtained output state:

\begin{lemma}\label{lemma:isolating-single-pauli-coefficient}
    Let $H = \sum_{A\in\{0,1,2,3\}^n} \alpha(A) \sigma_A$, with $\alpha(A) = \tfrac{1}{2^n} \tr[H\sigma_A] \in\mathbb{R}$.
    Let $A\in \{0,1,2,3\}^n$, $A\neq 0^n$.
    Let $1\leq j\leq n$ be such that $A_j\neq 0$. Let $B_j \in\{1,2,3\}\setminus\{A_i\}$ and set $B_k=0$ for all $k\neq i$.
    Define the quantum state $\rho = \tfrac{\mathbbm{1}_{2^n} + \nicefrac{i [\sigma_A, \sigma_B]}{2}}{2^n}$.
    Then
    \begin{equation}
        \tr[\sigma_B \cdot (-i [H, \rho])]
        = 2 \alpha (A) .
    \end{equation}
\end{lemma}
\begin{proof}
    This can be seen by a direct computation:
    \begin{align}
        \tr[\sigma_B \cdot (-i [H, \rho])]
        &= \tr[(i [H, \sigma_B] )\cdot \rho]\\
        &= -\frac{1}{2^{n+1}} \tr[[H, \sigma_B] \cdot [\sigma_A, \sigma_B]]\\
        &= -\frac{1}{2^{n+1}} \sum_{C\in\{0,1,2,3\}^n} \alpha(C) \tr[[\sigma_C, \sigma_B] \cdot [\sigma_A, \sigma_B]]\\
        &= \frac{1}{2^{n+1}} \sum_{\substack{C\in\{0,1,2,3\}^n\\ [\sigma_B, \sigma_C]\neq 0}} \alpha(C) \tr[[\sigma_C, \sigma_B] \cdot [\sigma_A, \sigma_B]]\\
        &= \frac{1}{2^{n+1}} \sum_{\substack{C\in\{0,1,2,3\}^n\\ [\sigma_B, \sigma_C]\neq 0}} \alpha(C) \tr[2 \sigma_B \sigma_C \cdot 2 \sigma_A \sigma_B]\\
        &= \frac{2}{2^{n}} \sum_{\substack{C\in\{0,1,2,3\}^n\\ [\sigma_B, \sigma_C]\neq 0}} \alpha(C) \tr[\sigma_C \cdot \sigma_A]\\
        &= 2\alpha (A) ,
    \end{align}
    as claimed.
\end{proof}

Recalling \Cref{eq:time-evolution-second-order-expansion}, we conclude: For any $A\neq 0^n$, if we manage to extract the first-order time derivative at time zero of $\tr[\sigma_B \mathcal{U}_t (\rho)]$ with $B$ and $\rho$ as in \Cref{lemma:isolating-single-pauli-coefficient}, then that tells us the coefficient $\alpha (A)$. 
We can achieve this by performing polynomial interpolation for the function $t\mapsto \tr[\sigma_B \mathcal{U}_t (\rho)]$. 
To ensure that polynomial interpolation gives an accurate estimate for the first derivative, we need bounds on higher order derivatives. The next lemma establishes such bounds for a general Hamiltonian:

\begin{lemma}\label{lemma:derivative-bounds}
    Let $B\in\{0,1,2,3\}^n$ be an arbitrary string and let $\rho\in\mathcal{S}((\mathbb{C}^2)^{\otimes n})$ be an arbitrary $n$-qubit state. 
    Then, for any $k\in\mathbb{N}$ and for any $\tau \geq0$,
    \begin{equation}
        \left\lvert \dv[k]{t}\eval{\Tr [\sigma_B\mathcal{U}_{t}(\rho)]}_{t=\tau}\right\rvert
        \leq (2\norm{H})^k\, .
    \end{equation}
\end{lemma}
\begin{proof}
    Using the power series expansion of $\mathcal{U}_t (\rho)$, one can evaluate the time derivative of a time-dependent expectation value as $\dv[k]{t}\eval{\Tr [\sigma_B\mathcal{U}_{t}(\rho)]}_{t=\tau} = (i)^k\Tr [C_k^H(\sigma_B) \mathcal{U}_{\tau}(\rho)]$. 
    As the operator norm of an iterated commutator satisfies $\norm{C_k^X (Y)}\leq (2 \norm{X})^k \norm{Y}$, an application of H\"older's inequality gives
    \begin{equation}
        \left\lvert \dv[k]{t}\eval{\Tr [\sigma_B\mathcal{U}_{t}(\rho)]}_{t=\tau}\right\rvert
        \leq \norm{C_k^H(\sigma_B)}\cdot \norm{\mathcal{U}_{\tau}(\rho)}_1
        \leq (2\norm{H})^k \norm{\sigma_B}
        = (2\norm{H})^k,
    \end{equation}
    since $\mathcal{U}_{\tau}$ is completely positive and trace-preserving for every $\tau\geq 0$ and $\rho$ is a quantum state.
\end{proof}

Together with polynomial interpolation guarantees for derivative estimation, our PTM learning protocol from \Cref{corollary:pauli-sparse-expectations-from-choi-shadow-tomography} now has the following consequence: 

\begin{theorem}\label{theorem:hamiltonian-learning}
    There is a learning procedure with quantum memory from parallel access that, using $m_1 = \tilde{\mathcal{O}}\left(\tfrac{n + \log(1/\delta)}{\varepsilon^4}\cdot\norm{H}^4\right)$ parallel queries to $\mathcal{O}(\tfrac{1}{\norm{H}})$-time evolutions along $H$ with a total evolution time of $T_1 = \tilde{\mathcal{O}}\left(\frac{n + \log(1/\delta)}{\varepsilon^4}\cdot \norm{H}^3\right)$, produces a classical description $\hat{H}$ of $H$ from which any $M$ Pauli coefficients $\alpha (A^{(i)})$, $A^{(i)}\in\{0,1,2,3\}^n$, $1\leq i\leq M$, can be estimated to accuracy $\varepsilon$ with $m_2 = \tilde{\mathcal{O}}\left(\tfrac{\log(M) + \log(1/\delta)}{\varepsilon^2}\cdot \norm{H}^2\right)$ additional parallel queries to $\mathcal{O}(\tfrac{1}{\norm{H}})$-time evolutions along $H$ with a total evolution time of $T_2 = \tilde{\mathcal{O}}\left(\tfrac{n + \log(1/\delta)}{\varepsilon^2}\cdot \norm{H}\right)$ and classical computation time $\tilde{\mathcal{O}}\left(\tfrac{n^2 + n \log(1/\delta)}{\varepsilon^4}\cdot \norm{H}^4 M\right)$, with success probability $\geq 1-\delta$.
    The classical representation $\hat{H}$ of $H$ consists of $\tilde{\mathcal{O}}\left(\tfrac{n^2 + n \log(1/\delta)}{\varepsilon^4}\cdot\norm{H}^4\right)$ real numbers stored in classical memory.
\end{theorem}
\begin{proof}[Proof Sketch]
    We give a detailed proof in \Cref{appendix:proofs}.
    The learning procedure works as follows:
    First, we use \Cref{corollary:pauli-sparse-expectations-from-choi-shadow-tomography} to build classical representations $\hat{\mathcal{U}}_t$ of $\mathcal{U}_t$ for different times $t$. The totality of those classical representations makes up our classical representation $\hat{H}$ of $H$.
    When asked to estimate a Pauli coefficients $\alpha (A^{(i)})$, $1\leq i\leq M$, from that classical representation, we use each $\hat{\mathcal{U}}_t$ to obtain an estimate of $\tfrac{1}{2}\tr[\sigma_{B^{(i)}} \mathcal{U}_t (\rho^{(i)})]$ with $B^{(i)}$ and $\rho^{(i)}$ as in \Cref{lemma:isolating-single-pauli-coefficient} for different times, and then use those estimates to approximately evaluate the first-order time derivative at time $0$ via Chebyshev polynomial interpolation for each $i$. 
    Motivated by \Cref{eq:time-evolution-second-order-expansion} and \Cref{lemma:isolating-single-pauli-coefficient}, these approximate derivatives are our estimates for the Pauli coefficients $\alpha (A^{(i)})$.
\end{proof}

Assuming an upper bound of $\norm{H}\leq\mathcal{O}(\poly (n))$, \Cref{theorem:hamiltonian-learning} gives an efficient procedure for learning the Pauli coefficients of $H$ in $\ell_\infty$-norm. Crucially, and in contrast to results from prior work, we do not need any structural assumptions about $H$, it can be an arbitrarily non-local Hamiltonian.
However, \Cref{theorem:hamiltonian-learning} is not directly comparable to prior works on learning Hamiltonians from real time dynamics. Methods making use of known structure in the Hamiltonian achieve exponentially better $n$ dependence and polynomially better $\varepsilon$ dependence than \Cref{theorem:hamiltonian-learning}, 
and they do not scale directly with $\norm{H}$ but rather with the the $\ell_\infty$-norm of the vector of Pauli coefficients of $H$ and with, for example, the assumed locality parameter \cite{haah2021optimal, stilck-franca2022efficient, gu2022practical, huang2022heisenberg-scaling}.


\section*{Acknowledgments}

MCC thanks Hsin-Yuan Huang, Sumeet Khatri, John Preskill, and Asad Raza for insightful discussions, as well as Antonio Anna Mele, Andreas Bluhm, Alexander Nietner, Daniel Stilck França, Srinivasan Arunachalam, and the anonymous reviewers at FOCS 2023 for helpful feedback on an earlier draft of this paper.
MCC was supported by a DAAD PRIME fellowship. 
The Institute for Quantum Information and Matter is an NSF Physics Frontiers Center.

\newpage
\setcounter{secnumdepth}{0}
\defbibheading{head}{\section{References}}
\sloppy
\printbibliography[heading=head]

\newpage
\appendix
\setcounter{secnumdepth}{2}

\section{Learning quantum channels with parallel access}\label{appendix:parallel-access}

We present one more access model for quantum channel learning, this one lying in between the general access model and the Choi access model.
Here, we suppose that the learner has to perform all oracle accesses to the unknown quantum channel in parallel. We formalize this as follows:

\begin{definition}[Learning quantum channels without quantum memory from parallel access]\label{definition:qchannel-learning-without-qmemory-parallel}
    A $T$-query algorithm for learning an unknown $n$-qubit quantum channel $\mathcal{N}$ without quantum memory from parallel access is a $T$-query algorithm for learning an unknown $n$-qubit quantum channel $\mathcal{N}$ without quantum memory under the additional restriction that the input states must not be chosen adaptively.
\end{definition}

\begin{definition}[Learning quantum channels with quantum memory from parallel access]\label{definition:qchannel-learning-with-qmemory-parallel}
    A $T$-query algorithm for learning an unknown $n$-qubit quantum channel $\mathcal{N}$ with quantum memory from parallel access is a $1$-query algorithm for learning the unknown $(n T)$-qubit quantum channel $\mathcal{N}^{\otimes T}$ with quantum memory.
    That is, such an algorithm can prepare an input state $\rho_0\in\mathcal{S}((\mathbb{C}^2)^{\otimes n_{\mathrm{aux}}}((\mathbb{C}^2)^{\otimes n})^{\otimes T})$, access the oracle for $\mathcal{N}$ in parallel $T$ times to have the state $\rho_T^{\mathcal{N}} = (\operatorname{id}_{\mathrm{aux}}\otimes \mathcal{N}^{\otimes T})(\rho_0)$ in the quantum memory, and then perform a joint POVM $\{F_s\}_s\subseteq \mathcal{B}((\mathbb{C}^2)^{\otimes n_{\mathrm{aux}}}\otimes ((\mathbb{C}^2)^{\otimes n})^{\otimes T})$ on $\rho_T^{\mathcal{N}}$ to predict properties of $\mathcal{N}$.
\end{definition}

\begin{figure}
    \centering
    \includegraphics[width = 0.6\textwidth]{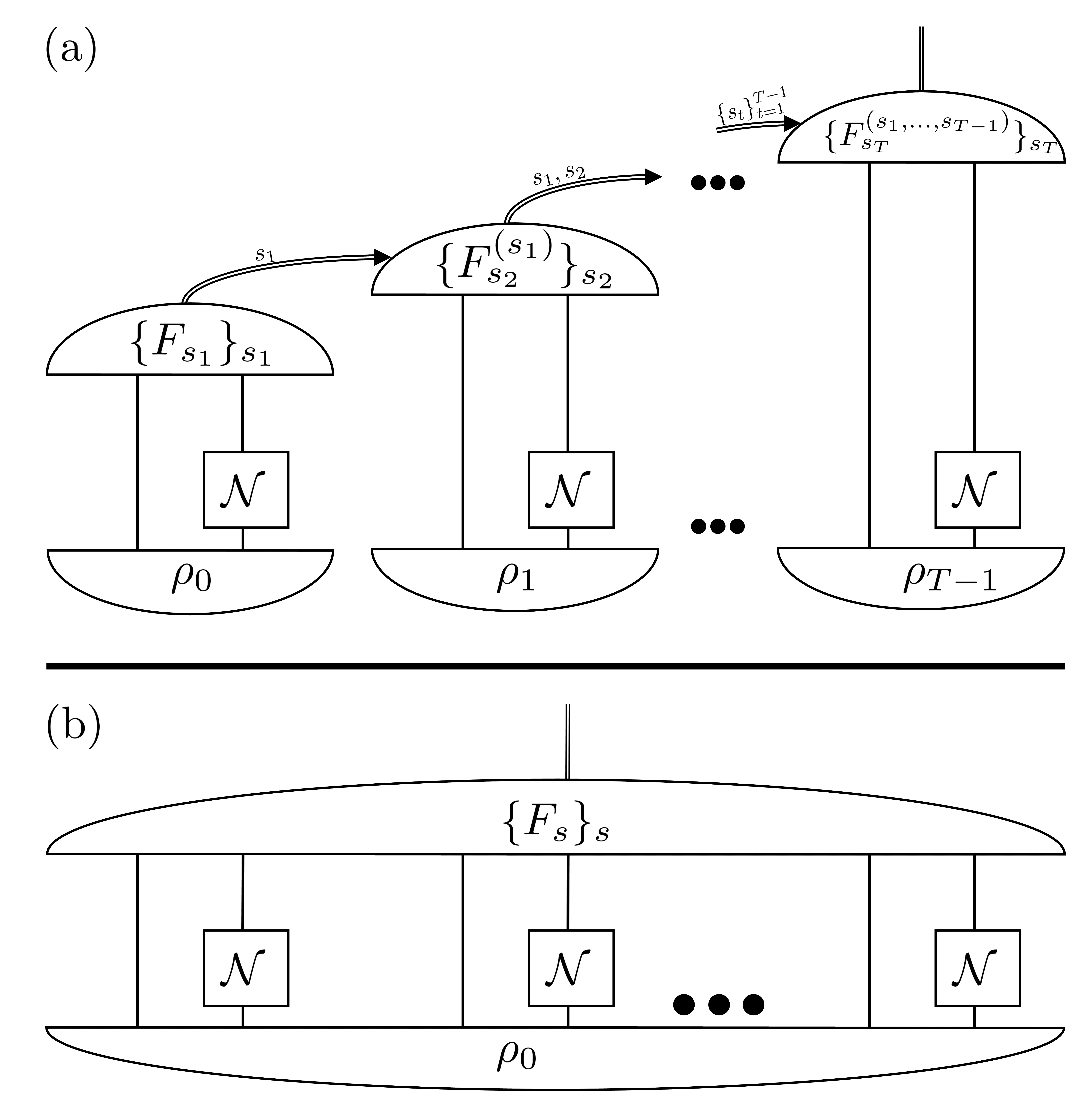}
    \caption{Illustration of learning quantum channels from parallel access, each panel to be read from bottom to top. Panel (a) depicts a learner without quantum memory, panel (b) depicts a learner with quantum memory.}
    \label{fig:parallel-access}
\end{figure}

These two definitions are illustrated in \Cref{fig:parallel-access}. Any procedure for learning an unknown quantum channel from Choi access can also be realized with parallel access to the unknown channel. And any procedure with parallel access is clearly an instance of a general channel learning procedure. Thus, the parallel access model sits in between Choi access and general channel access.

\section{Transfer Matrix Learning With Quantum Memory From Polynomially Many Queries}\label{section:tm-learning}

While the focus in \Cref{section:ptm-learning} was on the PTM, we can extend our reasoning to general TMs when replacing the Pauli shadow tomography procedure of \cite{huang2021information-theoretic} by the more general shadow tomography procedure of \cite{aaronson2018shadow, badescu2021improved}. In this appendix, we describe how that leads to a query-efficient procedure for learning general TMs, assuming a quantum memory and Choi state access. 
According to the results of \Cref{section:tm-estimates-give-expectation-estimates}, this allows us to estimate sparse expectation values for an unknown quantum channel. 
Moreover, as discussed in \Cref{section:different-learning-models} and \Cref{appendix:parallel-access}, this then in particular implies information-theoretically efficient learners with quantum memory from parallel access and from more general sequential access in terms of its power.

First, we demonstrate that applying existing state shadow tomography methods to the Choi state of the unknown quantum channel is sufficient to obtain accurate estimates for all transfer matrix entries simultaneously, while using only polynomially many Choi state copies.

\begin{theorem}\label{theorem:tm-entries-from-choi-shadow-tomography}
    There is a learning procedure with quantum memory from Choi access that, using $m=\tilde{\mathcal{O}}\left(\tfrac{n^3 + n \log(1/\delta)}{\varepsilon^4}\cdot \left( 2^{n} \cdot \max_{1\leq i\leq 4^n}\lVert Q_i \rVert_\infty^2\right)^4\right)$ copies of the Choi state $\tfrac{1}{2^n}\Gamma^{\mathcal{N}}$ of an unknown $n$-qubit channel $\mathcal{N}$, outputs, with success probability $\geq 1-\delta$, numbers $\hat{r}_{i,j}^\mathcal{Q}$ for $1\leq i,j\leq 4^n$ such that $\lvert  \hat{r}_{i,j}^\mathcal{Q} - (R_\mathcal{N}^\mathcal{Q})_{i,j}\rvert\leq\varepsilon$ holds simultaneously for all $1\leq i,j\leq 4^n$.
\end{theorem}

Notice that the copy complexity upper bound in \Cref{theorem:tm-entries-from-choi-shadow-tomography} becomes polynomial in $n$ if we assume $\lVert Q_i \rVert_\infty = \tfrac{1}{\sqrt{2^n}}$ for all $1\leq i\leq 4^n$. This is, for example, satisfied in the case of the Pauli ONB.
In fact, it is easy to see from $\norm{Q_i}_2\leq\sqrt{2^n}\norm{Q_i}$ that $\max_{1\leq i\leq 4^n}\lVert Q_i \rVert = \tfrac{1}{\sqrt{2^n}}$ holds if and only if $\mathcal{Q}$ is an ONB of normalized unitaries.

\begin{proof}
    We first define a set of effect operators, to which we then later apply \cite{badescu2021improved}'s improvement of \cite{aaronson2018shadow}'s shadow tomography procedure.
    Namely, for each $1\leq i,j\leq 4^n$, we define 
    \begin{equation}
        E_{i,j}
        =\frac{1}{2}\left(\mathds{1}\otimes \mathds{1} + \frac{1}{\lVert Q_j \rVert_\infty} Q_j^{\top}\otimes \frac{1}{\lVert Q_i \rVert_\infty} Q_i\right)\, .
    \end{equation}
    Notice that, since every $Q_i$ is Hermitian, so is every $E_{i,j}$.
    Moreover, since each $\frac{1}{\lVert Q_j \rVert_\infty} Q_j^{\top}\otimes \frac{1}{\lVert Q_i \rVert_\infty} Q_i$ is normalized in operator norm, an application of the triangle inequality shows that $0\leq E_{i,j}\leq\mathbbm{1}\otimes\mathbbm{1}$ holds for all $1\leq i,j\leq 4^n$. Thus, every $E_{i,j}$ is a valid effect operator.
    Also, note that by definition, we have:
    \begin{align}
        \Tr [Q_i \mathcal{N}(Q_j)]
        &= 2^n \Tr \left[(Q_j^{\top} \otimes Q_i) \frac{1}{2^n}\Gamma^\mathcal{N}\right] \label{eq:general-tm-via-choi-expectation}\\
        &= 2^n \cdot \lVert Q_i \rVert_\infty\cdot \lVert Q_j \rVert_\infty\cdot \Tr \left[\left(\frac{1}{\lVert Q_i \rVert_\infty} Q_i^{\top} \otimes \frac{1}{\lVert Q_j \rVert_\infty} Q_j \right) \frac{1}{2^n}\Gamma^\mathcal{N}\right]\\
        &= 2^n \cdot \lVert Q_i \rVert_\infty \cdot \lVert Q_j \rVert_\infty\cdot 2\cdot\left( \Tr \left[E_{i,j} \frac{1}{2^n}\Gamma^\mathcal{N}\right] - \frac{1}{2} \right) \, .
    \end{align}
    Hence, an $\tilde{\varepsilon}$-accurate estimate for the expectation value $\Tr \left[E_{i,j} \tfrac{1}{2^n}\Gamma^\mathcal{N}\right]$ gives rise to a $\left(\tilde{\varepsilon}\cdot 2^{n+1} \cdot \max_{1\leq i\leq 4^n}\lVert Q_i \rVert_\infty^2\right)$-accurate estimate of the transfer matrix entry $\left( R_\mathcal{N}^\mathcal{Q}\right)_{i,j}=\Tr [Q_i \mathcal{N}(Q_j)]$.
    
    Let $\varepsilon >0$ and $\delta\in (0,1)$ be arbitrary, define $\tilde{\varepsilon}=\frac{\varepsilon}{2^{n+1} \cdot \max_{1\leq i\leq 4^n}\lVert Q_i \rVert_\infty^2}>0$.
    According to \cite[Theorem 1.4]{badescu2021improved}, if we apply the shadow tomography procedure of \cite{badescu2021improved} for the $M=16^n$ effect operators $E_{i,j}$, $1\leq i,j\leq 4^n$, and the $(2n)$-qubit quantum state $\tfrac{1}{2^n}\Gamma^\mathcal{N}$, we see that
    \small
    \begin{equation}
        \frac{\left(\log^2 (M) + \log\left(\tfrac{2n}{\delta \varepsilon}\right)\right) \cdot 2n}{\tilde{\varepsilon}^4}\cdot \mathcal{O}\left(\log\left(\tfrac{2n}{\delta \varepsilon}\right)\right)
        = \tilde{\mathcal{O}}\left(\frac{n^3 + n\log(1/\delta)}{\varepsilon^4}\cdot \left( 2^{n} \cdot \max_{1\leq i\leq 4^n}\lVert Q_i \rVert_\infty^2\right)^4\right)
    \end{equation}
    \normalsize
    copies of $\tfrac{1}{2^n}\Gamma^\mathcal{N}$ suffice to produce estimates $\tilde{r}_{i,j}$ satisfying $\left\lvert \tilde{r}_{i,j} - \Tr \left[E_{i,j} \tfrac{1}{2^n}\Gamma^\mathcal{N}\right]\right\rvert\leq\tilde{\varepsilon}$ for all $1\leq i,j\leq 4^n$, with success probability $\geq 1-\delta$.
    By our above reasoning, these simultaneously $\tilde{\varepsilon}$-accurate expectation value estimates now give us  simultaneously $\varepsilon$-accurate estimates 
    \begin{equation}
        \hat{r}_{i,j}^\mathcal{Q}
        = 2^n \cdot \lVert Q_i \rVert_\infty \cdot \lVert Q_j \rVert_\infty\cdot 2\cdot\left( \tilde{r}_{i,j} - \frac{1}{2} \right)
    \end{equation}
    of the TM entries, with success probability $\geq 1-\delta$.
\end{proof}

\begin{remark}
    There are two immediate extensions of \Cref{theorem:tm-entries-from-choi-shadow-tomography} that we want to highlight.
    First, we have phrased \Cref{theorem:tm-entries-from-choi-shadow-tomography} in a way suggesting that the learner knows the ONB $\mathcal{Q}$ in advance. This is not needed. As the shadow tomography procedure in \cite{badescu2021improved} allows for adaptively chosen measurements, also in our case the same copy complexity bound holds if the learner does not know $\mathcal{Q}$ in advance and the ONB elements are only revealed adaptively. Note, however, that the learner must know an upper bound on $\max_{1\leq i\leq 4^n}\lVert Q_i \rVert_\infty^2$ in advance to determine the suitable copy complexity.
    
    Second, while there is only a single ONB $\mathcal{Q}$ in \Cref{theorem:tm-entries-from-choi-shadow-tomography}, we can also allow for $N$ different ONBs $\mathcal{Q}^{(1)},\ldots,\mathcal{Q}^{(N)}$. This will increase the number of effect operators in the proof from $16^n$ to $N\cdot 16^n$, but since that number only enters the sample complexity bound for shadow tomography logarithmically, the final copy complexity will still be bounded as $\tilde{\mathcal{O}}\left(\tfrac{n^3 + n\log (N) + n \log(1/\delta)}{\varepsilon^4}\cdot \left( 2^{n} \cdot \max_{1\leq i\leq 4^n, 1\leq L\leq N}\lVert Q_i^{(L)} \rVert_\infty^2\right)^4\right)$. 
    
    In summary, we can sample-efficiently simultaneously learn TMs w.r.t.~exponentially many different adaptively chosen ONBs, if we know in advance how many different ONBs there are and that the maximum operator norm of any ONB element is $\tfrac{1}{\sqrt{2^n}}$.
    These two extensions also carry over to the next result, but we do not state the corresponding variants explicitly.
\end{remark}

\begin{remark}
    Comparing \Cref{theorem:tm-entries-from-choi-shadow-tomography} with \Cref{theorem:ptm-entries-from-choi-shadow-tomography}, we see that our bound in the Pauli-specific case provides an asymptotic copy complexity improvement by a factor of $n^2$ compared to simply applying the bound obtained when applying \Cref{theorem:tm-entries-from-choi-shadow-tomography} to the Pauli ONB (up to logarithmic factors). Thus, both \Cref{theorem:tm-entries-from-choi-shadow-tomography,theorem:ptm-entries-from-choi-shadow-tomography} give polynomial-in-$n$ copy complexity bounds for learning the Pauli transfer matrix, but the former gives an effectively cubic $n$-dependence, which the latter improves to linear.
    Moreover, as the general state shadow tomography results of \cite{aaronson2018shadow, badescu2021improved} do not come with bounds on classical memory or computational complexity, we also do not have such bounds in the case of a general TM, in contrast to the PTM case (\Cref{corollary:pauli-sparse-expectations-from-choi-shadow-tomography}).
\end{remark}

Combining \Cref{theorem:tm-entries-from-choi-shadow-tomography} with \Cref{corollary:tm-estimates-give-sparse-expectation-estimates}, we obtain:

\begin{corollary}\label{corollary:sample-complexity-predicting-sparse-expectation-values}
    There is a learning procedure with quantum memory from Choi access that, using
    \begin{equation}\label{eq:sample-complexity-predicting-sparse-expectation-values}
        m
        =\tilde{\mathcal{O}}\left(\frac{n^3 + n\log(1/\delta)}{\varepsilon^4}\cdot B^4 s_\rho^4 s_O^2 \cdot \left(\sqrt{2^{n}} \cdot \max_{1\leq i\leq 4^n}\lVert Q_i \rVert_\infty\right)^{12} \right)
    \end{equation}
    copies of the Choi state $\tfrac{1}{2^n}\Gamma^{\mathcal{N}}$ of an $n$-qubit channel $\mathcal{N}$, produces, with success probability $\geq 1-\delta$, a classical description $\hat{\mathcal{N}}$ of $\mathcal{N}$ from which all expectation values of the form $\operatorname{tr}[O\mathcal{N}(\rho)]$, where $O$ is an $n$-qubit observable with $s_O$-sparse $\mathcal{Q}$-basis expansion such that $\lVert O\rVert_2\leq B\sqrt{2^n}$ and $\rho$ is an $n$-qubit state with $s_\rho$-sparse $\mathcal{Q}$-basis expansion, can be predicted to accuracy $\varepsilon$.
\end{corollary}
\begin{proof}
    By \Cref{corollary:tm-estimates-give-sparse-expectation-estimates}, we can get the desired $\varepsilon$-accurate expectation value estimates from $\tilde{\varepsilon}$-accurate transfer matrix entry estimates if we set $\tilde{\varepsilon}=\tfrac{\varepsilon}{s_\rho\sqrt{s_O} \norm{O}_2 \max_{1\leq i\leq 4^n} \lVert Q_i\rVert_\infty}$.
    Plugging this accuracy $\tilde{\varepsilon}$ into the copy complexity bound of \Cref{theorem:tm-entries-from-choi-shadow-tomography} and inserting the assumed inequality $\lVert O\rVert_2\leq B\sqrt{2^n}$ now gives the stated copy complexity bound.
\end{proof}

Note that in particular any $n$-qubit observable $O$ with $\lVert O\rVert_\infty\leq B$ satisfies the Hilbert-Schmidt norm bound $\lVert O\rVert_2\leq B\sqrt{2^n}$. Thus, the corollary can be applied to bounded $\mathcal{Q}$-sparse observables.

Exactly analogous to the Pauli case, if we are promised in advance that the unknown quantum channel has a sparse transfer matrix, then we can combine \Cref{theorem:tm-entries-from-choi-shadow-tomography} with a cut-off argument to obtain the following variant of our result:

\begin{lemma}\label{lemma:sparse-channel-tm-entries-from-choi-shadow-tomography}
    There is a learning procedure with quantum memory from Choi access that, using $m=\tilde{\mathcal{O}}\left(\tfrac{n^3 + n \log(1/\delta)}{\varepsilon^4}\cdot \left( 2^{n} \cdot \max_{1\leq i\leq 4^n}\lVert Q_i \rVert_\infty^2\right)^4\right)$ copies of the Choi state $\tfrac{1}{2^n}\Gamma^{\mathcal{N}}$ of an unknown $n$-qubit channel $\mathcal{N}$ with $s_\mathcal{N}$-sparse TM w.r.t.~$\mathcal{Q}$, outputs, with success probability $\geq 1-\delta$, a list $\{((i_{s},j_{s}), \hat{r}_{i_s,j_s}^\mathcal{Q})\}_{s=1}^{S}$ of length $S$ such that
    \begin{enumerate}
        \item $(i,j)\not\in \{(i_{s},j_{s})\}_{s=1}^{S}$ $\Rightarrow$ $\lvert (R_\mathcal{N}^\mathcal{Q})_{i,j}\rvert\leq\varepsilon$,
        \item $(i,j)\in \{(i_{s},j_{s})\}_{s=1}^{S}$ $\Rightarrow$ $\lvert (R_\mathcal{N}^\mathcal{Q})_{i,j}\rvert > \nicefrac{\varepsilon}{6}$,
        \item $S\leq s_\mathcal{N}$, and
        \item $\lvert  \hat{r}_{i_s,j_s}^\mathcal{Q} - (R_\mathcal{N}^\mathcal{Q})_{i_s,j_s}\rvert\leq\varepsilon$ holds for all $1\leq s \leq S$.
    \end{enumerate}
\end{lemma}
\begin{proof}
    The proof is essentially identical to that of \Cref{corollary:sample-complexity-predicting-pauli-sparse-channel-expectation-values}, we only have to replace \Cref{theorem:ptm-entries-from-choi-shadow-tomography} by \Cref{theorem:tm-entries-from-choi-shadow-tomography}.
\end{proof}

Combining \Cref{lemma:sparse-channel-tm-entries-from-choi-shadow-tomography} with \Cref{lemma:tm-entry-estimates-give-sparse-channel-expectation-estimates}, we get a variant of \Cref{corollary:sample-complexity-predicting-sparse-expectation-values} in which we replace the sparsity assumption on states and observables by a sparsity assumption on the unknown channel:

\begin{corollary}\label{corollary:sample-complexity-predicting-sparse-channel-expectation-values}
    There is a learning procedure with quantum memory from Choi access that, using
    \begin{equation}
        m
        =\tilde{\mathcal{O}}\left(\frac{n^3 + n\log(1/\delta)}{\varepsilon^4}\cdot B^4 s_\mathcal{N}^4 \cdot \left(\sqrt{2^{n}} \cdot \max_{1\leq i\leq 4^n}\lVert Q_i \rVert_\infty\right)^{12} \right)
    \end{equation}
    copies of the Choi state $\tfrac{1}{2^n}\Gamma^{\mathcal{N}}$ of an $n$-qubit channel $\mathcal{N}$ with $s_\mathcal{N}$-sparse TM w.r.t.~$\mathcal{Q}$, produces, with success probability $\geq 1-\delta$, a classical description $\hat{\mathcal{N}}$ of $\mathcal{N}$ from which all expectation values of the form $\operatorname{tr}[O\mathcal{N}(\rho)]$, where $O$ is an $n$-qubit observable with $\lVert O\rVert_2\leq B\sqrt{2^n}$ and $\rho$ is an $n$-qubit state, can be predicted to accuracy $\varepsilon$.
\end{corollary}
\begin{proof}
    Recalling \Cref{lemma:tm-entry-estimates-give-sparse-channel-expectation-estimates}, we see that $\tilde{\varepsilon}$-accurate transfer matrix entry estimates are sufficient to get the desired $\varepsilon$-accurate expectation value estimates if we take $\tilde{\varepsilon}=\tfrac{\varepsilon}{s_\mathcal{N} \norm{O}_2 \max_{1\leq i\leq 4^n} \lVert Q_i\rVert_\infty}$.
    Plugging this accuracy $\tilde{\varepsilon}$ into the copy complexity bound of \Cref{lemma:sparse-channel-tm-entries-from-choi-shadow-tomography} and inserting the assumed inequality $\lVert O\rVert_2\leq B\sqrt{2^n}$ now gives the stated copy complexity bound.
\end{proof}

If $\mathcal{Q}$ is an ONB of normalized unitaries, then \Cref{corollary:sample-complexity-predicting-sparse-channel-expectation-values} describes a procedure for sample-efficiently predicting expectation values for arbitrary input states and arbitrary bounded output observables if the unknown quantum channel is promised to have a polynomially-sparse TM. 
Here, no a priori knowledge about the structure of the (sparse) TM is needed.

\section{Learning Tree Formalism for Learning Without Quantum Memory}\label{appendix:learning-tree-formalism}

Here, we recall (part of) the learning tree formalism introduced in \cite{chen2022exponential} to establish query complexity lower bounds for learners without quantum memory. 
More precisely, we need the following definition of the tree representation for learning a quantum channel without quantum memory:

\begin{definition}[Tree representation for learning channels {\cite[Definition 7.1]{chen2021exponential-arxiv}}]\label{definition:learning-tree-representation-channel}
Let $\mathcal{N}$ be an $n$-qubit quantum channel.
An algorithm for learning from $T$ queries to $\mathcal{N}$ without quantum memory can be represented as a rooted tree $\mathcal{T}$ of depth $T$ in which each node encodes all measurement outcomes the algorithm has received thus far.  
The tree has the following properties:
\begin{itemize}
    \item Each node $u$ has an associated probability $p^{\mathcal{N}}(u)$.
    \item The probability associated to the root $r$ of the tree equals $1$, $p^{\mathcal{N}}(r) = 1$.
    \item At each non-leaf node $u$, we prepare a state $\ket{\phi_u}\in (\mathbb{C}^2)^{\otimes n_{\mathrm{aux}}}\otimes (\mathbb{C}^2)^{\otimes n}$, apply the channel $\operatorname{id}_{\mathrm{aux}}\otimes \mathcal{N}$ to $\ketbra{\phi_u}{\phi_u}$, and measure a ($u$-dependent) rank-1 POVM $\{\sqrt{w_s^u 2^n 2^{n_{\mathrm{aux}}}} \ketbra{\psi_s^u}{\psi_s^u} \}_s\subseteq \mathcal{B}( (\mathbb{C}^2)^{\otimes n_{\mathrm{aux}}}\otimes (\mathbb{C}^2)^{\otimes n})$ to obtain a classical outcome $s$. 
    \item If $v$ is a child node of $u$, then the probability of the child node is related to the probability of its parent via
    \begin{equation}
        p^{\mathcal{N}}(v) 
        = p^{\mathcal{N}}(u) \cdot w_s^u 2^n 2^{n_{\mathrm{aux}}} \cdot \bra{\psi_v} (\operatorname{id}_{\mathrm{aux}}\otimes \mathcal{N})(\ketbra{\phi_u}{\phi_u}) \ket{\psi_v}.
    \end{equation}
    \item Each root-to-leaf path in $\mathcal{T}$ is of length $T$.
    If $\ell$ is a leaf node, then $p^{\mathcal{C}}(\ell)$ is the probability of the classical memory being in state $\ell$ after the learning procedure.
\end{itemize}
\end{definition}

While \Cref{definition:learning-tree-representation-channel} implicitly assumes all input states to be pure and all measurements to be rank-$1$ POVMs, these assumptions can be made without loss of generality via standard purification arguments and via simulating general POVMs with rank-$1$ POVMs. (Compare \cite[Lemma 4.8 and Remark 4.19]{chen2021exponential-arxiv} for the latter.)
Also, we can w.l.o.g.~assume the auxiliary system size $n_{\mathrm{aux}}$ to be the same for each query to the channel oracle.

\section{Additional Proofs}\label{appendix:proofs}

\subsection{Proofs of \Cref{lemma:lower-bound-ptm-learning-without-quantum-memory-from-choi-state-copies,lemma:lower-bound-ptm-learning-without-quantum-memory-from-doubly-stochastic-choi-state-copies}}

To prove \Cref{lemma:lower-bound-ptm-learning-without-quantum-memory-from-choi-state-copies,lemma:lower-bound-ptm-learning-without-quantum-memory-from-doubly-stochastic-choi-state-copies}, we make use of the following two variations of \cite[Theorem 5.5]{chen2021exponential-arxiv}:

\begin{lemma}[Shadow tomography lower bound for Choi states]\label{lemma:lower-bound-choi-shadow-tomography-general}
    Let $O_1,\ldots, O_M\in\mathcal{B}((\mathbb{C}^2)^{\otimes 2n})$ be $M\in 2\mathbb{N}$ traceless and self-adjoint $(2n)$-qubit observables satisfying the following properties:
    \begin{enumerate}[(i)]
        \item For every $1\leq i\leq M$, $\sigma (O_i)\subseteq\{-1,1\}$. In particular, $\norm{O_i} = 1$ holds for all $1\leq i\leq M$.
        \item For every $1\leq i\leq M$, $\tr_{n+1,\ldots, 2n}[O_i]=0$.
        \item For every $1\leq i\leq \nicefrac{M}{2}$, $O_i = - O_{i + \nicefrac{M}{2}}$.
    \end{enumerate}
    Any algorithm without quantum memory requires
    \begin{equation}
        \Omega \left( \frac{1}{\varepsilon^2 \delta (O_1,\ldots,O_M)} \right)
    \end{equation}
    copies of the Choi state $\tfrac{1}{2^n}\Gamma^\mathcal{N}$ of an unknown $n$-qubit channel $\mathcal{N}$ to simultaneously predict all expectation values $\tr[O_i \tfrac{1}{2^n}\Gamma]$, $1\leq i\leq M$, up to accuracy $\varepsilon$ with a success probability $\geq\tfrac{2}{3}$.
    Here, we defined
    \begin{equation}
        \delta (O_1,\ldots,O_M)
        \coloneqq \sup_{\ket{\phi}\in (\mathbb{C}^2)^{\otimes 2n}: \norm{\phi}=1} \frac{2}{M}\sum_{i=1}^{\nicefrac{M}{2}} \bra{\phi}O_i\ket{\phi}^2 .
    \end{equation}
\end{lemma}
\begin{proof}
    We only have to notice that the maximally mixed state on $2n$ qubits is a valid Choi state, and that the assumptions on the $O_i$ ensure that
    \begin{equation}
        \frac{1}{2^n}\Gamma_i
        \coloneqq \frac{\mathbbm{1}_{2}^{\otimes 2n} + 3\varepsilon O_i}{2^{2n}}
    \end{equation}
    is a valid Choi state for every $1\leq i\leq M$, as long as $\varepsilon < \nicefrac{1}{3}$.
    Now, the remainder of the proof is exactly the same as the proof of \cite[Theorem 5.5]{chen2021exponential-arxiv}.
\end{proof}

\begin{lemma}[Shadow tomography lower bound for Choi states of doubly-stochastic sparse channels]\label{lemma:lower-bound-doubly-stochastic-choi-shadow-tomography-general}
    Let $O_1,\ldots, O_M\in\mathcal{B}((\mathbb{C}^2)^{\otimes 2n})$ be $M\in 2\mathbb{N}$ traceless and self-adjoint $(2n)$-qubit observables satisfying the following properties:
    \begin{enumerate}[(i)]
        \item For every $1\leq i\leq M$, $\sigma (O_i)\subseteq\{-1,1\}$. In particular, $\norm{O_i} = 1$ holds for all $1\leq i\leq M$.
        \item For every $1\leq i\leq M$, $\tr_{1,\ldots, n}[O_i]=0$ and $\tr_{n+1,\ldots, 2n}[O_i]=0$.
        \item For every $1\leq i\leq \nicefrac{M}{2}$, $O_i = - O_{i + \nicefrac{M}{2}}$.
        \item For every $1\leq i\leq M$, $O_i$ has a $\mathcal{O}(1)$-sparse Pauli ONB expansion.
        \item For every $1\leq i\leq M$, $O_i$ is a tensor product of $2n$ single-qubit observables.
    \end{enumerate}
    Any algorithm without quantum memory requires
    \begin{equation}
        \Omega \left( \frac{1}{\varepsilon^2 \delta (O_1,\ldots,O_M)} \right)
    \end{equation}
    copies of the Choi state $\tfrac{1}{2^n}\Gamma^\mathcal{N}$ of an unknown doubly-stochastic and entanglement-breaking $n$-qubit channel $\mathcal{N}$ with $\mathcal{O}(1)$-sparse PTM to simultaneously predict all expectation values $\tr[O_i \tfrac{1}{2^n}\Gamma]$, $1\leq i\leq M$, up to accuracy $\varepsilon$ with a success probability $\geq\tfrac{2}{3}$.
    Here, we defined
    \begin{equation}
        \delta (O_1,\ldots,O_M)
        \coloneqq \sup_{\ket{\phi}\in (\mathbb{C}^2)^{\otimes 2n}: \norm{\phi}=1} \frac{2}{M}\sum_{i=1}^{\nicefrac{M}{2}} \bra{\phi}O_i\ket{\phi}^2 .
    \end{equation}
\end{lemma}
\begin{proof}
    We only have to notice that the maximally mixed state on $2n$ qubits is a valid Choi state of a doubly-stochastic and entanglement-breaking quantum channel with $\mathcal{O}(1)$-sparse PTM, and that the assumptions (i) and (ii) on the $O_i$ ensure that
    \begin{equation}
        \frac{1}{2^n}\Gamma_i
        \coloneqq \frac{\mathbbm{1}_{2}^{\otimes 2n} + 3\varepsilon O_i}{2^{2n}}
    \end{equation}
    is a valid Choi state of a doubly-stochastic quantum channel $\mathcal{N}_i$ for every $1\leq i\leq M$, as long as $\varepsilon < \nicefrac{1}{3}$. 
    Moreover, the channel $\mathcal{N}_i$ has a $\mathcal{O}(1)$-sparse PTM because of \Cref{eq:ptm-as-pauli-choi-expectation-value} and assumption (iv), and $\mathcal{N}_i$ is entanglement-breaking because its Choi state $\tfrac{1}{2^n}\Gamma_i$ is separable by assumptions (i) and (v), which can be seen by expanding $O_i$ in its eigenbasis consisting of tensor products of eigenvectors.
    Now, the remainder of the proof is exactly the same as the proof of \cite[Theorem 5.5]{chen2021exponential-arxiv}.
\end{proof}

To apply these two lemmata, we need to establish upper bounds on $\delta (O_1,\ldots,O_M)$ for sets of Pauli observables satisfying the respective requirements. We do so in the following two lemmata, with a reasoning quite similar to that of \cite[Lemma 5.8]{chen2021exponential-arxiv}:

\begin{lemma}\label{lemma:delta-choi-pauli-shadow-tomography}
    Consider the following $2\cdot 4^n (4^{n} -1)$ many $(2n)$-qubit Pauli observables
    \begin{align}
        P_1,\ldots,P_{4^n (4^{n} -1)} &\in \{+ \sigma_A\}_{A\in\{0,1,2,3\}^{2n}}\setminus\{+ \sigma_B\otimes \mathbbm{1}_2^{\otimes n}\}_{B\in\{0,1,2,3\}^n}\\
        P_{4^n (4^{n} -1) + 1},\ldots,P_{2\cdot 4^n (4^{n} -1)} &\in \{- \sigma_A\}_{A\in\{0,1,2,3\}^{2n}}\setminus\{- \sigma_B\otimes \mathbbm{1}_2^{\otimes n}\}_{B\in\{0,1,2,3\}^n}.
    \end{align}
    The set of observables $\{P_i\}_{i=1}^{2\cdot 4^n (4^{n} -1)}$, when suitably ordered, satisfies the conditions (i)-(iii) from \cref{lemma:lower-bound-choi-shadow-tomography-general}.
    Moreover, 
    \begin{equation}
        \delta (P_1,\ldots, P_{2\cdot 4^n (4^{n} -1)})
        = \frac{1}{4^n}.
    \end{equation}
\end{lemma}
\begin{proof}
    Each $P_i$ is a Pauli observable and thus clearly satisfies condition (i).
    Also, each $P_i$ satisfies condition (ii) because, for every $A\in\{0,1,2,3\}^{2n}$, $\tr_{n+1,\ldots,2n}[\sigma_A]\neq 0$ holds if and only if $A_{j}=0$ for all $n+1\leq j\leq 2n$.
    Condition (iii) is satisfied by construction (when ordering the observables suitably).
    It remains to compute $\delta (P_1,\ldots, P_{2\cdot 4^n (4^{n} -1)})$. To this end, let $\ket{\phi}\in (\mathbb{C}^2)^{\otimes 2n}$ be an arbitrary normalized pure state, $\norm{\phi}=1$. Then
    \small
    \begin{align}
        &\frac{1}{4^n (4^{n} -1)} \sum_{i=1}^{4^n (4^{n} -1)} \bra{\phi} P_i\ket{\phi}^2\\
        &= \frac{1}{4^n (4^{n} -1)} \tr\left[ \left(\sum_{i=1}^{4^n (4^{n} -1)} P_i\otimes P_i\right)(\ketbra{\phi}{\phi})^{\otimes 2} \right]\\
        &= \frac{1}{4^n (4^{n} -1)} \tr\Bigg[ 
        \Bigg(\Bigg(\sum_{A\in\{0,1,2,3\}^{2n}} \sigma_A\otimes \sigma_A\Bigg)\\
        &\hphantom{= \frac{1}{4^n (4^{n} -1)} \tr\Bigg[ 
        \Bigg(}-\Bigg(\sum_{B\in\{0,1,2,3\}^{n}} (\sigma_B\otimes \mathbbm{1}_2^{\otimes n})\otimes (\sigma_B\otimes \mathbbm{1}_2^{\otimes n})\Bigg) \Bigg)(\ketbra{\phi}{\phi})^{\otimes 2} \Bigg]\\
        &= \frac{1}{4^n (4^{n} -1)} \left(\tr\left[ 2^{2n} \operatorname{SWAP}_{1,\ldots,2n}(\ketbra{\phi}{\phi})^{\otimes 2}\right] - \tr\left[2^n \operatorname{SWAP}_{1,\ldots,n} (\tr_{n+1,\ldots,2n}[\ketbra{\phi}{\phi}])^{\otimes 2} \right]\right)\\
        &= \frac{1}{4^n (4^{n} -1)} \left(2^{2n} - 2^{n} \tr[\tr_{n+1,\ldots,2n}[\ketbra{\phi}{\phi}]^2 ]\right) .
    \end{align}
    \normalsize
    Here, we used the following well known identity:
    \begin{equation}\label{eq:swap-via-sum-of-paulis}
        \sum_{B\in\{0,1,2,3\}^{n}} \sigma_B\otimes\sigma_B
        = 2^n \operatorname{SWAP}_{1,\ldots,n} .
    \end{equation}
    This is a special case of the general representation of the swap operator in any orthonormal basis of operators, see for example \cite[Example 1.2]{wolf2012quantumchannels}.
    
    Accordingly, we get 
    \begin{align}
        \delta (P_1,\ldots, P_{2\cdot 4^n (4^{n} -1)})
        &= \sup_{\ket{\phi}\in (\mathbb{C}^2)^{\otimes 2 n}: \norm{\phi}=1} \frac{1}{4^n (4^{n} -1)} \sum_{i=1}^{4^n (4^{n} -1)} \bra{\phi} P_i\ket{\phi}^2\\
        &= \frac{2^{2n} - 2^{n}\cdot \tfrac{1}{2^n}}{4^n (4^{n} -1)}\\
        &= \frac{1}{4^n} ,
    \end{align}
    and the supremum is attained at $\ket{\phi}$ if and only if $\tr_{n+1,\ldots,2n}[\ketbra{\phi}{\phi}]$ is maximally mixed, which is equivalent to $\ket{\phi}$ being maximally entangled across the cut $\{1,\ldots,2n\}= \{1,\ldots,n\}\cup \{n+1,\ldots,2n\} $.
\end{proof}

\begin{lemma}\label{lemma:delta-doubly-stochastic-choi-pauli-shadow-tomography}
    Consider the following $2(4^{2n} - 2\cdot 4^n +1) = 2(4^n - 1)^2 $ many $(2n)$-qubit Pauli observables
    \begin{align}
        P_1,\ldots,P_{(4^n - 1)^2} &\in \{+ \sigma_A\}_{A\in\{0,1,2,3\}^{2n}:\exists 1\leq i\leq n, n+1\leq j\leq 2n: A_i\neq 0 \neq A_j}\\
        P_{4^{2n} - 2\cdot 4^n +2},\ldots,P_{2(4^n - 1)^2} &\in \{- \sigma_A\}_{A\in\{0,1,2,3\}^{2n}:\exists 1\leq i\leq n, n+1\leq j\leq 2n: A_i\neq 0 \neq A_j}.
    \end{align}
    The set of observables $\{P_i\}_{i=1}^{2 (4^n - 1)^2}$, when suitably ordered, satisfies the conditions (i)-(v) from \cref{lemma:lower-bound-doubly-stochastic-choi-shadow-tomography-general}.
    Moreover, 
    \begin{equation}
        \delta (P_1,\ldots, P_{2((4^n - 1)^2)})
        = \frac{1}{4^n-1}.
    \end{equation}
\end{lemma}
\begin{proof}
    Each $P_i$ is a Pauli tensor product observable and thus clearly satisfies conditions (i), (iv), and (v).
    Also, each $P_i$ satisfies condition (ii) because, for every $A\in\{0,1,2,3\}^{2n}$, $\tr_{n+1,\ldots,2n}[\sigma_A]\neq 0$ holds if and only if $A_{j}=0$ for all $n+1\leq j\leq 2n$, and $\tr_{1,\ldots, n}[\sigma_A]\neq 0$ holds if and only if $A_{i}=0$ for all $1\leq i\leq n$.
    Condition (iii) is satisfied by construction (when ordering the observables suitably).
    It remains to compute $\delta (P_1,\ldots, P_{2(4^n - 1)^2})$. To this end, let $\ket{\phi}\in (\mathbb{C}^2)^{\otimes 2n}$ be an arbitrary normalized pure state, $\norm{\phi}=1$. Then
    \small
    \begin{align}
        &\frac{1}{(4^n - 1)^2} \sum_{i=1}^{(4^n - 1)^2} \bra{\phi} P_i\ket{\phi}^2\\
        &= \frac{1}{(4^n - 1)^2} \tr\left[ \left(\sum_{i=1}^{(4^n - 1)^2} P_i\otimes P_i\right)(\ketbra{\phi}{\phi})^{\otimes 2} \right]\\
        &= \frac{1}{(4^n - 1)^2} \tr\Bigg[ 
        \Bigg(\Bigg(\sum_{A\in\{0,1,2,3\}^{2n}} \sigma_A\otimes \sigma_A\Bigg) -\Bigg(\sum_{B\in\{0,1,2,3\}^{n}} (\sigma_B\otimes \mathbbm{1}_2^{\otimes n})\otimes (\sigma_B\otimes \mathbbm{1}_2^{\otimes n})\Bigg) \\
        &\hphantom{= \frac{1}{(4^n - 1)^2} \tr\Bigg[} - \Bigg(\sum_{B\in\{0,1,2,3\}^{n}} (\mathbbm{1}_2^{\otimes n}\otimes \sigma_B)\otimes (\mathbbm{1}_2^{\otimes n}\otimes \sigma_B)\Bigg) + \mathbbm{1}_2^{\otimes 2n} \Bigg)(\ketbra{\phi}{\phi})^{\otimes 2} \Bigg]\\
        &= \frac{1}{(4^n - 1)^2} \big(\tr\left[ 2^{2n} \operatorname{SWAP}_{1,\ldots,2n}(\ketbra{\phi}{\phi})^{\otimes 2}\right] - \tr\left[2^n \operatorname{SWAP}_{1,\ldots,n} (\tr_{n+1,\ldots,2n}[\ketbra{\phi}{\phi}])^{\otimes 2} \right]\\
        &\hphantom{=\frac{1}{(4^n - 1)^2} \big(} - \tr\left[2^n \operatorname{SWAP}_{n+1,\ldots,2n} (\tr_{1,\ldots,n}[\ketbra{\phi}{\phi}])^{\otimes 2} \right] + 1\big)\\
        &= \frac{1}{(4^n - 1)^2} \left(2^{2n} - 2^{n} (\tr[\tr_{n+1,\ldots,2n}[\ketbra{\phi}{\phi}]^2 ] + \tr[\tr_{1,\ldots,n}[\ketbra{\phi}{\phi}]^2 ]) + 1 \right) ,
    \end{align}
    \normalsize
    where we again used \Cref{eq:swap-via-sum-of-paulis} multiple times in the computation.
    
    Accordingly, we get 
    \begin{align}
        \delta (P_1,\ldots, P_{2(4^{n} -1)^2})
        &= \sup_{\ket{\phi}\in (\mathbb{C}^2)^{\otimes 2 n}: \norm{\phi}=1} \frac{1}{(4^n - 1)^2} \sum_{i=1}^{(4^n - 1)^2} \bra{\phi} P_i\ket{\phi}^2\\
        &= \frac{2^{2n} - 2^{n}\cdot \tfrac{2}{2^n} + 1}{(4^{n} -1)^2}\\
        &= \frac{1}{4^n-1} ,
    \end{align}
    and the supremum is attained at $\ket{\phi}$ if and only if both $\tr_{1,\ldots,n}[\ketbra{\phi}{\phi}]$ and $\tr_{n+1,\ldots,2n}[\ketbra{\phi}{\phi}]$ are maximally mixed, which is equivalent to $\ket{\phi}$ being maximally entangled across the cut $\{1,\ldots,2n\}= \{1,\ldots,n\}\cup \{n+1,\ldots,2n\} $.
\end{proof}

We are now in a position to complete the proofs of \Cref{lemma:lower-bound-ptm-learning-without-quantum-memory-from-choi-state-copies,lemma:lower-bound-ptm-learning-without-quantum-memory-from-doubly-stochastic-choi-state-copies}:

\begin{proof}[Proof of \Cref{lemma:lower-bound-ptm-learning-without-quantum-memory-from-choi-state-copies,lemma:lower-bound-ptm-learning-without-quantum-memory-from-doubly-stochastic-choi-state-copies}]
    Recalling from the proof of \Cref{theorem:ptm-entries-from-choi-shadow-tomography} that estimating all entries of $R_\mathcal{N}^\mathcal{P}$ up to accuracy $\varepsilon$ is equivalent to estimating all expectation values $\tr[\sigma_A \tfrac{1}{2^n}\Gamma^{\mathcal{N}}]$, $A\in\{0,1,2,3\}^{2n}$, up to accuracy $\varepsilon$, we can now combine \Cref{lemma:lower-bound-choi-shadow-tomography-general,lemma:delta-choi-pauli-shadow-tomography} to obtain  \Cref{lemma:lower-bound-ptm-learning-without-quantum-memory-from-choi-state-copies}, and combine \Cref{lemma:lower-bound-doubly-stochastic-choi-shadow-tomography-general,lemma:delta-doubly-stochastic-choi-pauli-shadow-tomography} to obtain \Cref{lemma:lower-bound-ptm-learning-without-quantum-memory-from-doubly-stochastic-choi-state-copies}.
\end{proof}

\subsection{Proof of \Cref{theorem:lower-bound-ptm-learning-without-quantum-memory-from-channel-queries,theorem:lower-bound-ptm-learning-without-quantum-memory-from-doubly-stochastic-channel-queries}}

Next, we prove \Cref{theorem:lower-bound-ptm-learning-without-quantum-memory-from-channel-queries,theorem:lower-bound-ptm-learning-without-quantum-memory-from-doubly-stochastic-channel-queries}. To this end, we first establish analogues of \Cref{lemma:lower-bound-choi-shadow-tomography-general,lemma:lower-bound-doubly-stochastic-choi-shadow-tomography-general} for the case of sequential channel access. 
The proof uses the notation introduced in \Cref{definition:learning-tree-representation-channel}.

\begin{lemma}\label{lemma:lower-bound-choi-shadow-tomography-general-channel-access}
    Let $O_1,\ldots, O_M\in\mathcal{B}((\mathbb{C}^2)^{\otimes 2n})$ be $M\in 2\mathbb{N}$ traceless and self-adjoint $(2n)$-qubit observables satisfying the following properties:
    \begin{enumerate}[(i)]
        \item For every $1\leq i\leq M$, $\sigma (O_i)\subseteq\{-1,1\}$. In particular, $\norm{O_i} = 1$ holds for all $1\leq i\leq M$.
        \item For every $1\leq i\leq M$, $\tr_{n+1,\ldots, 2n}[O_i]=0$.
        \item For every $1\leq i\leq \nicefrac{M}{2}$, $O_i = - O_{i + \nicefrac{M}{2}}$.
    \end{enumerate}
    Any algorithm for learning without quantum memory requires
    \begin{equation}
        \Omega \left( \frac{1}{\varepsilon^2 \Delta (O_1,\ldots,O_M)} \right)
    \end{equation}
    queries to an $n$-qubit channel $\mathcal{N}$ to simultaneously predict all expectation values $\tr[O_i \tfrac{1}{2^n}\Gamma]$, $1\leq i\leq M$, up to accuracy $\varepsilon$ with a success probability $\geq\tfrac{2}{3}$.
    Here, we defined
    \small
    \begin{equation}
        \Delta (O_1,\ldots,O_M)
        \coloneqq  \sup_{\substack{n_{\mathrm{aux}}\in\mathbb{N}\\ \ket{\phi}\in (\mathbb{C}^2)^{\otimes n_{\mathrm{aux}}}\otimes (\mathbb{C}^2)^{\otimes n}: \norm{\phi}=1\\ \rho\in \mathcal{S}((\mathbb{C}^2)^{\otimes n_{\mathrm{aux}}}\otimes (\mathbb{C}^2)^{\otimes n})}}\frac{2}{M}\sum_{i-1}^{\nicefrac{M}{2}} \left(\frac{\bra{\phi} \tr_{\mathrm{in}}[(\rho^{\top_{\mathrm{in}}}\otimes\mathbbm{1}_{\mathrm{out}})(\mathbbm{1}_{\mathrm{aux}}\otimes O_i)] \ket{\phi}}{\bra{\phi} (\rho_{\mathrm{aux}} \otimes \mathbbm{1}_{\mathrm{out}}) \ket{\phi}}\right)^2.
    \end{equation}
    \normalsize
\end{lemma}

Analogously to how $\delta (O_1,\ldots,O_M)$ was interpreted in \cite{chen2022exponential} as characterizing the hardness of state shadow tomography without quantum memory for the observables $O_1,\ldots,O_M$, we can view $\Delta (O_1,\ldots,O_M)$ as a quantity characterizing how challenging it is to perform Choi state shadow tomography without quantum memory for these observables, even given query access to the corresponding channel. In particular, if $O_1,\ldots,O_M$ are chosen to be Pauli observables $\sigma_{B_1}\otimes\sigma_{A_1},\ldots,\sigma_{B_M}\otimes\sigma_{A_M}$, then, via \Cref{eq:ptm-as-pauli-choi-expectation-value}, the quantity $\Delta (O_1,\ldots,O_M)$ expresses the hardness of simultaneously predicting the PTM entries $(\mathcal{R}^\mathcal{P}_\mathcal{N})_{A_i, B_i}$ when given general access to $\mathcal{N}$, but without a quantum memory being available.

\begin{proof}
    We first show that any algorithm that can solve the learning task can also solve a certain many-versus-one distinguishing task. Then we show the query complexity lower bound for the latter distinghuishing task.
    As in the proof of \Cref{lemma:lower-bound-choi-shadow-tomography-general}, consider the states 
    \begin{equation}
        \frac{1}{2^n}\Gamma_i
        \coloneqq \frac{\mathbbm{1}_{2}^{\otimes 2n} + 3\varepsilon O_i}{2^{2n}},\quad 1\leq i\leq M,
    \end{equation}
    which are valid Choi states by conditions (i) and (ii). Notice that also the maximally mixed state $\tfrac{1}{2^{2n}}\mathbbm{1}_2^{\otimes 2n}$ is a valid Choi state.
    Moreover, our assumptions on the observables $O_i$ imply both that $\tr[O_i \frac{1}{2^n}\Gamma_i] = 3\varepsilon$ and that $\tr[O_i \tfrac{1}{2^{2n}}\mathbbm{1}_2^{\otimes 2n}]=0$ for all $1\leq i\leq M$.
    Therefore, any algorithm that can predict all the expectation values $\tr[O_i \tfrac{1}{2^n}\Gamma]$, $1\leq i\leq M$, up to accuracy $\varepsilon$ with success probability $\geq\tfrac{2}{3}$ immediately gives rise to an algorithm that, with success probability $\geq\tfrac{2}{3}$, solves the many-versus-one distinguishing task between the unknown channel having maximally mixed Choi state $\tfrac{1}{2^{2n}}\mathbbm{1}_2^{\otimes 2n}$ and the unknown channel being drawn uniformly at random from the $M$ channels $\mathcal{N}_i$ with Choi states $\frac{1}{2^n}\Gamma_i$.
    Hence, any query complexity lower bound for achieving the latter task without a quantum memory directly implies the same query complexity lower bound for the former task under the same restrictions on the algorithm. Thus, we now establish a query complexity lower bound for performing the many-versus-one distinguishing task without quantum memory.
    
    Consider the tree representation $\mathcal{T}$ of the learning algorithm. The probability of arriving at a leaf $\ell$ when having access to $\mathcal{N}$ is
    \begin{equation}
        p^{\mathcal{N}}(\ell)
        = \prod_{t=1}^T w_{v_t} 2^{n} 2^{n_{\mathrm{aux}}} \bra{\psi_{v_t}}(\operatorname{id}_{\mathrm{aux}}\otimes \mathcal{N})(\ketbra{\phi_{v_{t-1}}}{\phi_{v_{t-1}}})\ket{\psi_{v_t}} .
    \end{equation}
    Using that, for every bipartite state $\rho\in\mathcal{S}( (\mathbb{C}^2)^{n_{\mathrm{aux}}} \otimes (\mathbb{C}^2)^{n_{\mathrm{in}}})$, we have
    \begin{equation}
        (\operatorname{id}_{\mathrm{aux}}\otimes \mathcal{N})(\rho)
        = 2^n \tr_{\mathrm{in}}\left[ (\rho^{\top_{\mathrm{in}}}\otimes\mathbbm{1}_{\mathrm{out}})(\mathbbm{1}_{\mathrm{aux}}\otimes \frac{1}{2^n}\Gamma_{\mathrm{in},\mathrm{out}}^{\mathcal{N}})\right],
    \end{equation}
    we can compute
    \begin{equation}
        (\operatorname{id}_{\mathrm{aux}}\otimes \mathcal{N}_i)(\rho)
        = \frac{\rho_{\mathrm{aux}}\otimes \mathbbm{1}_{\mathrm{out}} + 3\varepsilon \tr_{\mathrm{in}}[(\rho^{\top_{\mathrm{in}}}\otimes\mathbbm{1}_{\mathrm{out}})(\mathbbm{1}_{\mathrm{aux}}\otimes O_i)]}{2^n}
    \end{equation}
    as well as
    \begin{equation}
        (\operatorname{id}_{\mathrm{aux}}\otimes \mathcal{N}_{\mathrm{max-mixed}})(\rho)
        = \frac{\rho_{\mathrm{aux}}\otimes \mathbbm{1}_{\mathrm{out}}}{2^n} .
    \end{equation}
    Using these expressions as well as the notational shorthands $\rho_{v_{t-1}} = \ketbra{\phi_{v_{t-1}}}{\phi_{v_{t-1}}}$ and $p_t = \bra{\psi_{v_t}} \rho_{v_{t-1},\mathrm{aux}}\otimes \mathbbm{1}_{\mathrm{out}} \ket{\psi_{v_t}}$, we can now perform the following calculation:
    \small
    \begin{align}
        &\frac{\mathbb{E}_{i\sim\mathrm{Uniform}(\{1,\ldots,M\})} [p^{\mathcal{N}_i}(\ell)]}{p^{\mathcal{N}_{\mathrm{max-mixed}}}(\ell)}\\
        &= \mathbb{E}_{i\sim\mathrm{Uniform}(\{1,\ldots,M\})} \left[ \prod_{t=1}^T \frac{w_{v_t} 2^{n} 2^{n_{\mathrm{aux}}} \left( p_t + 3\varepsilon \bra{\psi_{v_t}} \tr_{\mathrm{in}}[(\rho_{v_{t-1}}^{\top_{\mathrm{in}}}\otimes\mathbbm{1}_{\mathrm{out}})(\mathbbm{1}_{\mathrm{aux}}\otimes O_i)] \ket{\psi_{v_t}} \right)}{w_{v_t} 2^{n} 2^{n_{\mathrm{aux}}} p_t}  \right]\\
        &= \mathbb{E}_{i\sim\mathrm{Uniform}(\{1,\ldots,M\})} \left[ \exp\left(\sum_{t=1}^T \log\left(1 + 3\varepsilon \frac{\bra{\psi_{v_t}} \tr_{\mathrm{in}}[(\rho_{v_{t-1}}^{\top_{\mathrm{in}}}\otimes\mathbbm{1}_{\mathrm{out}})(\mathbbm{1}_{\mathrm{aux}}\otimes O_i)] \ket{\psi_{v_t}}}{p_t}\right) \right)\right]\\
        &\geq \exp\left(\sum_{t=1}^T \mathbb{E}_{i\sim\mathrm{Uniform}(\{1,\ldots,M\})} \left[ \log\left(1 + 3\varepsilon \frac{\bra{\psi_{v_t}} \tr_{\mathrm{in}}[(\rho_{v_{t-1}}^{\top_{\mathrm{in}}}\otimes\mathbbm{1}_{\mathrm{out}})(\mathbbm{1}_{\mathrm{aux}}\otimes O_i)] \ket{\psi_{v_t}}}{p_t}\right)\right] \right)\\
        &= \exp\left(\sum_{t=1}^T \frac{1}{M}\sum_{i-1}^{\nicefrac{M}{2}}\log\left(1 - 9\varepsilon^2 \left(\frac{\bra{\psi_{v_t}} \tr_{\mathrm{in}}[(\rho_{v_{t-1}}^{\top_{\mathrm{in}}}\otimes\mathbbm{1}_{\mathrm{out}})(\mathbbm{1}_{\mathrm{aux}}\otimes O_i)] \ket{\psi_{v_t}}}{p_t}\right)^2\right) \right)\\
        &\geq \exp\left( - \sum_{t=1}^T \frac{18}{M}\sum_{i-1}^{\nicefrac{M}{2}} \varepsilon^2 \left(\frac{\bra{\psi_{v_t}} \tr_{\mathrm{in}}[(\rho_{v_{t-1}}^{\top_{\mathrm{in}}}\otimes\mathbbm{1}_{\mathrm{out}})(\mathbbm{1}_{\mathrm{aux}}\otimes O_i)] \ket{\psi_{v_t}}}{p_t}\right)^2\right)\\
        &\geq \exp\left( - 9 T \varepsilon^2 \Delta (O_1,\ldots,O_M)\right) .
    \end{align}
    \normalsize
    Here, the first two steps are simple rewritings. 
    The third step is by Jensen's inequality.
    The fourth step uses condition (iii).
    The fifth step uses the numerical inequality $\log (1-x)\geq -2x$, which is valid for $x\in [0, 0.79]$, and which we can apply for $\varepsilon >0 $ small enough.
    The final step holds by definition of $\Delta (O_1,\ldots,O_M)$.
    With this, we have shown that
    \begin{equation}
        \frac{\mathbb{E}_{i\sim\mathrm{Uniform}(\{1,\ldots,M\})} [p^{\mathcal{N}_i}(\ell)]}{p^{\mathcal{N}_{\mathrm{max-mixed}}}(\ell)}
        \geq \exp\left( - 9 T \varepsilon^2 \Delta (O_1,\ldots,O_M)\right)
        \geq 1 - 9 T \varepsilon^2 \Delta (O_1,\ldots,O_M) .
    \end{equation}
    Using the one-sided version of Le Cam's two-point method \cite[Lemma 5.4]{chen2021exponential-arxiv}, this tells us that the learning algorithm without quantum memory solves the many-versus-one distinguishing task correctly with probability at most $9 T \varepsilon^2 \Delta (O_1,\ldots,O_M)$.
    Thus, to achieve a success probability $\geq \nicefrac{2}{3}$, the number of queries has to satisfy 
    \begin{equation}
        T\geq\Omega\left( \frac{1}{\varepsilon^2 \Delta (O_1,\ldots,O_M)} \right) 
    \end{equation}
    By the reduction between learning and distinguishing discussed at the beginning of the proof, this establishes the claimed query complexity lower bound.
\end{proof}

By slightly varying the assumptions on the observables, we get the following version for doubly-stochastic and entanglement-breaking channels with sparse PTM:

\begin{lemma}\label{lemma:lower-bound-doubly-stochastic-choi-shadow-tomography-general-channel-access}
    Let $O_1,\ldots, O_M\in\mathcal{B}((\mathbb{C}^2)^{\otimes 2n})$ be $M\in 2\mathbb{N}$ traceless and self-adjoint $(2n)$-qubit observables satisfying the following properties:
    \begin{enumerate}[(i)]
        \item For every $1\leq i\leq M$, $\sigma (O_i)\subseteq\{-1,1\}$. In particular, $\norm{O_i} = 1$ holds for all $1\leq i\leq M$.
        \item For every $1\leq i\leq M$, $\tr_{1,\ldots, n}[O_i]=0$ and $\tr_{n+1,\ldots, 2n}[O_i]=0$.
        \item For every $1\leq i\leq \nicefrac{M}{2}$, $O_i = - O_{i + \nicefrac{M}{2}}$.
        \item For every $1\leq i\leq M$, $O_i$ has a $\mathcal{O}(1)$-sparse Pauli ONB expansion.
        \item For every $1\leq i\leq M$, $O_i$ is a tensor product of $2n$ single-qubit observables.
    \end{enumerate}
    Any algorithm for learning without quantum memory requires
    \begin{equation}
        \Omega \left( \frac{1}{\varepsilon^2 \Delta (O_1,\ldots,O_M)} \right)
    \end{equation}
    queries to a doubly-stochastic and entanglement-breaking $n$-qubit channel $\mathcal{N}$ with $\mathcal{O}(1)$-sparse PTM to simultaneously predict all expectation values $\tr[O_i \tfrac{1}{2^n}\Gamma]$, $1\leq i\leq M$, up to accuracy $\varepsilon$ with a success probability $\geq\tfrac{2}{3}$.
    Here, we again defined
    \small
    \begin{equation}
        \Delta (O_1,\ldots,O_M)
        \coloneqq  \sup_{\substack{n_{\mathrm{aux}}\in\mathbb{N}\\ \ket{\phi}\in (\mathbb{C}^2)^{\otimes n_{\mathrm{aux}}}\otimes (\mathbb{C}^2)^{\otimes n}: \norm{\phi}=1\\ \rho\in \mathcal{S}((\mathbb{C}^2)^{\otimes n_{\mathrm{aux}}}\otimes (\mathbb{C}^2)^{\otimes n})}}\frac{2}{M}\sum_{i-1}^{\nicefrac{M}{2}} \left(\frac{\bra{\phi} \tr_{\mathrm{in}}[(\rho^{\top_{\mathrm{in}}}\otimes\mathbbm{1}_{\mathrm{out}})(\mathbbm{1}_{\mathrm{aux}}\otimes O_i)] \ket{\phi}}{\bra{\phi} (\rho_{\mathrm{aux}}\otimes\mathbbm{1}_{\mathrm{out}})\ket{\phi}}\right)^2.
    \end{equation}
    \normalsize
\end{lemma}
\begin{proof}
    The proof follows the same steps as the proof of \Cref{lemma:lower-bound-choi-shadow-tomography-general-channel-access}, we only have to notice that, with the changed condition (ii) and the added conditions (iv) and (v), all the $\tfrac{1}{2^n}\Gamma_i$ as well as the maximally mixed state are valid Choi states of doubly-stochastic and entanglement-breaking quantum channels with $\mathcal{O}(1)$-sparse PTM.
\end{proof}

As before, to make use of these lower bounds we need to evaluate (or at least bound) the quantity $\Delta (O_1,\ldots,O_M)$ for sets of observables of interest. We will do so for the same sets of observables that we have already considered in \Cref{lemma:delta-choi-pauli-shadow-tomography,lemma:delta-doubly-stochastic-choi-pauli-shadow-tomography}.

\begin{lemma}\label{lemma:Delta-choi-pauli-shadow-tomography-channel-access}
    Consider the following $2\cdot 4^n (4^{n} -1)$ many $(2n)$-qubit Pauli observables
    \begin{align}
        P_1,\ldots,P_{4^n (4^{n} -1)} &\in \{+ \sigma_A\}_{A\in\{0,1,2,3\}^{2n}}\setminus\{+ \sigma_B\otimes \mathbbm{1}_2^{\otimes n}\}_{B\in\{0,1,2,3\}^n}\\
        P_{4^n (4^{n} -1) + 1},\ldots,P_{2\cdot 4^n (4^{n} -1)} &\in \{- \sigma_A\}_{A\in\{0,1,2,3\}^{2n}}\setminus\{- \sigma_B\otimes \mathbbm{1}_2^{\otimes n}\}_{B\in\{0,1,2,3\}^n}.
    \end{align}
    The set of observables $\{P_i\}_{i=1}^{2\cdot 4^n (4^{n} -1)}$, when suitably ordered, satisfies the conditions (i)-(iii) from \cref{lemma:lower-bound-choi-shadow-tomography-general-channel-access}.
    Moreover, 
    \begin{equation}
        \Delta (P_1,\ldots, P_{2\cdot 4^n (4^{n} -1)})
        = \frac{1}{4^n -1} .
    \end{equation}
\end{lemma}
\begin{proof}
    We have already established conditions (i)-(iii) when proving \Cref{lemma:delta-choi-pauli-shadow-tomography}.
    We upper bound $\Delta (P_1,\ldots, P_{2\cdot 4^n (4^{n} -1)})$ as follows: Let $n_{\mathrm{aux}}\in\mathbb{N}$, $\ket{\phi}\in (\mathbb{C}^2)^{\otimes n_{\mathrm{aux}}}\otimes (\mathbb{C}^2)^{\otimes n}$ with $\norm{\phi}=1$, and $\rho\in \mathcal{S}((\mathbb{C}^2)^{\otimes n_{\mathrm{aux}}}\otimes (\mathbb{C}^2)^{\otimes n}))$ be arbitrary. Then, using \Cref{eq:swap-via-sum-of-paulis} similarly to before, we obtain:
    \small
    \begin{align}
        &\frac{1}{4^n (4^{n} -1)} \sum_{i=1}^{4^n (4^{n} -1)} \bra{\phi} \tr_{\mathrm{in}}[(\rho^{\top_{\mathrm{in}}}\otimes\mathbbm{1}_{\mathrm{out}})(\mathbbm{1}_{\mathrm{aux}}\otimes P_i)] \ket{\phi}^2\\
        &= \frac{1}{4^n (4^{n} -1)} \sum_{i=1}^{4^n (4^{n} -1)}  \tr\left[ \tr_{\mathrm{in}}[(\rho^{\top_{\mathrm{in}}}\otimes\mathbbm{1}_{\mathrm{out}})(\mathbbm{1}_{\mathrm{aux}}\otimes P_i)] \ketbra{\phi}{\phi}\right]^2\\
        &= \frac{1}{4^n (4^{n} -1)} \sum_{i=1}^{4^n (4^{n} -1)}  \tr\left[ \left((\rho^{\top_{\mathrm{in}}}\otimes\mathbbm{1}_{\mathrm{out}})(\mathbbm{1}_{\mathrm{aux}}\otimes P_i)\right) \left(\mathbbm{1}_{\mathrm{in}}\otimes \ketbra{\phi}{\phi}\right)\right]^2\\
        &= \frac{1}{4^n (4^{n} -1)} \tr\left[\left(\sum_{i=1}^{4^n (4^{n} -1)} (\mathbbm{1}_{\mathrm{aux}}\otimes P_i)\otimes (\mathbbm{1}_{\mathrm{aux}}\otimes P_i) \right) \left(\left(\mathbbm{1}_{\mathrm{in}}\otimes \ketbra{\phi}{\phi}\right)(\rho^{\top_{\mathrm{in}}}\otimes\mathbbm{1}_{\mathrm{out}})\right)^{\otimes 2}\right]\\
        &= \frac{1}{4^n (4^{n} -1)} \Bigg( \tr\Bigg[ 2^{2n} \operatorname{SWAP_{\mathrm{in},\mathrm{out}}} (\tr_{\mathrm{aux}}[\left(\mathbbm{1}_{\mathrm{in}}\otimes \ketbra{\phi}{\phi}\right)(\rho^{\top_{\mathrm{in}}}\otimes\mathbbm{1}_{\mathrm{out}})])^{\otimes 2}\Bigg]\\
        &\hphantom{= \frac{1}{4^n (4^{n} -1)} \Bigg(} - \tr\Bigg[2^n \operatorname{SWAP}_{\mathrm{in}} (\tr_{\mathrm{aux}, \mathrm{out}}[\left(\mathbbm{1}_{\mathrm{in}}\otimes \ketbra{\phi}{\phi}\right)(\rho^{\top_{\mathrm{in}}}\otimes\mathbbm{1}_{\mathrm{out}})])^{\otimes 2} \Bigg] \Bigg)\\
        &= \frac{1}{4^n (4^{n} -1)} \Bigg( 2^{2n} \tr\Bigg[(\tr_{\mathrm{aux}}[\left(\mathbbm{1}_{\mathrm{in}}\otimes \ketbra{\phi}{\phi}\right)(\rho^{\top_{\mathrm{in}}}\otimes\mathbbm{1}_{\mathrm{out}})])^{2}\Bigg]\\
        &\hphantom{= \frac{1}{4^n (4^{n} -1)} \Bigg(} - 2^n \tr\Bigg[(\tr_{\mathrm{aux}, \mathrm{out}}[\left(\mathbbm{1}_{\mathrm{in}}\otimes \ketbra{\phi}{\phi}\right)(\rho^{\top_{\mathrm{in}}}\otimes\mathbbm{1}_{\mathrm{out}})])^{2} \Bigg] \Bigg) .
    \end{align}
    \normalsize
    Therefore, we have
    \begin{align}
        &\frac{1}{4^n (4^{n} -1)} \sum_{i=1}^{4^n (4^{n} -1)} \left(\frac{\bra{\phi} \tr_{\mathrm{in}}[(\rho^{\top_{\mathrm{in}}}\otimes\mathbbm{1}_{\mathrm{out}})(\mathbbm{1}_{\mathrm{aux}}\otimes P_i)] \ket{\phi}}{\bra{\phi} (\rho_{\mathrm{aux}}\otimes\mathbbm{1}_{\mathrm{out}})\ket{\phi}}\right)^2\\
        &= \frac{1}{4^n (4^{n} -1)} \Bigg( 2^{2n} \frac{\tr\Bigg[(\tr_{\mathrm{aux}}[\left(\mathbbm{1}_{\mathrm{in}}\otimes \ketbra{\phi}{\phi}\right)(\rho^{\top_{\mathrm{in}}}\otimes\mathbbm{1}_{\mathrm{out}})])^{2}\Bigg]}{\bra{\phi} (\rho_{\mathrm{aux}}\otimes\mathbbm{1}_{\mathrm{out}})\ket{\phi}^2}\\
        &\hphantom{= \frac{1}{4^n (4^{n} -1)} \Bigg(} - 2^n \frac{\tr\Bigg[(\tr_{\mathrm{aux}, \mathrm{out}}[\left(\mathbbm{1}_{\mathrm{in}}\otimes \ketbra{\phi}{\phi}\right)(\rho^{\top_{\mathrm{in}}}\otimes\mathbbm{1}_{\mathrm{out}})])^{2} \Bigg]}{\bra{\phi} (\rho_{\mathrm{aux}}\otimes\mathbbm{1}_{\mathrm{out}})\ket{\phi}^2} \Bigg) .
    \end{align}
    Notice that 
    \begin{equation}
        \tr [\left(\mathbbm{1}_{\mathrm{in}}\otimes \ketbra{\phi}{\phi}\right)(\rho^{\top_{\mathrm{in}}}\otimes\mathbbm{1}_{\mathrm{out}}) ] 
        = \bra{\phi} (\rho_{\mathrm{aux}}\otimes \mathbbm{1}_{\mathrm{out}})\ket{\phi}
        = \tr [\left(\mathbbm{1}_{\mathrm{in}}\otimes \ketbra{\phi}{\phi}\right)(\rho\otimes\mathbbm{1}_{\mathrm{out}}) ],
    \end{equation}
    since $\tr_{\mathrm{in}}[\rho^{\top_{\mathrm{in}}}\otimes\mathbbm{1}_{\mathrm{out}}] = \rho_{\mathrm{aux}}\otimes \mathbbm{1}_{\mathrm{out}}$.
    As the partial trace of an operator has the same trace as the original operator, we can thus rewrite the above as follows:
    \begin{align}
        &\frac{1}{4^n (4^{n} -1)} \sum_{i=1}^{4^n (4^{n} -1)} \left(\frac{\bra{\phi} \tr_{\mathrm{in}}[(\rho^{\top_{\mathrm{in}}}\otimes\mathbbm{1}_{\mathrm{out}})(\mathbbm{1}_{\mathrm{aux}}\otimes P_i)] \ket{\phi}}{\bra{\phi} (\rho_{\mathrm{aux}}\otimes\mathbbm{1}_{\mathrm{out}})\ket{\phi}}\right)^2\\
        &= \frac{1}{4^n (4^{n} -1)} \left( 2^{2n} \frac{\tr\Bigg[(\tr_{\mathrm{aux}}[\left(\mathbbm{1}_{\mathrm{in}}\otimes \ketbra{\phi}{\phi}\right)(\rho^{\top_{\mathrm{in}}}\otimes\mathbbm{1}_{\mathrm{out}})])^{2}\Bigg]}{\tr [\left(\mathbbm{1}_{\mathrm{in}}\otimes \ketbra{\phi}{\phi}\right)(\rho\otimes\mathbbm{1}_{\mathrm{out}}) ]^2} - 2^n \right) .
    \end{align}
    It remains to bound the first summand. We first rewrite it using the identity $\tr_A[X Y ]= \tr_A [X^{\top_A} Y^{\top_A}]$ $\forall X,Y\in\mathcal{B}(\mathbb{C}^{d_A}\otimes\mathbb{C}^{d_B})$ together with the identity $( \theta_{\mathrm{in}} \otimes \operatorname{id}_{\mathrm{out}} ) \circ \tr_{\mathrm{aux}} = \tr_{\mathrm{aux}} \circ ( \theta_{\mathrm{in}} \otimes \operatorname{id}_{\mathrm{aux},\mathrm{out}} )$ for the transpose map $\theta_\mathrm{in}$ viewed as a linear superoperator on $\mathcal{B}((\mathbb{C}^2)^{\otimes n})$, which allows us to replace $\rho^{\top_{\mathrm{in}}}$ by $\rho$ in the enumerator. (We explain this in full detail in the proof of \Cref{lemma:Delta-doubly-stochastic-choi-pauli-shadow-tomography-channel-access}.)
    Then, we do a direct computation. Suppose $\ket{\phi}_{\mathrm{aux},\mathrm{out}}$ has Schmidt decomposition $\ket{\phi}_{\mathrm{aux},\mathrm{out}} = \sum_{i=1}^{\min\{2^{n_{\mathrm{aux}}}, 2^{n}\}} \sqrt{\lambda_i} \ket{e_i}_{\mathrm{aux}}\otimes\ket{f_i}_{\mathrm{out}}$, where the $\lambda_i\geq 0$ satisfy $\sum_{i=1}^{\min\{2^{n_{\mathrm{aux}}}, 2^{n}\}}\lambda_i =1$, and where $\{\ket{e_i}\}_{i=1}^{2^{n_{\mathrm{aux}}}}$ and $\{\ket{f_i}\}_{i=1}^{2^{n}}$ are ONBs of $(\mathbb{C}^2)^{\otimes n_{\mathrm{aux}}}$ and $(\mathbb{C}^2)^{\otimes n}$, respectively. 
    As $n_{\mathrm{in}}=n=n_{\mathrm{out}}$, we can expand $\rho_{\mathrm{in},\mathrm{aux}}$ w.r.t.~the tensor product basis $\{\ketbra{f_i}{f_j}\otimes\ketbra{e_k}{e_\ell}\}_{i,j,k,\ell}$ of $\mathcal{B}((\mathbb{C}^2)^{\otimes n}\otimes (\mathbb{C}^2)^{\otimes n_{\mathrm{aux}}})$ as
    \begin{equation}
        \rho_{\mathrm{in},\mathrm{aux}}
        = \sum_{i,j,k,\ell} \rho_{i j k \ell} \ketbra{f_i}{f_j}_{\mathrm{in}}\otimes\ketbra{e_k}{e_\ell}_\mathrm{aux} .
    \end{equation}
    Using these explicit expansions for $\ket{\phi}$ and $\rho$, we can compute
    \begin{equation}
        \left(\mathbbm{1}_{\mathrm{in}}\otimes \ketbra{\phi}{\phi}\right)(\rho\otimes\mathbbm{1}_{\mathrm{out}})])
        = \sum_{i,j,k,\ell,n} \sqrt{\lambda_i \lambda_j} \rho_{k \ell j n} \ketbra{f_k}{f_\ell}_{\mathrm{in}}\otimes\ketbra{e_i}{e_n}_\mathrm{aux}\otimes \ketbra{f_i}{f_j}_{\mathrm{out}} ,
    \end{equation}
    leading to
    \begin{equation}
        \tr_{\mathrm{aux}}[\left(\mathbbm{1}_{\mathrm{in}}\otimes \ketbra{\phi}{\phi}\right)(\rho\otimes\mathbbm{1}_{\mathrm{out}})]
        = \sum_{i,j,k,\ell} \sqrt{\lambda_i \lambda_j} \rho_{k \ell j i} \ketbra{f_k}{f_\ell}_{\mathrm{in}}\otimes \ketbra{f_i}{f_j}_{\mathrm{out}} .
    \end{equation}
    From this, we obtain
    \begin{align}
        \tr\left[ \tr_{\mathrm{aux}}[\left(\mathbbm{1}_{\mathrm{in}}\otimes \ketbra{\phi}{\phi}\right)(\rho^{\top_{\mathrm{in}}}\otimes\mathbbm{1}_{\mathrm{out}})]\right]^2
        &= \tr\left[ \sum_{i,j,k,\ell} \sqrt{\lambda_i \lambda_j} \rho_{k \ell i j} \ketbra{f_k}{f_\ell}_{\mathrm{in}}\otimes \ketbra{f_i}{f_j}_{\mathrm{out}}\right]^2\\
        &= \left(\sum_{i,\ell} \lambda_i \rho_{\ell \ell i i}\right)^2\\
        &= \sum_{i,j,k,\ell} \lambda_i \lambda_j \rho_{\ell \ell i i} \rho_{k k j j} 
    \end{align}
    as well as
    \begin{align}
        &\tr\left[ \left(\tr_{\mathrm{aux}}[\left(\mathbbm{1}_{\mathrm{in}}\otimes \ketbra{\phi}{\phi}\right)(\rho^{\top_{\mathrm{in}}}\otimes\mathbbm{1}_{\mathrm{out}})]\right)^2\right]\\
        &= \tr\left[ \sum_{i,j,k,\ell,m,n,s,t} \sqrt{\lambda_i \lambda_j \lambda_s \lambda_t} \rho_{k \ell j i} \rho_{m n t s} \ketbra{k}{\ell}_{\mathrm{in}}\cdot\ketbra{m}{n}_{\mathrm{in}}\otimes \ketbra{i}{j}_{\mathrm{out}}\cdot\ketbra{s}{t}_{\mathrm{out}}\right] \\
        &= \sum_{i,j,k,\ell} \lambda_i \lambda_j \rho_{\ell k i j} \rho_{k \ell j i} . 
    \end{align}
    As $\rho$ is Hermitian and positive semidefinite, we have, for any $i,j,k,\ell$, $\rho_{\ell k i j} = \overline{\rho_{k \ell j i}}$ and $0\leq \det\begin{pmatrix} \rho_{k k i i} & \rho_{k \ell i j}\\ \overline{\rho_{k \ell j i}} & \rho_{\ell \ell i i} \end{pmatrix} = \rho_{\ell \ell i i} \rho_{k k j j} - \lvert \rho_{\ell k i j}\rvert^2 = \rho_{\ell \ell i i} \rho_{k k j j} - \rho_{\ell k i j} \rho_{k \ell j i} $. The latter holds because the matrix $\begin{pmatrix} \rho_{k k i i} & \rho_{k \ell i j}\\ \overline{\rho_{k \ell j i}} & \rho_{\ell \ell i i} \end{pmatrix}$ is a principal submatrix of the positive semidefinite matrix $\rho$ and thus itself positive semidefinite. 
    Altogether, we now have:
    \begin{align}
        \tr\left[ \left(\tr_{\mathrm{aux}}[\left(\mathbbm{1}_{\mathrm{in}}\otimes \ketbra{\phi}{\phi}\right)(\rho^{\top_{\mathrm{in}}}\otimes\mathbbm{1}_{\mathrm{out}})]\right)^2\right]
        &= \sum_{i,j,k,\ell} \underbrace{\lambda_i \lambda_j}_{\geq 0} \underbrace{\rho_{\ell k i j} \rho_{k \ell j i}}_{\leq \rho_{\ell \ell i i} \rho_{k k j j}}\\
        &\leq \sum_{i,j,k,\ell} \lambda_i \lambda_j \rho_{\ell \ell i i} \rho_{k k j j} \\
        &= \tr\left[ \tr_{\mathrm{aux}}[\left(\mathbbm{1}_{\mathrm{in}}\otimes \ketbra{\phi}{\phi}\right)(\rho^{\top_{\mathrm{in}}}\otimes\mathbbm{1}_{\mathrm{out}})]\right]^2 .
    \end{align}
    Plugging this inequality back into our expression of interest, we have shown
    \begin{align}
        \frac{1}{4^n (4^{n} -1)} \sum_{i=1}^{4^n (4^{n} -1)} \left(\frac{\bra{\phi} \tr_{\mathrm{in}}[(\rho^{\top_{\mathrm{in}}}\otimes\mathbbm{1}_{\mathrm{out}})(\mathbbm{1}_{\mathrm{aux}}\otimes P_i)] \ket{\phi}}{\bra{\phi} (\rho_{\mathrm{aux}}\otimes\mathbbm{1}_{\mathrm{out}})\ket{\phi}}\right)^2
        &\leq  \frac{2^{2n} - 2^n}{4^n (4^{n} -1)} \\
        &= \frac{2^n -1}{2^n (4^n -1)}\\
        &\leq \frac{1}{4^n -1} .
    \end{align}
    To finish the proof, it only remains to notice that equality can be attained, for example, if $\ket{\phi}_{\mathrm{aux},\mathrm{out}}$ and $\rho_{\mathrm{in},\mathrm{aux}}$ factorize as $\ket{\phi}_{\mathrm{aux},\mathrm{out}} = \ket{\psi}_{\mathrm{aux}}\otimes\ket{\varphi}_{\mathrm{out}}$ and $\rho_{\mathrm{in},\mathrm{aux}} = \ketbra{\varphi}{\varphi}_{\mathrm{in}}\otimes \ketbra{\psi}{\psi}_{\mathrm{aux}}$.
\end{proof}

\begin{lemma}\label{lemma:Delta-doubly-stochastic-choi-pauli-shadow-tomography-channel-access}
    Consider the following $2(4^{2n} - 2\cdot 4^n +1) = 2(4^n - 1)^2 $ many $(2n)$-qubit Pauli observables
    \begin{align}
        P_1,\ldots,P_{(4^n - 1)^2} &\in \{+ \sigma_A\}_{A\in\{0,1,2,3\}^{2n}:\exists 1\leq i\leq n, n+1\leq j\leq 2n: A_i\neq 0 \neq A_j}\\
        P_{4^{2n} - 2\cdot 4^n +2},\ldots,P_{2(4^n - 1)^2} &\in \{- \sigma_A\}_{A\in\{0,1,2,3\}^{2n}:\exists 1\leq i\leq n, n+1\leq j\leq 2n: A_i\neq 0 \neq A_j}.
    \end{align}
    The set of observables $\{P_i\}_{i=1}^{2 (4^n - 1)^2}$, when suitably ordered, satisfies the conditions (i)-(v) from \cref{lemma:lower-bound-doubly-stochastic-choi-shadow-tomography-general-channel-access}.
    Moreover, 
    \begin{equation}
        \frac{1}{(2^n + 1)^2}
        \leq \Delta (P_1,\ldots, P_{2((4^n - 1)^2)})
        \leq \frac{1}{(2^n - 1)^2}.
    \end{equation}
\end{lemma}
\begin{proof}
    We have already established conditions (i)-(v) when proving \Cref{lemma:delta-doubly-stochastic-choi-pauli-shadow-tomography}.
    We bound $\Delta (P_1,\ldots, P_{2 (4^n - 1)^2})$ as follows: Let $n_{\mathrm{aux}}\in\mathbb{N}$, $\ket{\phi}\in (\mathbb{C}^2)^{\otimes n_{\mathrm{aux}}}\otimes (\mathbb{C}^2)^{\otimes n}$ with $\norm{\phi}=1$, and $\rho\in \mathcal{S}((\mathbb{C}^2)^{\otimes n_{\mathrm{aux}}}\otimes (\mathbb{C}^2)^{\otimes n}))$ be arbitrary. Then, using \Cref{eq:swap-via-sum-of-paulis} similarly to before, we obtain:
    \begingroup
    \allowdisplaybreaks
    \begin{align}
        &\frac{1}{(4^n - 1)^2} \sum_{i=1}^{(4^n - 1)^2} \bra{\phi} \tr_{\mathrm{in}}[(\rho^{\top_{\mathrm{in}}}\otimes\mathbbm{1}_{\mathrm{out}})(\mathbbm{1}_{\mathrm{aux}}\otimes P_i)] \ket{\phi}^2\\
        &= \frac{1}{(4^n - 1)^2} \sum_{i=1}^{(4^n - 1)^2}  \tr\left[ \tr_{\mathrm{in}}[(\rho^{\top_{\mathrm{in}}}\otimes\mathbbm{1}_{\mathrm{out}})(\mathbbm{1}_{\mathrm{aux}}\otimes P_i)] \ketbra{\phi}{\phi}\right]^2\\
        &= \frac{1}{(4^n - 1)^2} \sum_{i=1}^{(4^n - 1)^2}  \tr\left[ \left((\rho^{\top_{\mathrm{in}}}\otimes\mathbbm{1}_{\mathrm{out}})(\mathbbm{1}_{\mathrm{aux}}\otimes P_i)\right) \left(\mathbbm{1}_{\mathrm{in}}\otimes \ketbra{\phi}{\phi}\right)\right]^2\\
        &= \frac{1}{(4^n - 1)^2} \tr\left[\left(\sum_{i=1}^{(4^n - 1)^2} (\mathbbm{1}_{\mathrm{aux}}\otimes P_i)\otimes (\mathbbm{1}_{\mathrm{aux}}\otimes P_i) \right) \left(\left(\mathbbm{1}_{\mathrm{in}}\otimes \ketbra{\phi}{\phi}\right)(\rho^{\top_{\mathrm{in}}}\otimes\mathbbm{1}_{\mathrm{out}})\right)^{\otimes 2}\right]\\
        &= \frac{1}{(4^n - 1)^2} \Bigg( \tr\Bigg[ 2^{2n} \operatorname{SWAP_{\mathrm{in},\mathrm{out}}} (\tr_{\mathrm{aux}}[\left(\mathbbm{1}_{\mathrm{in}}\otimes \ketbra{\phi}{\phi}\right)(\rho^{\top_{\mathrm{in}}}\otimes\mathbbm{1}_{\mathrm{out}})])^{\otimes 2}\Bigg]\\
        &\hphantom{= \frac{1}{(4^n - 1)^2} \Bigg(} - \tr\Bigg[2^n \operatorname{SWAP}_{\mathrm{in}} (\tr_{\mathrm{aux}, \mathrm{out}}[\left(\mathbbm{1}_{\mathrm{in}}\otimes \ketbra{\phi}{\phi}\right)(\rho^{\top_{\mathrm{in}}}\otimes\mathbbm{1}_{\mathrm{out}})])^{\otimes 2} \Bigg]\\
        &\hphantom{= \frac{1}{(4^n - 1)^2} \Bigg(} - \tr\Bigg[2^n \operatorname{SWAP}_{\mathrm{out}} (\tr_{\mathrm{aux}, \mathrm{in}}[\left(\mathbbm{1}_{\mathrm{in}}\otimes \ketbra{\phi}{\phi}\right)(\rho^{\top_{\mathrm{in}}}\otimes\mathbbm{1}_{\mathrm{out}})])^{\otimes 2} \Bigg]\\
        &\hphantom{= \frac{1}{(4^n - 1)^2} \Bigg(} + \tr\Bigg[(\mathbbm{1}_{\mathrm{aux}}\otimes \mathbbm{1}_{\mathrm{in},\mathrm{out}}) (\tr_{\mathrm{aux}, \mathrm{in}}[\left(\mathbbm{1}_{\mathrm{in}}\otimes \ketbra{\phi}{\phi}\right)(\rho^{\top_{\mathrm{in}}}\otimes\mathbbm{1}_{\mathrm{out}})])^{\otimes 2} \Bigg]\Bigg)\\
        &= \frac{1}{(4^n - 1)^2} \Bigg( 2^{2n} \tr\Bigg[(\tr_{\mathrm{aux}}[\left(\mathbbm{1}_{\mathrm{in}}\otimes \ketbra{\phi}{\phi}\right)(\rho^{\top_{\mathrm{in}}}\otimes\mathbbm{1}_{\mathrm{out}})])^{2}\Bigg]\\
        &\hphantom{= \frac{1}{(4^n - 1)^2} \Bigg(} - 2^n \tr\Bigg[(\tr_{\mathrm{aux}, \mathrm{out}}[\left(\mathbbm{1}_{\mathrm{in}}\otimes \ketbra{\phi}{\phi}\right)(\rho^{\top_{\mathrm{in}}}\otimes\mathbbm{1}_{\mathrm{out}})])^{2} \Bigg]\\
        &\hphantom{= \frac{1}{(4^n - 1)^2} \Bigg(} - 2^n \tr\Bigg[(\tr_{\mathrm{aux}, \mathrm{in}}[\left(\mathbbm{1}_{\mathrm{in}}\otimes \ketbra{\phi}{\phi}\right)(\rho^{\top_{\mathrm{in}}}\otimes\mathbbm{1}_{\mathrm{out}})])^{2} \Bigg]\\
        &\hphantom{= \frac{1}{(4^n - 1)^2} \Bigg(} + \tr [\left(\mathbbm{1}_{\mathrm{in}}\otimes \ketbra{\phi}{\phi}\right)(\rho^{\top_{\mathrm{in}}}\otimes\mathbbm{1}_{\mathrm{out}}) ]^2\Bigg) .
    \end{align}
    \endgroup
    Therefore, we have
    \begin{align}
        &\frac{1}{(4^n - 1)^2} \sum_{i=1}^{(4^n - 1)^2} \left(\frac{\bra{\phi} \tr_{\mathrm{in}}[(\rho^{\top_{\mathrm{in}}}\otimes\mathbbm{1}_{\mathrm{out}})(\mathbbm{1}_{\mathrm{aux}}\otimes P_i)] \ket{\phi}}{\bra{\phi} (\rho_{\mathrm{aux}}\otimes \mathbbm{1}_{\mathrm{out}})\ket{\phi}}\right)^2\\
        &= \frac{1}{(4^n - 1)^2} \Bigg( 2^{2n} \frac{\tr\Bigg[(\tr_{\mathrm{aux}}[\left(\mathbbm{1}_{\mathrm{in}}\otimes \ketbra{\phi}{\phi}\right)(\rho^{\top_{\mathrm{in}}}\otimes\mathbbm{1}_{\mathrm{out}})])^{2}\Bigg]}{\bra{\phi} (\rho_{\mathrm{aux}}\otimes \mathbbm{1}_{\mathrm{out}})\ket{\phi}^2}\\
        &\hphantom{= \frac{1}{(4^n - 1)^2} \Bigg(} - 2^n \frac{\tr\Bigg[(\tr_{\mathrm{aux}, \mathrm{out}}[\left(\mathbbm{1}_{\mathrm{in}}\otimes \ketbra{\phi}{\phi}\right)(\rho^{\top_{\mathrm{in}}}\otimes\mathbbm{1}_{\mathrm{out}})])^{2} \Bigg]}{\bra{\phi} (\rho_{\mathrm{aux}}\otimes \mathbbm{1}_{\mathrm{out}})\ket{\phi}^2}\\
        &\hphantom{= \frac{1}{(4^n - 1)^2} \Bigg(} - 2^n \frac{\tr\Bigg[(\tr_{\mathrm{aux}, \mathrm{in}}[\left(\mathbbm{1}_{\mathrm{in}}\otimes \ketbra{\phi}{\phi}\right)(\rho^{\top_{\mathrm{in}}}\otimes\mathbbm{1}_{\mathrm{out}})])^{2} \Bigg]}{\bra{\phi} (\rho_{\mathrm{aux}}\otimes \mathbbm{1}_{\mathrm{out}})\ket{\phi}^2}\\
        &\hphantom{= \frac{1}{(4^n - 1)^2} \Bigg(} + \frac{\tr [\left(\mathbbm{1}_{\mathrm{in}}\otimes \ketbra{\phi}{\phi}\right)(\rho^{\top_{\mathrm{in}}}\otimes\mathbbm{1}_{\mathrm{out}}) ]^2}{\bra{\phi} (\rho_{\mathrm{aux}}\otimes \mathbbm{1}_{\mathrm{out}})\ket{\phi}^2}\Bigg) 
    \end{align}
    Notice that 
    \begin{equation}
        \tr [\left(\mathbbm{1}_{\mathrm{in}}\otimes \ketbra{\phi}{\phi}\right)(\rho^{\top_{\mathrm{in}}}\otimes\mathbbm{1}_{\mathrm{out}}) ] 
        = \bra{\phi} (\rho_{\mathrm{aux}}\otimes \mathbbm{1}_{\mathrm{out}})\ket{\phi} 
        = \tr [\left(\mathbbm{1}_{\mathrm{in}}\otimes \ketbra{\phi}{\phi}\right)(\rho\otimes\mathbbm{1}_{\mathrm{out}}) ],
    \end{equation}
    since $\tr_{\mathrm{in}}[\rho^{\top_{\mathrm{in}}}\otimes\mathbbm{1}_{\mathrm{out}}] = \rho_{\mathrm{aux}}\otimes \mathbbm{1}_{\mathrm{out}} = \tr_{\mathrm{in}}[\rho\otimes\mathbbm{1}_{\mathrm{out}}]$.
    As the partial trace of an operator has the same trace as the original operator, we can thus rewrite the above as follows:
    \begin{align}
        &\frac{1}{(4^n - 1)^2} \sum_{i=1}^{(4^n - 1)^2} \left(\frac{\bra{\phi} \tr_{\mathrm{in}}[(\rho^{\top_{\mathrm{in}}}\otimes\mathbbm{1}_{\mathrm{out}})(\mathbbm{1}_{\mathrm{aux}}\otimes P_i)] \ket{\phi}}{\bra{\phi} (\rho_{\mathrm{aux}}\otimes \mathbbm{1}_{\mathrm{out}})\ket{\phi}}\right)^2\\
        &= \frac{1}{(4^n - 1)^2} \Bigg( 2^{2n} \frac{\tr\Bigg[(\tr_{\mathrm{aux}}[\left(\mathbbm{1}_{\mathrm{in}}\otimes \ketbra{\phi}{\phi}\right)(\rho^{\top_{\mathrm{in}}}\otimes\mathbbm{1}_{\mathrm{out}})])^{2}\Bigg]}{ \tr\Bigg[\tr_{\mathrm{aux}}[\left(\mathbbm{1}_{\mathrm{in}}\otimes \ketbra{\phi}{\phi}\right)(\rho\otimes\mathbbm{1}_{\mathrm{out}})]\Bigg]^{2}}\\
        &\hphantom{= \frac{1}{(4^n - 1)^2} \Bigg(} - 2^n \frac{\tr\Bigg[(\tr_{\mathrm{aux}, \mathrm{out}}[\left(\mathbbm{1}_{\mathrm{in}}\otimes \ketbra{\phi}{\phi}\right)(\rho^{\top_{\mathrm{in}}}\otimes\mathbbm{1}_{\mathrm{out}})])^{2} \Bigg]}{\tr\Bigg[\tr_{\mathrm{aux}, \mathrm{out}}[\left(\mathbbm{1}_{\mathrm{in}}\otimes \ketbra{\phi}{\phi}\right)(\rho\otimes\mathbbm{1}_{\mathrm{out}})] \Bigg]^{2}}\\
        &\hphantom{= \frac{1}{(4^n - 1)^2} \Bigg(} - 2^n \frac{\tr\Bigg[(\tr_{\mathrm{aux}, \mathrm{in}}[\left(\mathbbm{1}_{\mathrm{in}}\otimes \ketbra{\phi}{\phi}\right)(\rho^{\top_{\mathrm{in}}}\otimes\mathbbm{1}_{\mathrm{out}})])^{2} \Bigg]}{\tr\Bigg[\tr_{\mathrm{aux}, \mathrm{in}}[\left(\mathbbm{1}_{\mathrm{in}}\otimes \ketbra{\phi}{\phi}\right)(\rho\otimes\mathbbm{1}_{\mathrm{out}})] \Bigg]^{2}}
        + 1 \Bigg) .
    \end{align}
    Next we use the identity
    \begin{equation}\label{eq:trace-of-product-equals-trace-of-product-of-partial-traces}
        \tr[X Y]
        = \tr[ X^{\top_A} Y^{\top_A}],\quad \forall X,Y\in\mathcal{B}(\mathbb{C}^{d_A}\otimes\mathbb{C}^{d_B}) ,
    \end{equation}
    which can, e.g., be checked by expanding in a tensor product basis.
    Let $\theta$ denote the transpose viewed as linear superoperator. We clearly have $( \theta_{\mathrm{in}} \otimes \operatorname{id}_{\mathrm{out}} ) \circ \tr_{\mathrm{aux}} = \tr_{\mathrm{aux}} \circ ( \theta_{\mathrm{in}} \otimes \operatorname{id}_{\mathrm{aux},\mathrm{out}} )$.
    Combining this with \Cref{eq:trace-of-product-equals-trace-of-product-of-partial-traces}, we can rewrite the enumerator in the first summand above as
    \begin{align}
        &\tr\Bigg[(\tr_{\mathrm{aux}}[\left(\mathbbm{1}_{\mathrm{in}}\otimes \ketbra{\phi}{\phi}\right)(\rho^{\top_{\mathrm{in}}}\otimes\mathbbm{1}_{\mathrm{out}})])^{2}\Bigg]\\
        &= \tr\Bigg[ \left( (\theta_{\mathrm{in}} \otimes \operatorname{id}_{\mathrm{out}})\left( \tr_{\mathrm{aux}}[\left(\mathbbm{1}_{\mathrm{in}}\otimes \ketbra{\phi}{\phi}\right)(\rho^{\top_{\mathrm{in}}}\otimes\mathbbm{1}_{\mathrm{out}})]\right)\right)^{2}\Bigg]\\
        &= \tr\Bigg[ \left(  \tr_{\mathrm{aux}}\left[(\theta_{\mathrm{in}} \otimes \operatorname{id}_{\mathrm{aux},\mathrm{out}})\left(\left(\mathbbm{1}_{\mathrm{in}}\otimes \ketbra{\phi}{\phi}\right)(\rho^{\top_{\mathrm{in}}}\otimes\mathbbm{1}_{\mathrm{out}})\right)\right]\right)^{2}\Bigg]\\
        &= \tr\Bigg[ \left(  \tr_{\mathrm{aux}}\left[\left(\mathbbm{1}_{\mathrm{in}}\otimes \ketbra{\phi}{\phi}\right)(\rho\otimes\mathbbm{1}_{\mathrm{out}})\right]\right)^{2}\Bigg] .
    \end{align}
    Similarly, we can combine the identity $\tr[XY]=\tr[X^\top Y^\top]$ $\forall X,Y\in\mathcal{B}(\mathcal{H})$ with the commutation  $\theta_{\mathrm{in}} \circ \tr_{\mathrm{aux},\mathrm{out}} = \tr_{\mathrm{aux},\mathrm{out}} \circ ( \theta_{\mathrm{in}} \otimes \operatorname{id}_{\mathrm{aux},\mathrm{out}} )$ to rewrite the enumerator in the second summand above as
    \begin{align}
        &\tr\Bigg[(\tr_{\mathrm{aux}, \mathrm{out}}[\left(\mathbbm{1}_{\mathrm{in}}\otimes \ketbra{\phi}{\phi}\right)(\rho^{\top_{\mathrm{in}}}\otimes\mathbbm{1}_{\mathrm{out}})])^{2} \Bigg]\\
        &= \tr\Bigg[\left(\theta_{\mathrm{in}}(\tr_{\mathrm{aux}, \mathrm{out}}[\left(\mathbbm{1}_{\mathrm{in}}\otimes \ketbra{\phi}{\phi}\right)(\rho^{\top_{\mathrm{in}}}\otimes\mathbbm{1}_{\mathrm{out}})])\right)^{2} \Bigg]\\
        &= \tr\Bigg[\left(\tr_{\mathrm{aux}, \mathrm{out}}\left[( \theta_{\mathrm{in}} \otimes \operatorname{id}_{\mathrm{aux},\mathrm{out}} )\left(\left(\mathbbm{1}_{\mathrm{in}}\otimes \ketbra{\phi}{\phi}\right)(\rho^{\top_{\mathrm{in}}}\otimes\mathbbm{1}_{\mathrm{out}})\right)\right]\right)^{2} \Bigg]\\
        &= \tr\Bigg[(\tr_{\mathrm{aux}, \mathrm{out}}[\left(\mathbbm{1}_{\mathrm{in}}\otimes \ketbra{\phi}{\phi}\right)(\rho\otimes\mathbbm{1}_{\mathrm{out}})])^{2} \Bigg] .
    \end{align}
    Finally, using the identity $\tr_A[X Y ]= \tr_A [X^{\top_A} Y^{\top_A}]$ $\forall X,Y\in\mathcal{B}(\mathbb{C}^{d_A}\otimes\mathbb{C}^{d_B})$, we can rewrite the enumerator in the third summand above as
    \begin{align}
        &\tr\Bigg[(\tr_{\mathrm{aux}, \mathrm{in}}[\left(\mathbbm{1}_{\mathrm{in}}\otimes \ketbra{\phi}{\phi}\right)(\rho^{\top_{\mathrm{in}}}\otimes\mathbbm{1}_{\mathrm{out}})])^{2} \Bigg]\\
        &= \tr\Bigg[\left(\tr_{\mathrm{aux}, \mathrm{in}}\left[( \theta_{\mathrm{in}} \otimes \operatorname{id}_{\mathrm{aux},\mathrm{out}} )\left(\left(\mathbbm{1}_{\mathrm{in}}\otimes \ketbra{\phi}{\phi}\right)(\rho^{\top_{\mathrm{in}}}\otimes\mathbbm{1}_{\mathrm{out}})\right)\right]\right)^{2} \Bigg]\\
        &= \tr\Bigg[(\tr_{\mathrm{aux}, \mathrm{in}}[\left(\mathbbm{1}_{\mathrm{in}}\otimes \ketbra{\phi}{\phi}\right)(\rho\otimes\mathbbm{1}_{\mathrm{out}})])^{2} \Bigg] .
    \end{align}
    In summary, we now have the following expression:
    \begin{align}
        &\frac{1}{(4^n - 1)^2} \sum_{i=1}^{(4^n - 1)^2} \left(\frac{\bra{\phi} \tr_{\mathrm{in}}[(\rho^{\top_{\mathrm{in}}}\otimes\mathbbm{1}_{\mathrm{out}})(\mathbbm{1}_{\mathrm{aux}}\otimes P_i)] \ket{\phi}}{\bra{\phi} (\rho_{\mathrm{aux}}\otimes \mathbbm{1}_{\mathrm{out}})\ket{\phi}}\right)^2\\
        &= \frac{1}{(4^n - 1)^2} \Bigg( 2^{2n} \frac{\tr\Bigg[(\tr_{\mathrm{aux}}[\left(\mathbbm{1}_{\mathrm{in}}\otimes \ketbra{\phi}{\phi}\right)(\rho\otimes\mathbbm{1}_{\mathrm{out}})])^{2}\Bigg]}{ \tr\Bigg[\tr_{\mathrm{aux}}[\left(\mathbbm{1}_{\mathrm{in}}\otimes \ketbra{\phi}{\phi}\right)(\rho\otimes\mathbbm{1}_{\mathrm{out}})]\Bigg]^{2}}\\
        &\hphantom{= \frac{1}{(4^n - 1)^2} \Bigg(} - 2^n \frac{\tr\Bigg[(\tr_{\mathrm{aux}, \mathrm{out}}[\left(\mathbbm{1}_{\mathrm{in}}\otimes \ketbra{\phi}{\phi}\right)(\rho\otimes\mathbbm{1}_{\mathrm{out}})])^{2} \Bigg]}{\tr\Bigg[\tr_{\mathrm{aux}, \mathrm{out}}[\left(\mathbbm{1}_{\mathrm{in}}\otimes \ketbra{\phi}{\phi}\right)(\rho\otimes\mathbbm{1}_{\mathrm{out}})] \Bigg]^{2}}\\
        &\hphantom{= \frac{1}{(4^n - 1)^2} \Bigg(} - 2^n \frac{\tr\Bigg[(\tr_{\mathrm{aux}, \mathrm{in}}[\left(\mathbbm{1}_{\mathrm{in}}\otimes \ketbra{\phi}{\phi}\right)(\rho\otimes\mathbbm{1}_{\mathrm{out}})])^{2} \Bigg]}{\tr\Bigg[\tr_{\mathrm{aux}, \mathrm{in}}[\left(\mathbbm{1}_{\mathrm{in}}\otimes \ketbra{\phi}{\phi}\right)(\rho\otimes\mathbbm{1}_{\mathrm{out}})] \Bigg]^{2}}
        + 1 \Bigg) .
    \end{align}
    To simplify the second and third summand further, we use that the partial trace is cyclic w.r.t.~operators that act non-trivially only on the system(s) being traced over. This allows us to do the following rewritings: As $\ketbra{\phi}{\phi} = (\ketbra{\phi}{\phi})^2$, we get
    \begin{align}
        \tr_{\mathrm{aux}, \mathrm{out}}[\left(\mathbbm{1}_{\mathrm{in}}\otimes \ketbra{\phi}{\phi}\right)(\rho\otimes\mathbbm{1}_{\mathrm{out}})]
        &= \tr_{\mathrm{aux}, \mathrm{out}}[\left(\mathbbm{1}_{\mathrm{in}}\otimes \ketbra{\phi}{\phi}\right)^2(\rho\otimes\mathbbm{1}_{\mathrm{out}})]\\
        &= \tr_{\mathrm{aux}, \mathrm{out}}[\left(\mathbbm{1}_{\mathrm{in}}\otimes \ketbra{\phi}{\phi}\right)(\rho\otimes\mathbbm{1}_{\mathrm{out}})\left(\mathbbm{1}_{\mathrm{in}}\otimes \ketbra{\phi}{\phi}\right)] .
    \end{align}
    \normalsize
    And as $\rho = \sqrt{\rho}^2$, we get
    \small
    \begin{align}
        \tr_{\mathrm{aux}, \mathrm{in}}[\left(\mathbbm{1}_{\mathrm{in}}\otimes \ketbra{\phi}{\phi}\right)(\rho\otimes\mathbbm{1}_{\mathrm{out}})]
        &= \tr_{\mathrm{aux}, \mathrm{in}}[\left(\mathbbm{1}_{\mathrm{in}}\otimes \ketbra{\phi}{\phi}\right)(\sqrt{\rho}\otimes\mathbbm{1}_{\mathrm{out}})^2]\\
        &= \tr_{\mathrm{aux}, \mathrm{in}}[(\sqrt{\rho}\otimes\mathbbm{1}_{\mathrm{out}})\left(\mathbbm{1}_{\mathrm{in}}\otimes \ketbra{\phi}{\phi}\right)(\sqrt{\rho}\otimes\mathbbm{1}_{\mathrm{out}})] .
    \end{align}
    \normalsize
    Therefore, both $\tr_{\mathrm{aux}, \mathrm{out}}[\left(\mathbbm{1}_{\mathrm{in}}\otimes \ketbra{\phi}{\phi}\right)(\rho\otimes\mathbbm{1}_{\mathrm{out}})]$ and $\tr_{\mathrm{aux}, \mathrm{in}}[\left(\mathbbm{1}_{\mathrm{in}}\otimes \ketbra{\phi}{\phi}\right)(\rho\otimes\mathbbm{1}_{\mathrm{out}})]$ are partial traces of positive semidefinite operators and thus themselves positive semidefinite.
    As the inequality $\lvert \tr[X^2]\rvert = \tr[X^2]\leq (\tr[X])^2$ holds for any positive semidefinite $X$, we now arrive at the following upper bound:
    \begin{align}
        &\frac{1}{(4^n - 1)^2} \sum_{i=1}^{(4^n - 1)^2} \left(\frac{\bra{\phi} \tr_{\mathrm{in}}[(\rho^{\top_{\mathrm{in}}}\otimes\mathbbm{1}_{\mathrm{out}})(\mathbbm{1}_{\mathrm{aux}}\otimes P_i)] \ket{\phi}}{\bra{\phi} (\rho_{\mathrm{aux}}\otimes \mathbbm{1}_{\mathrm{out}})\ket{\phi}}\right)^2\\
        &\leq \frac{1}{(4^n - 1)^2} \left( 2^{2n} \frac{\tr\Bigg[(\tr_{\mathrm{aux}}[\left(\mathbbm{1}_{\mathrm{in}}\otimes \ketbra{\phi}{\phi}\right)(\rho\otimes\mathbbm{1}_{\mathrm{out}})])^{2}\Bigg]}{ \tr\Bigg[\tr_{\mathrm{aux}}[\left(\mathbbm{1}_{\mathrm{in}}\otimes \ketbra{\phi}{\phi}\right)(\rho\otimes\mathbbm{1}_{\mathrm{out}})]\Bigg]^{2}} + 2\cdot 2^n + 1 \right) .
    \end{align}
    It remains to recall from the proof of \Cref{lemma:Delta-choi-pauli-shadow-tomography-channel-access} that 
    \begin{equation}
        \tr\Bigg[(\tr_{\mathrm{aux}}[\left(\mathbbm{1}_{\mathrm{in}}\otimes \ketbra{\phi}{\phi}\right)(\rho\otimes\mathbbm{1}_{\mathrm{out}})])^{2}\Bigg]
        \leq \tr\Bigg[\tr_{\mathrm{aux}}[\left(\mathbbm{1}_{\mathrm{in}}\otimes \ketbra{\phi}{\phi}\right)(\rho\otimes\mathbbm{1}_{\mathrm{out}})]\Bigg]^{2}.
    \end{equation}
    Therefore, we have shown the upper bound:
    \begin{align}
        \frac{1}{(4^n - 1)^2} \sum_{i=1}^{(4^n - 1)^2} \left(\frac{\bra{\phi} \tr_{\mathrm{in}}[(\rho^{\top_{\mathrm{in}}}\otimes\mathbbm{1}_{\mathrm{out}})(\mathbbm{1}_{\mathrm{aux}}\otimes P_i)] \ket{\phi}}{\bra{\phi} (\rho_{\mathrm{aux}}\otimes \mathbbm{1}_{\mathrm{out}})\ket{\phi}}\right)^2
        &\leq \frac{2^{2n} + 2\cdot 2^n + 1}{(4^n - 1)^2}\\
        &= \frac{(2^n +1)^2}{(4^n - 1)^2}\\
        &= \frac{1}{(2^n - 1)^2} .
    \end{align}
    This finishes the proof for the upper bound. To establish the lower bound we can for example consider the case in which $\ket{\phi}_{\mathrm{aux},\mathrm{out}}$ and $\rho_{\mathrm{in},\mathrm{aux}}$ factorize as $\ket{\phi}_{\mathrm{aux},\mathrm{out}} = \ket{\psi}_{\mathrm{aux}}\otimes\ket{\varphi}_{\mathrm{out}}$ and $\rho_{\mathrm{in},\mathrm{aux}} = \ketbra{\varphi}{\varphi}_{\mathrm{in}}\otimes \ketbra{\psi}{\psi}_{\mathrm{aux}}$. Plugging these choices into our exact expression above, we get 
    \begin{align}
        \frac{1}{(4^n - 1)^2} \sum_{i=1}^{(4^n - 1)^2} \left(\frac{\bra{\phi} \tr_{\mathrm{in}}[(\rho^{\top_{\mathrm{in}}}\otimes\mathbbm{1}_{\mathrm{out}})(\mathbbm{1}_{\mathrm{aux}}\otimes P_i)] \ket{\phi}}{\bra{\phi} (\rho_{\mathrm{aux}}\otimes \mathbbm{1}_{\mathrm{out}})\ket{\phi}}\right)^2
        &= \frac{2^{2n} - 2\cdot 2^n + 1}{(4^n -1)^2}\\
        &= \frac{(2^n - 1)^2}{(4^n - 1)^2}\\
        &= \frac{1}{(2^n + 1)^2} .
    \end{align}
    This concludes the proof.
\end{proof}

Having established these lemmata, we can complete the proofs of \Cref{theorem:lower-bound-ptm-learning-without-quantum-memory-from-channel-queries,theorem:lower-bound-ptm-learning-without-quantum-memory-from-doubly-stochastic-channel-queries}:

\begin{proof}[Proof of \Cref{theorem:lower-bound-ptm-learning-without-quantum-memory-from-channel-queries,theorem:lower-bound-ptm-learning-without-quantum-memory-from-doubly-stochastic-channel-queries}]
    Recalling from the proof of \Cref{theorem:ptm-entries-from-choi-shadow-tomography} that estimating all entries of $R_\mathcal{N}^\mathcal{P}$ up to accuracy $\varepsilon$ is equivalent to estimating all expectation values $\tr[\sigma_A \tfrac{1}{2^n}\Gamma^{\mathcal{N}}]$, $A\in\{0,1,2,3\}^{2n}$, up to accuracy $\varepsilon$, we can now combine \Cref{lemma:lower-bound-choi-shadow-tomography-general-channel-access,lemma:Delta-choi-pauli-shadow-tomography-channel-access} to obtain  \Cref{theorem:lower-bound-ptm-learning-without-quantum-memory-from-channel-queries}, and combine \Cref{lemma:lower-bound-doubly-stochastic-choi-shadow-tomography-general,lemma:Delta-doubly-stochastic-choi-pauli-shadow-tomography-channel-access} to obtain \Cref{theorem:lower-bound-ptm-learning-without-quantum-memory-from-doubly-stochastic-channel-queries}.
\end{proof}

\begin{remark}
    A natural alternative proof strategy for establishing \Cref{theorem:lower-bound-ptm-learning-without-quantum-memory-from-channel-queries,theorem:lower-bound-ptm-learning-without-quantum-memory-from-doubly-stochastic-channel-queries} would be to reduce them to their respective Choi access analogues. 
    It is, however, not immediately obvious how to obtain such a reduction information-theoretically efficiently because protocols like implementation of a channel via Choi state teleportation have a success probability that is exponentially small in the number of qubits. While this success probability can be boosted if local unitary invariances of the Choi state are known in advance (compare, e.g., \cite[Section 2.1]{wolf2012quantumchannels}), we would need the unknown quantum channels to share the same invariances for this insight to be useful in a reduction for our lower bounds. We were unable to identify a suitable shared invariance for the Choi states used in our proofs and therefore resorted to a proof strategy other than a reduction to learning from Choi access. 
    It would be interesting to see whether the proof of \cite[Theorem 2, (iv)]{chen2022quantum}, which is based on teleportation stretching for Pauli channels, can be modified to obtain a similar reduction from general channel access to Choi access in our setting, as this may simplify our proofs.
\end{remark}

\subsection{Proof of \Cref{theorem:hamiltonian-learning}}

Here, we give the complete proof for our Hamiltonian learning guarantee in \Cref{theorem:hamiltonian-learning}:

\begin{proof}[Proof of \Cref{theorem:hamiltonian-learning}]
    We follow the strategy sketched in the main text.
    We choose the parameters for the polynomial interpolation as follows.
    We set the maximal evolution time to be $T = \nicefrac{1}{\norm{H}}$ and the Chebyshev degree to be $L = \left\lceil 2\log\left(\tfrac{8\norm{H}}{\sqrt{2\pi}\ln(2)\varepsilon}\right)\right\rceil$.
    Also, we set our interpolation times to be $t_\ell = \tfrac{T}{2}(1+z_\ell)$, where $z_\ell = -\cos(\frac{2\ell-1}{2L} \pi)$, for $1\leq \ell\leq L$.
    Moreover, we set $\tilde{\varepsilon} = \tfrac{3 T \varepsilon}{4 (L-1)L (2L-1)}$.
    
    As first part of the procedure, for each $1\leq \ell\leq L$, we apply \Cref{corollary:pauli-sparse-expectations-from-choi-shadow-tomography} to obtain a classical description $\hat{\mathcal{U}}_{t_\ell}$ of $\mathcal{U}_{t_\ell}$ that allows to predict Pauli-sparse expectation values, with success probability $\geq 1-\delta$. 
    This uses $\mathcal{O}\left(\tfrac{n + \log(1/\delta)}{\tilde{\varepsilon}^4}\right)$ parallel queries to $\mathcal{U}_{t_\ell}$ for each $1\leq \ell\leq L$.
    In the second part of the procedure, to estimate $M$ Pauli coefficients $\alpha (A^{(i)})$, $A^{(i)}\in\{0,1,2,3\}^n\setminus\{0^n\}$, $1\leq i\leq M$, we first use the classical representations $\hat{\mathcal{U}}_{\ell}$ to obtain $\tilde{\varepsilon}$-accurate estimates $2\hat{\alpha}_{t_\ell}^{(i)}$ of the expectation values $\tr[\sigma_{B^{(i)}} \mathcal{U}_{t_\ell} (\rho^{(i)})]$, with $B^{(i)}$ and $\rho^{(i)}$ as in \Cref{lemma:isolating-single-pauli-coefficient}. 
    This uses $m_2 = \mathcal{O}\left(\tfrac{ \log (M /\delta)}{\tilde{\varepsilon}^2}\right)$ additional parallel queries to $\mathcal{U}_{t_\ell}$ as well as classical computation time $\mathcal{O}\left(\tfrac{n + \log (1/\delta)}{\tilde{\varepsilon}^4 }\right)$, for each $1\leq \ell\leq L$.
    Next, we perform Chebyshev interpolation with these estimated values, which gives Chebyshev coefficients $\hat{b}_{0}^{(i)} = \tfrac{1}{L}\sum_{\ell=1}^L \hat{\alpha}_{t_\ell}^{(i)}$ and $\hat{b}_{m}^{(i)} = \tfrac{2}{L}\sum_{\ell=1}^L \hat{\alpha}_{t_\ell}^{(i)} T_m (z_\ell)$, $1\leq m\leq L-1$, where $T_m$ is the $m$th Chebyshev polynomial (of the first kind).
    These Chebyshev coefficients in turn give rise to the first-order Taylor coefficients $\hat{c}_1^{(i)} = -\tfrac{2}{T} \sum_{m=0}^{L-1} (-1)^m \hat{b}_{m}^{(i)} m^2$, $1\leq i\leq M$, as shown for example in \cite[Lemma 6]{gu2022practical}.
    
    It remains to show that indeed $\lvert \hat{c}_1^{(i)} - \alpha (A^{(i)})\rvert\leq \varepsilon$ holds simultaneously for all $1\leq i\leq M$, with success probability $\geq 1-\delta$.
    To see this, first recall that by \Cref{corollary:pauli-sparse-expectations-from-choi-shadow-tomography}, we know that, with success probability $\geq 1-\delta$,
    \begin{equation}
        \left\lvert 2\hat{\alpha}_{t_\ell}^{(i)} - \tr[\sigma_{B^{(i)}} \mathcal{U}_{t_\ell} (\rho^{(i)})]\right\rvert\leq\tilde{\varepsilon}\quad \forall 1\leq i\leq M.
    \end{equation}
    For the remainder of the proof, we condition on this success event.
    Consider an ``ideal'' datasets $(t_1, \tfrac{1}{2}\tr[\sigma_{B^{(i)}} \mathcal{U}_{t_1} (\rho^{(i)})]),\ldots,(t_L, \tfrac{1}{2}\tr[\sigma_{B^{(i)}} \mathcal{U}_{t_L} (\rho^{(i)})])$ -- to which the algorithm does not have access, but which is a useful tool for our analysis. Let $b_{m}^{(i)}$ be the coefficients of the Chebyshev interpolation polynomial for that ideal dataset.
    Let $c_1^{(i)} = -\tfrac{2}{T} \sum_{m=0}^{L-1} (-1)^m b_{m}^{(i)} m^2$ be the corresponding estimates for the first order Taylor coefficients. 
    By the triangle inequality, $\lvert \hat{c}_1^{(i)} - \alpha (A^{(i)})\rvert\leq \lvert \hat{c}_1^{(i)} - c_1^{(i)} \rvert + \lvert c_1^{(i)} - \alpha (A^{(i)}) \rvert$.
    The first of these two summands is easy to bound: As we are in the success event, we have $\lvert \hat{b}_{m}^{(i)} - b_{m}^{(i)}\rvert \leq 2\tilde{\varepsilon}$ for every $0\leq m\leq L-1$ and for every $1\leq i\leq M$.
    Therefore,
    \begin{equation}
        \lvert \hat{c}_1^{(i)} - c_1^{(i)} \rvert
        \leq \frac{2}{T} \sum_{m=0}^{L-1} \lvert \hat{b}_m - b_m\rvert m^2
        \leq \frac{4\tilde{\varepsilon}}{T} \sum_{m=0}^{L-1} m^2
        = \frac{2\tilde{\varepsilon} (L-1)L(2L-1)}{3T}
        = \frac{\varepsilon}{2} ,
    \end{equation}
    where the last inequality is by our choice of $\tilde{\varepsilon}$.
    For the second summand, we know -- from \Cref{lemma:isolating-single-pauli-coefficient} and from the definition of $c_1^{(i)}$ -- that $\alpha (A^{(i)})$ and $c_1^{(i)}$ are exactly the first order Taylor coefficients of the smooth $\mathbb{R}$-valued function $t\mapsto \tfrac{1}{2}\tr[\sigma_{B^{(i)}} \mathcal{U}_{t} (\rho^{(i)})]\equiv f^{(i)}(t)$ and of the corresponding degree-$(L-1)$ Chebyshev interpolation, respectively.
    We can therefore use \cite[Theorem 3]{howell1991derivative}, which tells us
    \begin{equation}
        \lvert c_1^{(i)} - \alpha (A^{(i)}) \rvert
        \leq \norm{\omega^{(1)}}_\infty\cdot \frac{\norm{\dv[L]{t}f^{(i)}}_\infty}{1!\cdot (L-1)!}\, .
    \end{equation}
    where $\omega(t)\coloneqq (t-t_0)\cdot (t-t_1)\cdot\ldots\cdot (t-t_N)$.
    To bound the right-hand side, we combine \Cref{lemma:derivative-bounds} with the observation that
    \begin{align}
        \norm{\omega^{(1)}}_\infty
        &= \max_{\tau\in [0,T]} \eval{\dv{t} \left((t-t_1)\cdot (t-t_2)\cdot\ldots\cdot (t-t_L)\right)}_{t=\tau}\\
        &= \max_{\tau\in [0,T]} \eval{\sum_{i=1}^L \prod_{k\neq i} (t - t_k)}_{t=\tau}\\
        &\leq \max_{\zeta\in [-1,1]} \sum_{i=1}^L \eval{\prod_{k\neq i} \frac{T}{2}(z - z_k)}_{z=\zeta}\\
        &\leq \sum_{i=1}^L T^{L-1}\\
        &\leq L T^{L-1}\, .
    \end{align}
    Combining these bounds we obtain
    \begin{equation}
        \lvert c_1^{(i)} - \alpha (A^{(i)}) \rvert
        \leq \frac{LT^{L-1}}{(L-1)!}\cdot (2\norm{H})^L
        = 2L \cdot \norm{H}\cdot\frac{(2T\norm{H})^{L-1}}{(L-1)!}\, .
    \end{equation}
    We can further upper bound this using the Stirling approximation bounds $\sqrt{2\pi(L-1)}\left(\tfrac{L-1}{e}\right)^{L-1}e^{\tfrac{1}{12(L-1)+1}}\leq (L-1)!\leq \sqrt{2\pi(L-1)}\left(\tfrac{L-1}{e}\right)^{L-1}e^{\tfrac{1}{12(L-1)}}$ from \cite{robbins1955remark-stirling}, which give us
    \begin{align}
        2L \cdot \norm{H}\cdot\frac{(2T\norm{H})^{L-1}}{(L-1)!}
        &\leq 2L \cdot \norm{H}\cdot\frac{(2eT\norm{H})^{L-1}}{\sqrt{2\pi(L-1)} (L-1)^{L-1}} e^{-\tfrac{1}{12(L-1)+1}}\\
        &\leq \frac{2L}{\sqrt{2\pi(L-1)}}\cdot \norm{H}\cdot\left(\frac{1}{2}\right)^{L-1}\\
        &\leq \frac{4\norm{H}}{\sqrt{2\pi}\ln(2)} \cdot \ln(2)L e^{-\ln(2) L}\\
        &\leq \frac{4\norm{H}}{\sqrt{2\pi}\ln(2)} \cdot \left(\frac{1}{2}\right)^{\tfrac{L}{2}}\\
        &\leq \frac{\varepsilon}{2}\, ,
    \end{align}
    where the last inequality is by our choice of $L$. In that computation, we also used the fact that $x e^{-x} \leq e^{-\tfrac{x}{2}}$ holds for any $x\geq 0$.
    Altogether, we have shown that, conditioned on our success event, $\lvert \hat{c}_1^{(i)} - \alpha (A^{(i)})\rvert \leq \varepsilon$ holds for all $1\leq i\leq M$.
    
    To finish the proof, let us summarize the different complexity bounds.
    To build the classical representations in the first phase, we use $\mathcal{O}\left(\tfrac{n + \log(1/\delta)}{\tilde{\varepsilon}^4}\right)$ queries to $\mathcal{U}_{t_\ell}$ for each $1\leq \ell\leq L$. 
    Overall, the first phase thus uses $\mathcal{O}\left(L\cdot \tfrac{n + \log(1/\delta)}{\tilde{\varepsilon}^4}\right) = \tilde{\mathcal{O}}\left(\tfrac{n + \log(1/\delta)}{\varepsilon^4}\cdot \norm{H}^4\right)$ 
    parallel queries to time evolutions along $H$, each with a time of order $\mathcal{O}(T)=\mathcal{O}(\nicefrac{1}{\norm{H}})$, leading to a total evolution time for the first phase of
    \begin{align}
        T_{1,\mathrm{total}}
        &= \mathcal{O}\left(\frac{n + \log(1/\delta)}{\tilde{\varepsilon}^4}\right)\cdot \sum_{\ell=1}^L t_\ell \\
        &= \tilde{\mathcal{O}}\left(\frac{n + \log(1/\delta)}{\varepsilon^4}\cdot \norm{H}^4 
        \right)\cdot \sum_{\ell=1}^L \frac{T}{2}(1+z_\ell)\\
        &= \tilde{\mathcal{O}}\left(\frac{n + \log(1/\delta)}{\varepsilon^4}\cdot \norm{H}^4 
        \right)\cdot  \frac{1}{2 \norm{H}}\sum_{\ell=1}^L \left(1-\cos(\frac{2\ell-1}{L} \pi)\right)\\
        &\leq \tilde{\mathcal{O}}\left(\frac{n + \log(1/\delta)}{\varepsilon^4}\cdot \norm{H}^4 
        \right)\cdot  \frac{L}{\norm{H}}\\
        &\leq \tilde{\mathcal{O}}\left(\frac{n + \log(1/\delta)}{\varepsilon^4}\cdot \norm{H}^3 
        \right) .
    \end{align}
    
    The second phase produces the $\tilde{\varepsilon}$-accurate estimates $\hat{\alpha}_{t_\ell}^{(i)}$, for $1\leq \ell\leq L$, $1\leq i \leq M$.
    This uses $\mathcal{O}\left(\tfrac{ \log (M /\delta)}{\tilde{\varepsilon}^2}\right)$ additional parallel queries to $\mathcal{U}_{t_\ell}$ for each $1\leq \ell\leq L$. Accordingly, this uses $\mathcal{O}\left(\tfrac{ L\cdot \log (M /\delta)}{\tilde{\varepsilon}^2}\right) = \tilde{\mathcal{O}}\left(\tfrac{\log (M) + \log(1/\delta)}{\varepsilon^2}\cdot  \norm{H}^2
    \right)$ additional parallel queries to $\mathcal{O}(\nicefrac{1}{\norm{H}})$-time evolutions along $H$, and the total evolution time for the second phase is $T_{2,\mathrm{total}} \leq \tilde{\mathcal{O}}\left(\tfrac{n + \log(1/\delta)}{\varepsilon^2}\cdot \norm{H} 
    \right)$.
    Producing the estimates $\hat{\alpha}_{t_\ell}^{(i)}$ uses classical computation time $\mathcal{O}\left(LM\cdot \tfrac{n^2 + n\log (1/\delta)}{\tilde{\varepsilon}^4 }\right)= \tilde{\mathcal{O}}\left(\tfrac{n^2 + n\log(1/\delta)}{\varepsilon^4}\cdot\norm{H}^4 
    M\right)$.
    This dominates the classical computation time since the Chebyshev interpolation and consequent derivative estimation takes time $\mathcal{O}(L)$.
    
    Finally, the classical representation $\hat{H}$ of $H$ consists of the classical representations $\hat{\mathcal{U}}_{t_\ell}$ of $\mathcal{U}_{t_\ell}$  for $1\leq \ell\leq L$. As each $\hat{\mathcal{U}}_{t_\ell}$ consists of $\mathcal{O}\left(\tfrac{n^2 + n \log(1/\delta)}{\tilde{\varepsilon}^4}\right)$ real numbers, we see that $\hat{H}$ consists of $\mathcal{O}\left(L\cdot \tfrac{n^2 + n \log(1/\delta)}{\tilde{\varepsilon}^4}\right) = \tilde{\mathcal{O}}\left(\tfrac{n^2 + n \log(1/\delta)}{\varepsilon^4}\cdot\norm{H}^4
    \right)$ real numbers stored in classical memory.
\end{proof}

\section{Query Complexity Lower Bound for Pauli Transfer Matrix Learning With Quantum Memory from Choi access}\label{appendix:query-lower-bound-ptm-learning-with-quantum-memory-choi-access}

Complementing the discussion of learning (Pauli) transfer matrices and predicting sparse expectation values from copies of the Choi state from \Cref{section:ptm-learning,section:tm-learning}, we give an information-theoretic argument to show that the linear-in-$n$ sample complexity scaling of \Cref{theorem:ptm-entries-from-choi-shadow-tomography} and \Cref{corollary:pauli-sparse-expectations-from-choi-shadow-tomography} is optimal for learning with quantum memory from Choi access. 

\begin{lemma}\label{lemma:ptm-lower-bound-with-quantum-memory}
    Any algorithm for learning with quantum memory from Choi access requires $\Omega (\tfrac{n}{\varepsilon^2})$ copies of the Choi state of an unknown $n$-qubit channel $\mathcal{N}$ to estimate all entries of $R_\mathcal{N}^\mathcal{P}$ up to accuracy $\varepsilon$ with success probability $\geq \tfrac{2}{3}$.
\end{lemma}
\begin{proof}
    Consider the set of observables $\{P_i\}_{i=1}^{2\cdot 4^n (4^{n} -1)}$ defined in \Cref{lemma:delta-choi-pauli-shadow-tomography} and the corresponding Choi states $\frac{1}{2^n}\Gamma_i = \frac{\mathbbm{1}_{2}^{\otimes 2n} + 3\varepsilon P_i}{2^{2n}}$.
    By construction, we have
    \begin{equation}
        \tr[P_i \frac{1}{2^n}\Gamma_j] = 3\varepsilon (\delta_{i, j} - \delta_{i+4^n (4^{n} -1), j}) \quad\forall 1\leq i,j\leq 2\cdot 4^n (4^{n} -1),
    \end{equation}
    where the addition in the index is to be understood modulo $2\cdot 4^n (4^{n} -1)$.
    Therefore, the ability to predict $\tr[P_i \frac{1}{2^n}\Gamma_j]$ up to accuracy $\varepsilon$ for every $j$ -- which is no harder than estimating all PTM entries up to accuracy $\varepsilon$ by \Cref{eq:ptm-as-pauli-choi-expectation-value} -- from $T$ copies of the Choi state implies the ability to identify $i$ from $T$ copies of the Choi state.
    If we consider $i$ to be chosen uniformly at random, then this implies that the mutual information between the classical and the quantum subststem in the classical-quantum state
    \begin{equation}
        \rho_{\mathrm{c}\mathrm{q}}
        = \mathbb{E}_{I\sim\mathcal{U}(\{1,\ldots, 2\cdot 4^n (4^{n} -1))\}}\left[\ketbra{I}{I}\otimes \left(\frac{1}{2^n}\Gamma_I\right)^{\otimes T} \right] 
    \end{equation}
    is at least $\log_2 \left(2\cdot 4^n (4^{n} -1)\right) = \Omega (n)$. That is, we have
    \small
    \begin{align}
        \Omega (n)
        &\leq I(\mathrm{c};\mathrm{q})_{\rho} \\
        &= S \left( \mathbb{E}_{I\sim\mathcal{U}(\{1,\ldots, 2\cdot 4^n (4^{n} -1))\}}\left[\left(\frac{1}{2^n}\Gamma_I\right)^{\otimes T} \right] \right) - \mathbb{E}_{I\sim\mathcal{U}(\{1,\ldots, 2\cdot 4^n (4^{n} -1))\}}\left[S\left(\left(\frac{1}{2^n}\Gamma_I\right)^{\otimes T}\right) \right]  \\
        &\leq T\left( 2 n - \mathbb{E}_{I\sim\mathcal{U}(\{1,\ldots, 2\cdot 4^n (4^{n} -1))\}}\left[S\left(\frac{1}{2^n}\Gamma_I\right) \right]\right).
    \end{align}
    \normalsize
    Here, the second step used that the quantum mutual information of a classical-quantum state equals its Holevo information (compare \cite[Exercise 11.6.9]{wilde2011classical}) and the third step uses that any $(2nT)$-qubit state has a von Neumann entropy of at most $2nT$, and that the von Neumann entropy of a $T$-fold tensor power equals $T$ times the von Neumann entropy of a single tensor factor.
    
    Next, we compute $S\left(\frac{1}{2^n}\Gamma_i\right)$ for arbitrary $i$. To this end, simply observe that half of the eigenvalues of $P_i$ are $\pm 1$, so half of the eigenvalues of $\frac{1}{2^n}\Gamma_i$ are $\tfrac{1 \pm 3\varepsilon}{2^{2n}}$. Therefore, we get
    \begin{align}
        S\left(\frac{1}{2^n}\Gamma_i\right)
        &= \frac{2^{2n}}{2}\cdot \frac{1 + 3\varepsilon}{2^{2n}}\log\left(\frac{2^{2n}}{1 + 3\varepsilon}\right) + \frac{2^{2n}}{2}\cdot \frac{1 - 3\varepsilon}{2^{2n}}\log\left(\frac{2^{2n}}{1 - 3\varepsilon}\right)\\
        &= 2n -  \left( 1 - \frac{1 + 3\varepsilon}{2}\log\left(\frac{2}{1 + 3\varepsilon}\right) - \frac{1 - 3\varepsilon}{2}\log\left(\frac{2}{1 - 3\varepsilon}\right)\right)\\
        &= 2n - \mathcal{O}(\varepsilon^2),
    \end{align}
    where the last step is via Taylor expansion. Therefore, we have shown
    \begin{equation}
        \Omega (n)
        \leq I(\mathrm{c};\mathrm{q})_{\rho} 
        \leq  \mathcal{O}(T \varepsilon^2),
    \end{equation}
    which we can rearrange to the claimed $T=\Omega (\tfrac{n}{\varepsilon^2})$.
\end{proof}

The proof strategy used here is not new, see for example \cite[Section 6]{aaronson2018shadow} for a similar reasoning.
In fact, we can easily apply the same reasoning with the set of observables $\{P_i\}_{i=1}^{2\cdot 4^n (4^{n} -1)}$ from \Cref{lemma:delta-doubly-stochastic-choi-pauli-shadow-tomography} to show that the $\Omega (\tfrac{n}{\varepsilon^2})$ lower bound of \Cref{lemma:ptm-lower-bound-with-quantum-memory} even holds if the unkown channel is promised to be doubly-stochastic and to have a $\mathcal{O}(1)$-sparse PTM. 
Following the same reasoning that we have already used to go from \Cref{theorem:tm-entries-from-choi-shadow-tomography} to \Cref{corollary:sample-complexity-predicting-sparse-expectation-values}, this lower bound immediately carries over to the task of predicting expectation values for states and observables with $\mathcal{O}(1)$-sparse Pauli ONB expansion for an unknown doubly-stochastic quantum channel with $\mathcal{O}(1)$-sparse PTM.

\end{document}